\documentclass{article}

\usepackage[utf8]{inputenc}

\usepackage[english]{babel}

\usepackage{fullpage}
\usepackage{soul} 
\usepackage{amsmath}
\usepackage{amssymb}
\usepackage{amsthm}
\usepackage{bbm}
\usepackage{mathrsfs}
\usepackage{amsfonts}
\usepackage{dsfont}
\usepackage{xcolor}

\usepackage{authblk}

\usepackage{mathtools}
\mathtoolsset{showonlyrefs}

\usepackage[round]{natbib}

\usepackage{float}

\usepackage{subcaption}

\usepackage{enumitem}

\numberwithin{equation}{section}

\newcommand{\T}{\top}
\newcommand{\diag}{{\text{\texttt{diag}}}}
\newcommand{\del}{\partial}
\newcommand{\RR}{\mathbb{R}}
\newcommand{\R}{\ensuremath{\mathbb{R}}}
\newcommand{\PP}{\mathbb{P}}
\renewcommand{\P}{\ensuremath{\mathbb{P}}}
\newcommand{\N}{\ensuremath{\mathbb{N}}}
\newcommand{\eps}{\varepsilon}
\newcommand{\FF}{\mathcal{F}}
\newcommand{\1}{\mathbbm{1}}

\newcommand{\V}{\ensuremath{\mathbb{V}}}

\newcommand{\sampleSize}{\ensuremath{n}}
\newcommand{\aLocus}{\ensuremath{a}}
\newcommand{\bLocus}{\ensuremath{b}}

\newcommand{\invPopSize}{\ensuremath{\lambda}}
\newcommand{\cumInvPopSize}{\ensuremath{\Lambda}}

\newcommand{\bigTime}{\ensuremath{T}}
\renewcommand{\time}{\ensuremath{t}}
\newcommand{\tMRCA}{T_\text{MRCA}}
\newcommand{\ancestralProcess}{\ensuremath{A}}
\newcommand{\ancStates}{\ensuremath{\mathcal{S}}}
\newcommand{\ancestralRateMatrix}{\ensuremath{Q}}
\newcommand{\distAncestral}{\ensuremath{p}}
\newcommand{\distVecAncestral}{\ensuremath{\mathbf{p}}}
\newcommand{\eff}{\ensuremath{F}}
\newcommand{\vecEff}{\ensuremath{\mathbf{F}}}

\newcommand{\stateFunction}{\ensuremath{v}}
\newcommand{\stateFunctionMatrix}{\ensuremath{V}}

\newcommand{\treeLength}{\ensuremath{L}}
\newcommand{\totalTreeLength}{\ensuremath{\mathcal{L}}}
\newcommand{\treeDimA}{\ensuremath{x}}
\newcommand{\treeDimB}{\ensuremath{y}}

\newcommand{\recoRate}{\ensuremath{\rho}}
\newcommand{\ancestralRecoProcess}{\ensuremath{A^{\rho}}}
\newcommand{\numLineagesRV}{\ensuremath{K}}
\newcommand{\numLineages}{\ensuremath{k}}
\newcommand{\recoStates}{\ensuremath{\mathcal{S}^{\rho}}}
\newcommand{\state}{\ensuremath{s}}
\newcommand{\ancestralRecoMatrix}{\ensuremath{\tilde{Q}}}
\newcommand{\ancestralRecoCoalMatrix}{\ensuremath{Q^{c}}}
\newcommand{\ancestralRecoOnlyMatrix}{\ensuremath{Q^{\rho}}}

\newcommand{\ancestralRecoLimitedProcess}{\ensuremath{\bar{A}^{\rho}}}
\newcommand{\numRecoRV}{\ensuremath{\bar{R}}}
\newcommand{\numReco}{\ensuremath{r}}
\newcommand{\recoLimitedStates}{\ensuremath{\bar{\mathcal{S}}^{\rho}}}
\newcommand{\ancestralRecoLimitedMatrix}{\ensuremath{\bar{Q}}}
\newcommand{\ancestralRecoCoalLimitedMatrix}{\ensuremath{\bar{Q}^{c}}}
\newcommand{\ancestralRecoOnlyLimitedMatrix}{\ensuremath{\bar{Q}^{\rho}}}
\newcommand{\numLineagesRVLimited}{\ensuremath{\bar{K}}}
\newcommand{\absStates}{\ensuremath{\Delta}}

\newcommand{\vecVee}{\ensuremath{{\mathbf{v}}}}
\newcommand{\vecEks}{\ensuremath{{\mathbf{x}}}}
\newcommand{\argChar}{\ensuremath{\tau}}
\newcommand{\rateODE}{\ensuremath{q}}
\newcommand{\inhODE}{\ensuremath{g}}
\newcommand{\rateInt}{\ensuremath{H}}
\newcommand{\numTimePoints}{\ensuremath{M}}
\newcommand{\timeGrid}{\ensuremath{T}}
\newcommand{\treeAGrid}{\ensuremath{X}}
\newcommand{\treeBGrid}{\ensuremath{Y}}
\newcommand{\numTreeAPoints}{\ensuremath{U}}
\newcommand{\intersectTimeA}{\ensuremath{\timeGrid_\treeDimA}}
\newcommand{\intersectTimeB}{\ensuremath{\timeGrid_\treeDimB}}

\newcommand{\firstRevision}[1]{{\color{black}#1}}
\newcommand{\edit}[1]{{\color{black}#1}}

\newtheorem{proposition}{Proposition}[section]

\newtheorem{corollary}[proposition]{Corollary}
\newtheorem{definition}[proposition]{Definition}
\newtheorem{lemma}[proposition]{Lemma}
\newtheorem{remark}[proposition]{Remark}

\def\mc{\mathcal}
\newcommand{\SP}{\hspace{1pt}}

\title{\firstRevision{Computing} the joint distribution of the total tree length across loci in populations with variable size \\ \vspace{10pt} \normalsize (\rm In: {\it Theoretical Population Biology} (2017), Vol. 118, p.1-19.) }

\author[1,3]{Alexey Miroshnikov}
\author[2,3]{Matthias Steinrücken}
\affil[1]{University of California, Los Angeles, Department of Mathematics}
\affil[2]{\edit{University of Chicago, Department of Ecology and Evolution}}
\affil[3]{University of Massachusetts Amherst, Department of Biostatistics and Epidemiology}

\date{}

\begin{document}

\maketitle 

\begin{abstract}


In recent years, a number of methods have been developed to infer complex demographic histories, especially historical population size changes, from genomic sequence data. Coalescent Hidden Markov Models have proven to be particularly useful for this type of inference.
Due to the Markovian structure of these models, an essential building block is the joint distribution of local genealogical trees, or statistics of these genealogies, at two neighboring loci in populations of variable size. Here, we present a novel method to compute the marginal and the joint distribution of the total length of the genealogical trees at two loci separated by \edit{at most one recombination event} for samples of arbitrary size. To our knowledge, no method to compute these distributions has been presented in the literature to date. We show that they can be obtained from the solution of certain hyperbolic systems of partial differential equations. We present a numerical algorithm, based on the method of characteristics, that can be used to efficiently and accurately solve these systems and compute the marginal and the joint distributions. We demonstrate its utility to study the properties of the joint distribution. Our flexible method can be straightforwardly extended to \edit{handle an arbitrary fixed number of recombination events, to} include the distributions of other statistics of the genealogies as well, and can also be applied in structured populations.
\end{abstract}

\smallskip


\medskip







\section{Introduction}

\firstRevision{Unraveling the complex demographic histories of humans or other species and understanding their effects on contemporary genetic variation is a central goal of population genetics.}
\firstRevision{In addition to} advancing our \firstRevision{knowledge} of the \firstRevision{evolutionary} processes that shape genomic variation, \firstRevision{demographic inference} is also an important step \firstRevision{towards understanding} disease related genetic variation. Recent rapid population growth, for \firstRevision{example}, severely affects the distribution of rare genetic variants~\citep{Keinan2012}, which have been linked to complex genetic diseases. \firstRevision{Moreover, ancient and contemporary population structure can lead to the accumulation of private genetic variation in certain sub-populations.}

\firstRevision{Methods to study genetic variation, or perform inference, in populations with varying size or more complex demographic histories have been developed based on the Wright-Fisher diffusion, describing the evolution of population allele frequencies forward in time~\citep{Griffiths2003,Zivkovic2015,Gutenkunst2009,Excoffier2013}, or the Coalescent process, a model for the genealogical relationship in a sample of individuals~\citep{Griffiths1994b,Griffiths1996,Griffiths1998,Zivkovic2008,Bhaskar2015,Kamm2017}.}
\firstRevision{A powerful representation of genetic variation data that has been used in this context is the Site-Frequency-Spectrum. In this representation, however, any linkage information present in the genetic data is ignored.}
\firstRevision{With the increasing availability of full-genomic sequence data, linkage information is more readily available, and}
approaches \firstRevision{based on Coalescent Hidden Markov Models (HMM) that use this linkage information} have proven to be particularly successful for demographic inference \firstRevision{and other} population genetic applications.

In a population-sample of genomic sequences, the genealogical relationships vary along the genome, due to intra-chromosomal recombination. The Coalescent-HMMs approximate the intricate correlation structure between these local \firstRevision{genealogical trees} by a Markov chain, the Sequentially Markovian Coalescent~\citep{Wiuf1999,McVean2005}.
\firstRevision{Due to the Markovian structure of the SMC-approximation, an essential building block is thus the transition or joint distribution of these local genealogies at two neighboring loci.
In a sample of size two, the local genealogies are simple trees with two leaves, that is, one-dimensional objects at each locus.
The transitions can be readily computed, and \cite{Li2011} employed this framework to develop a powerful approach to infer population size history.
\edit{Moreover, \cite{Dutheil2009} used Coalescent-HMMs to explore the divergence patterns between humans and great apes, using up to 4 genomic sequences, one for each species.}
However, due to the increase in complexity of the local genealogies with increasing sample size, these approaches cannot be generalized efficiently to larger sample sizes.}

\firstRevision{For large sample sizes, approaches that use Monte-Carlo Markov Chain techniques~\citep{Rasmussen2014}, suitable composite likelihood frameworks~\citep{Sheehan2013,Steinrucken2015}, or representations of the local genealogical trees by lower-dimensional summaries~\citep{Schiffels2014,Terhorst2017} have been developed.
In the latter, the choice on how to represent the local genealogical trees affects the performance of the inference procedure. \cite{Li2011} observed that using the coalescence time between two lineages lacks information in the more recent past, whereas using the first coalescence time in a large sample is less accurate for ancient times~\citep{Schiffels2014}.
A promising low-dimensional representation is the total branch length of the genealogical tree at each locus.
In expectation, this quantity grows without bound as the sample size increases, thus retaining not only information about ancient events, but also about the more recent dynamics.
However, to implement a Coalescent-HMM inference framework using the tree length, it is crucial to efficiently compute the joint distribution of the total tree length at two neighboring loci.
}


\firstRevision{Thus, in this paper,} we present a novel efficient and accurate method to \firstRevision{numerically} compute the joint distribution of the total branch length of the genealogical trees at two neighboring loci for a sample of arbitrary size $\sampleSize$ in populations of varying size, as well as the single-locus marginal distribution. To our knowledge, no method to compute these distributions has been presented in the literature to date that can be applied to arbitrary sample sizes. Moreover, even computing the marginal distribution of the total tree length at a single locus has only received limited attention~\citep{Pfaffelhuber2011,Wiuf1999}. \edit{We present analytical details and numerical results for the case of at most one recombination event separating the two loci, but our methodology can be readily extended to handle an arbitrary, but fixed, maximal number of recombination events, by suitably augmenting the underlying process.}

The inter-coalescent times $\bigTime^{(\sampleSize)}_k$, that is the time period during which $k$ lineages persist in the genealogical tree for a sample of size $\sampleSize$ can be used to compute the total branch length at a single locus  as
\begin{equation}\label{eq_basic_length}
  \totalTreeLength = \sum_{k=2}^n k \bigTime^{(\sampleSize)}_k,
\end{equation}
since in the period $\bigTime^{(\sampleSize)}_k$, $k$ lineages contribute towards the total length.
In the case of a panmictic population of constant size, formulas for the first two moments of the total tree length can be readily obtained using standard arguments for sums of the independently exponentially distributed random variables $\bigTime^{(\sampleSize)}_k$. \firstRevision{Furthermore, $\totalTreeLength$ is distributed like the maximum of $k-1$ exponential variables with intensity $\frac{1}{2}$~\cite[p.~255]{Wiuf1999}.} However, non-constant population size histories introduce intricate dependencies among the inter-coalescent times, and thus it is not straightforward  to generalize this approach.
\cite{Polanski2003} introduced a method to compute the expected inter-coalescence times under variable population size.
However, the coalescence rates of ancestral lineages in the genealogical process depend on past population sizes, whereas the rate for ancestral recombination is constant along each ancestral lineage. The approach of \cite{Polanski2003} depends on the fact that all rates of the process are rescaled uniformly with the same factor, and thus it cannot be extended to the case when ancestral recombination between two linked loci is taken into account.

\cite{Ferretti2013} used another approach to investigate the correlation between the times to the most recent common ancestor at two neighboring loci. The authors approached the problem using coalescent arguments to quantify the changes recombination induces on the local trees, but it is unclear how to generalize their approach efficiently to the total length of the genealogical trees. 
\edit{Furthermore,} \cite{Li2011} presented analytic formulas for the joint distribution of the local genealogies for a sample of size two under variable population size, \firstRevision{but these cannot readily be extended to an arbitrary sample size $n$.}
\edit{\cite{Eriksson2009} presented similar analytic formulas for a population of constant size and explored more complex demographic scenarios using simulations. Introducing suitable Markov chains, \cite{Hobolth2014} investigated the transitional distribution of the local genealogies for samples of size 4, and discussed approximations for larger sample sizes. These Markov chains are closely related to our methodology, but our focus is on exact computations for large sample sizes.}

\firstRevision{
Although we focus on the total tree length under variable population size in a single panmictic population in this paper, our approach can be extended to compute the transition densities for the coalescence time in a sample of size two~\citep{Li2011}, the coalescence time of two distinguished lineages~\citep{Terhorst2017}, and the time of the first coalescent event amongst the sampled sequences~\citep{Schiffels2014}. Furthermore, our method can be generalized to multiple sub-populations related by a complex demographic history (see discussion in Section~\ref{sec_discussion}).
}

\firstRevision{This article is structured as follows.} In Section~\ref{sec_background}, we introduce the requisite notation and the stochastic processes that are involved in computing the marginal and joint distributions. We further introduce a hyperbolic system of partial differential equations (PDEs) in Section~\ref{sec_cdf} that can be solved to compute the distributions of interest. We provide a proof of the main proposition used to derive these equations in Appendix~\ref{app_proof}. In Section~\ref{sec_cdf}, we also provide the details of our novel numerical algorithm based on the method of characteristics that can be used to efficiently compute the solutions to these PDEs. We demonstrate the accuracy of the method, and study the  properties of the joint distribution function in Section~\ref{sec_empirical}. Finally, we discuss the future applications and extensions of this method in Section~\ref{sec_discussion}.

\section{Background and Notation}
\label{sec_background}

In this section, we will introduce the necessary background and notation for the stochastic processes that we employ to compute the marginal and joint distribution of the length of the genealogical trees. We will also provide some details about computing the distribution of these processes, since our main result extends upon the underlying ideas.

\subsection{Ancestral Process at a Single Locus} 

The genealogical relationship of a sample of $\sampleSize$ haploid individuals in a panmictic population of constant size is commonly modeled using Kingman's coalescent~\citep{Kingman1982,Wakeley2008}, and this process and its extensions have found widespread applications. It is a Markov process that describes the dynamics of the ancestral lineages of the sample backwards in time.
Here we focus on the ancestral process $\ancestralProcess(\time)$ \citep[Chapter~4.1]{Tavare2004}. This coarser process records only the number of ancestral lineages in the coalescent process at time $\time$ before present, which is sufficient to compute the total branch length of the coalescent tree.
The initial number of lineages is equal to the sample size $\sampleSize$. Furthermore, at time $\time$, each pair of lineages coalesces at rate one, thus if there are $\ancestralProcess(\time) = k$ lineages at time $\time$, then coalescence of any two lineages happens at rate ${k \choose 2}$. 
This dynamics is followed until all lineages coalesced into a single lineage, and this time is denoted by $\tMRCA$, the time to the most recent common ancestor.




Variable population size is commonly modeled by a positive, real-valued function $\lambda(t)$, which provides the coalescent rate for each pair of ancestral lineages at time $\time$ in the past~\citep[Chapter~4.1]{Tavare2004}.
If the size of the population changes at different points in the past, the rate of coalescence at a given time is inversely proportional to the relative population size at that time. Intuitively, for two lineages to coalesce, they have to find a common ancestor. If the population consists of a large number of individuals, this happens at a lower rate, whereas in small populations, the ancestral lineages coalesce more quickly. 
In the remainder of this paper, we assume that $\invPopSize(\time)$ is continuous. If $\invPopSize(\time)$ is piece-wise continuous, we can obtain the same results by considering each continuous piece separately. For convenience, we further introduce the cumulative coalescent rate at time $\time$ as
\begin{equation}
	\cumInvPopSize (\time) = \int_0^\time \invPopSize(s) ds.
\end{equation}
These considerations yield the following definition.

\begin{definition}[Ancestral Process with variable population size]
	The \emph{ancestral process with variable population size} $\{ \ancestralProcess(\time) \}_{\time \in \R_+}$ is a time-inhomogeneous Markov chain on $\{1, \ldots, \sampleSize\}$ with initial state $\ancestralProcess(0) = \sampleSize$, and the transition rates at time $\time$ are given by the infinitesimal generator matrix
	\begin{equation}
		\ancestralRateMatrix (\time) = \invPopSize(\time) \ancestralRateMatrix,
	\end{equation}
	with
	\begin{equation}\label{eq_ancestral_rates}
		\ancestralRateMatrix_{k,j} := \begin{cases}
								- {k \choose 2},	& \text{if $j = k$},\\
								{k \choose 2},		& \text{if $j = k-1$},\\
								0,					& \text{otherwise}.\\
							\end{cases}
	\end{equation}
\end{definition}

\begin{remark}
	Note that we do require $\ancestralProcess(0) = \sampleSize$, and thus this definition of the ancestral process does depend on the sample size $\sampleSize$. However, for different sample sizes $\sampleSize'$, the rates of the process are given by equation~\eqref{eq_ancestral_rates} as well, only the initial state changes. The dynamics of the process is essentially the same, independent of the sample size, and we therefore do not include the dependence on the sample size explicitly in the notation for the remainder of this article.
\end{remark}

The ancestral process can be used to formally define the time to the most recent common ancestor as
\begin{equation}
	\tMRCA := \inf \big\{\time \in \R_+: \ancestralProcess(\time) \leq 1 \big\},
\end{equation}
the time when the number of lineages reaches one.
Furthermore, with
\begin{equation}
	\distAncestral_k(\time) := \P \big\{\ancestralProcess(\time) = k\}, 
\end{equation} 
for $k \in \{1,\ldots,\sampleSize\}$, the distribution of the ancestral process can be obtained by solving the Kolmogorov-forward-equation~\citep[Chapter~5]{Stroock2008}, a system of ordinary differential equations (ODEs) given by 
\begin{equation}\label{eq_kolmogorov_ancestral}
	\frac{d}{d\time} \big(\distAncestral_1(\time), \ldots, \distAncestral_\sampleSize(\time)\big) = \big(\distAncestral_1(\time), \ldots, \distAncestral_\sampleSize(\time)\big) \ancestralRateMatrix(t).
\end{equation}
Equivalently, perhaps more familiar to the reader, this system can be expressed as
\begin{equation}\label{eq_anc_proc_ode}
		\frac{d}{d\time} \distAncestral_\numLineages(\time) = \invPopSize(\time){k + 1\choose 2} \distAncestral_{\numLineages+1}(\time) - \invPopSize(\time){k \choose 2} \firstRevision{\distAncestral_{\numLineages}(\time)},
\end{equation}
for all $\numLineages \in \{1,\ldots,\sampleSize\}$. The latter version is more explicit about the influence of the number of ancestral lineages and the coalescent-speed function on the dynamics of the ODEs. The relevant solution is given by
\begin{equation}
	\big(\distAncestral_1(0), \ldots, \distAncestral_\sampleSize(0)\big) = (0,\ldots,0,1)
\end{equation}
and 
\begin{equation}\label{eq_dist_anc_matrix_exp}
	 \big(\distAncestral_1(\time), \ldots, \distAncestral_\sampleSize(\time)\big) = \firstRevision{\big(}e^{\cumInvPopSize(\time) \cdot \ancestralRateMatrix} \firstRevision{\big)}_{\sampleSize,\cdot}
\end{equation}
for $\time \in \R_+$, \firstRevision{where $(\cdot)_{n,\cdot}$ refers to the $n$-th row of the matrix.} In~\cite{Tavare2004}, the authors provide an analytic expression for these probabilities using the spectral decomposition of the rate matrix $\ancestralRateMatrix(\time)$. However, the resulting formulas are numerically unstable, so for practical purposes it can be more efficient to solve the system of ODEs numerically using step-wise solution schemes. Furthermore, note that
\begin{equation}\label{eq_distr_tmrca}
	\P \{\tMRCA \leq \time^*\} = \firstRevision{\big(}e^{\cumInvPopSize(\time^*) \cdot \ancestralRateMatrix} \firstRevision{\big)}_{\sampleSize,1}
\end{equation}
holds for $\time^* \in \R_+$, thus equation~\eqref{eq_dist_anc_matrix_exp} can also be used to compute the cumulative distribution function of the time to the most recent common ancestor.

We can employ the ancestral process to compute the total tree length as follows. If at a given time $\time$ there are $k$ ancestral lineages or branches in the coalescent tree, each branch extends further into the past. Thus, we can say that the total sum of branch lengths in the coalescent tree grows at a rate of $k$. Once all lineages have coalesced into a single common ancestral lineage, the most recent common ancestor is reached, and the coalescent tree stops growing. This motivates the following definition.

\begin{definition}
The \emph{accumulated tree length} $\treeLength (\time) \in \R_+$ by time $\time \in \R_+$ is given by
\begin{equation}\label{def_accu_length}
	\treeLength (\time) := \int_0^\time \1_{\{\ancestralProcess(s) > 1\}} \ancestralProcess(s) ds.
\end{equation}
\end{definition}
With this definition, the \emph{total tree length} or the \emph{total sum of the branch lengths} at a single locus is given by
\begin{equation}\label{eq_def_total_tree_length}
	\totalTreeLength := \treeLength \big( \tMRCA \big).
\end{equation}
Note that
\begin{equation}
  \totalTreeLength = \sum_{k=2}^n k \bigTime^{(\sampleSize)}_k
\end{equation}
holds, which is equal to equation~\eqref{eq_basic_length}. Here $\bigTime^{(\sampleSize)}_k$ is the period of time for which $k$ lineages persist in the ancestral process, the inter-coalescent time.
The main goal of this paper is to study the distribution of $\totalTreeLength$ for populations with arbitrary coalescent-rate function $\invPopSize(\time)$ marginally at a single locus and jointly at two loci, which can be computed using a system of hyperbolic PDEs that is closely related to the ODE~\eqref{eq_anc_proc_ode}. For the two-locus case, we will now introduce the joint ancestral process at two linked loci.


\subsection{Ancestral Process with Recombination}

The joint genealogy of the ancestral lineages for two loci, locus $\aLocus$ and $\bLocus$, separated by a recombination distance $\recoRate$ is commonly modeled by the coalescent with recombination~\citep{Hudson1990}.
The initial state in the coalescent with recombination for a sample of size $\sampleSize$ is comprised of $\sampleSize$ lineages, each ancestral to both loci of one sampled haplotype. As in the single-locus coalescent with variable population size, at time $\time$, each pair of lineages can coalesce at rate $\invPopSize(\time)$. In addition, ancestral recombination events happen at rate $\recoRate/2$ along each active lineage.
At a recombination event, the lineage splits into two new lineages, each ancestral to the respective haplotype of the original lineage at only one of the two loci.
Note that recombination happens along each lineage at a constant rate and, unlike the coalescent rate, is not affected by the population size, and thus it does not scale with $\invPopSize(\time)$.

Again, we do not focus on the exact genealogical relationships,
but only on the number of lineages at time $\time$ that are ancestral to a certain locus, given by the \emph{ancestral process with recombination} $\ancestralRecoProcess(\time)$. The process $\ancestralRecoProcess$ for a sample of size two under constant population size is described in detail by \cite{Simonsen1997}. Here we use an extension of this process to samples of arbitrary size $\sampleSize$ and variable population size. \edit{A similar model has also been introduced by \cite{Hobolth2014}.}

\begin{definition}[Ancestral Process with Recombination]\label{def_anc_proc_reco}
For a sample of size $\sampleSize \in \N$ and $\time \in \R_+$, the \emph{ancestral process with recombination} in a population of variable size
\begin{equation}
	\ancestralRecoProcess(\time) = \big( \numLineagesRV_{\aLocus\bLocus} (\time), \numLineagesRV_{\aLocus} (\time), \numLineagesRV_{\bLocus} (\time) \big)
\end{equation}
is a time-inhomogeneous Markov chain with state space
\begin{equation}
	\recoStates := \big\{ \state \in \N_0^3 \big| s_1 + \max\{\state_2,\state_3\} \leq \sampleSize \big\} \big\backslash \big\{ (0,0,0), (0,1,0), (0,0,1)\big\}.
\end{equation}
The component $\numLineagesRV_{\aLocus\bLocus} (\time)$ gives the number of lineages that are ancestral to both loci, $\numLineagesRV_{\aLocus} (\time)$ is the number ancestral to locus $\aLocus$ only, and $\numLineagesRV_{\bLocus} (\time)$ is the number ancestral to locus $\bLocus$ only. The initial state is
\begin{equation}
	\ancestralRecoProcess(0) = (\sampleSize,0,0),
\end{equation}
all $\sampleSize$ lineages ancestral to both loci. The transition rates are given by the infinitesimal generator matrix
\begin{equation}
	\ancestralRecoMatrix(t) = \invPopSize(\time) \ancestralRecoCoalMatrix + \ancestralRecoOnlyMatrix,
\end{equation}
where all off-diagonal entries of $\ancestralRecoCoalMatrix$ (coalescence) are zero, except
\begin{align}
	\ancestralRecoCoalMatrix_{(\numLineages_{\aLocus\bLocus},\numLineages_{\aLocus},\numLineages_{\bLocus}), (\numLineages_{\aLocus\bLocus}-1,\numLineages_{\aLocus},\numLineages_{\bLocus})} & = {\numLineages_{\aLocus\bLocus} \choose 2},\\
	\ancestralRecoCoalMatrix_{(\numLineages_{\aLocus\bLocus},\numLineages_{\aLocus},\numLineages_{\bLocus}), (\numLineages_{\aLocus\bLocus},\numLineages_{\aLocus}-1,\numLineages_{\bLocus})} & = {\numLineages_{\aLocus} \choose 2} + \numLineages_{\aLocus\bLocus} \numLineages_{\aLocus},\\
	\ancestralRecoCoalMatrix_{(\numLineages_{\aLocus\bLocus},\numLineages_{\aLocus},\numLineages_{\bLocus}), (\numLineages_{\aLocus\bLocus},\numLineages_{\aLocus},\numLineages_{\bLocus}-1)} & = {\numLineages_{\bLocus} \choose 2} + \numLineages_{\aLocus\bLocus} \numLineages_{\bLocus},\\
	\text{and} & \\
	\ancestralRecoCoalMatrix_{(\numLineages_{\aLocus\bLocus},\numLineages_{\aLocus},\numLineages_{\bLocus}), (\numLineages_{\aLocus\bLocus}+1,\numLineages_{\aLocus}-1,\numLineages_{\bLocus}-1)} & = \numLineages_{\aLocus} \numLineages_{\bLocus}, \label{eq_smc_rates}
\end{align}
and all off-diagonal entries of $\ancestralRecoOnlyMatrix$ (recombination) are zero, except
\begin{equation}
\begin{split}
	\ancestralRecoOnlyMatrix_{(\numLineages_{\aLocus\bLocus},\numLineages_{\aLocus},\numLineages_{\bLocus}), (\numLineages_{\aLocus\bLocus}-1,\numLineages_{\aLocus}+1,\numLineages_{\bLocus}+1)} & = \frac{\recoRate}{2} \numLineages_{\aLocus\bLocus}.\\
\end{split}
\end{equation}
The state $(1,0,0)$ is defined to be the absorbing state, so all rates leaving this state are set to zero. Furthermore, the diagonal entries of both matrices are set to minus the sum of the off-diagonal entries in the corresponding row.
\end{definition}


\begin{remark}
i) Two versions of the coalescent with recombination are commonly used in the literature, one version for the infinitely-many-sites (IMS) model~\citep{Hudson1990,Griffiths1997}, and another version for the finitely-many-sites (FMS) model~\citep{Paul2011,Steinrucken2015}. In the IMS version, the chromosome is modeled as the interval $[0,1]$, and whenever recombination occurs, it occurs at a uniformly chosen point in this interval. As a result, recombination always occurs at a novel site, and two neighboring local genealogies are separated by at most one recombination event. In the FMS version, multiple recombination events can occur between two loci. It can be obtained from the IMS version by considering the local genealogies at two fixed loci along the continuous chromosome that are separated by a certain fixed recombination distance. Our definition of the ancestral process with recombination is in line with the FMS version for two loci.
\\
ii) The ancestral process with recombination can be defined for an arbitrary number of loci. However, in the remainder of the paper, we will only use the process for two loci.\\
iii) In the literature, some authors use the `full' coalescent with recombination and others the `reduced' coalescent with recombination. The difference between the two is that the `full' version always keeps track of both ancestral lineages that branch off at a recombination event, whereas in the `reduced' version, lineages that do not leave any descendant ancestral material in the contemporary sample are not traced. Our definition of the ancestral process with recombination is compatible with the `reduced' version. Thus, the number of ancestral lineages is bounded, which is not the case in the `full' version.
\\
iv) Following the ideas of~\cite{Wiuf1999}, the correlation structure between all local genealogies along a chromosome can be approximated using the Sequentially Markovian Coalescent (SMC)~\citep{McVean2005}, or the modified version SMC'~\citep{Marjoram2006}. In the SMC, if a lineage has been hit by a recombination event and branches into two, subsequently, the two resulting branches are not allowed to coalesce with each other, whereas such events are permitted under the SMC'. Thus, under the SMC', the rates for coalescence of lineages with no overlapping ancestral material (equation~\eqref{eq_smc_rates}) are as given in Definition~\ref{def_anc_proc_reco}, whereas under the SMC, these rates have to be set to zero.
\end{remark}

Again, the Kolmogorov-forward-equation can be used to compute the distribution of the ancestral process $\ancestralRecoProcess(\time)$ as the solution of
\begin{equation}\label{eq_kolmogorov_ancestral_reco}
	\frac{d}{d\time} \distVecAncestral (\time) = \distVecAncestral (\time) \ancestralRecoMatrix(\time),
\end{equation}
where the row-vector $\distVecAncestral (\time)$ is defined by 
\begin{equation}
 	\distVecAncestral (\time) := \Big(\P \big\{\ancestralRecoProcess(\time) = \state \big\}\Big)_{\state \in \recoStates}.
\end{equation}

Note that the rate matrix $\ancestralRateMatrix(\time)$ in the ODE~\eqref{eq_kolmogorov_ancestral} for the ancestral process at a single locus is triangular for all $\time$. This simplifies approaches to compute solutions substantially, as the solutions can be obtained sequentially for each state of the corresponding Markov chain. In the ancestral process with recombination for two loci on the other hand, with a positive probability, the underlying Markov chain can transition back to a state it already visited before. Consequently, the rate matrix $\ancestralRecoMatrix(\time)$ in the ODE~\eqref{eq_kolmogorov_ancestral_reco} is not triangular, and it is also not possible to transform it into a triangular matrix by permuting the rows and columns.

Since a triangular rate matrix simplifies analytical and numerical approaches significantly, we introduce an approximation to the full ancestral process with recombination that exhibits this property and compute the distributions of the tree lengths under this approximation. To achieve this, we explicitly account for the number of recombination events that have occurred up to a certain time $\time$. For ease of exposition, we further limit the maximal number of recombination events to one. Since in most organism\firstRevision{s} the per generation recombination probability is very small between loci that are physically close, this approximation is justified. Furthermore, numerical experiments supporting this approximation are provided in Section~\ref{sec_empirical}.
\firstRevision{Note that this limiting the number of recombination events to one yields effectively a first-order approximation to the full ancestral process.}

\begin{definition}[Ancestral Process with Limited Recombination]
For a sample of size $\sampleSize \in \N$ and $\time \in \R_+$, the \emph{ancestral process with limited recombination}
\begin{equation}
	\ancestralRecoLimitedProcess(\time) = \big( \numLineagesRVLimited_{\aLocus\bLocus} (\time), \numLineagesRVLimited_{\aLocus} (\time), \numLineagesRVLimited_{\bLocus} (\time), \numRecoRV(\time) \big)
\end{equation}
is a time-inhomogeneous Markov chain with state space
\begin{equation}\label{def_reco_limited_states}
	\begin{split}
		\recoLimitedStates := & \Big(\{1,\ldots,\sampleSize\} \times \{(0,0,0)\} \\
			& \: \cup \{1,\ldots,\sampleSize\} \times \{0,1\} \times \{0,1\} \times \{1\} \Big)\\
			& \: \qquad \qquad \backslash \big\{(n,1,1,1),(n,1,0,1),(n,0,1,1)\big\}.
	\end{split}
\end{equation}
The components $\numLineagesRVLimited_{\aLocus\bLocus} (\time)$, $\numLineagesRVLimited_{\aLocus} (\time)$, and $\numLineagesRVLimited_{\bLocus} (\time)$ have the same interpretation as before, and $\numRecoRV(\time)$ is the number of recombination events that have happened by time $\time$. \edit{The first line in equation~\eqref{def_reco_limited_states} corresponds to the states that can be reached without recombination, and the second line to those that require one recombination event.} The initial state is
\begin{equation}
	\ancestralRecoLimitedProcess(0) = (\sampleSize,0,0,0),
\end{equation}
and the transition rates are given by the infinitesimal generator matrix
\begin{equation}\label{def_limited_reco_matrix}
	\ancestralRecoLimitedMatrix(t) = \invPopSize(\time) \ancestralRecoCoalLimitedMatrix + \ancestralRecoOnlyLimitedMatrix,
\end{equation}
where the entries of $\ancestralRecoCoalLimitedMatrix$ (coalescence) are given by
\begin{equation}
\begin{split}
	\ancestralRecoCoalLimitedMatrix_{(\numLineages_{\aLocus\bLocus},\numLineages_{\aLocus},\numLineages_{\bLocus},\numReco), (\numLineages_{\aLocus\bLocus},\numLineages_{\aLocus},\numLineages_{\bLocus},\numReco)} & = \ancestralRecoCoalMatrix_{(\numLineages_{\aLocus\bLocus},\numLineages_{\aLocus},\numLineages_{\bLocus}), (\numLineages_{\aLocus\bLocus},\numLineages_{\aLocus},\numLineages_{\bLocus})},
\end{split}
\end{equation}
and all off-diagonal entries of $\ancestralRecoOnlyMatrix$ (recombination) are zero, except
\begin{equation}
\begin{split}
	\ancestralRecoOnlyLimitedMatrix_{(\numLineages_{\aLocus\bLocus},\numLineages_{\aLocus},\numLineages_{\bLocus},0), (\numLineages_{\aLocus\bLocus}-1,\numLineages_{\aLocus}+1,\numLineages_{\bLocus}+1,1)} & = \frac{\recoRate}{2} \numLineages_{\aLocus\bLocus},\\
\end{split}
\end{equation}
allowing at most one recombination event.
The diagonal entries are set to minus the sum of the off-diagonal entries in the corresponding row.
The states $(1,0,0,0)$ and $(1,0,0,1)$ are absorbing states, so all rates leaving these states are set to zero. 
\end{definition}

For later convenience, define the relation $\prec$ on $\recoLimitedStates$ as
\begin{equation}\label{def_rel}
	\state \prec \state' :\Leftrightarrow \ancestralRecoLimitedMatrix_{\state',\state}(t) > 0,
\end{equation}
that is, $\state \prec \state'$ holds if $\state$ can be reached from $\state'$ in one step. Note that embedded into the ancestral process with recombination (limited or not) is a single-locus ancestral process for locus $\aLocus$ and for locus $\bLocus$. Thus, we can define the branch length of the genealogical tree at locus $\aLocus$ and $\bLocus$ similar to the one-locus case as follows, and study their joint distribution.

\begin{definition}
For a given time $\time \in \R_+$, the \emph{accumulated tree lengths} $\treeLength^\aLocus (\time) \in \R^+$ at locus $\aLocus$ and $\treeLength^\bLocus (\time) \in \R_+$ at locus $\bLocus$ are given by
\begin{equation}
	\treeLength^\aLocus (\time) := \int_0^\time \1_{\{\numLineagesRVLimited_{\aLocus\bLocus}(s) + \numLineagesRVLimited_{\aLocus}(s) > 1\}} \big( \numLineagesRVLimited_{\aLocus\bLocus}(s) + \numLineagesRVLimited_{\aLocus}(s) \big) ds,
\end{equation}
and
\begin{equation}
	\treeLength^\bLocus (\time) := \int_0^\time \1_{\{\numLineagesRVLimited_{\aLocus\bLocus}(s) + \numLineagesRVLimited_{\bLocus}(s) > 1\}} \big( \numLineagesRVLimited_{\aLocus\bLocus}(s) + \numLineagesRVLimited_{\bLocus}(s) \big) ds.
\end{equation}
\end{definition}
\begin{remark}
This definition of the accumulated tree length can be applied to $\ancestralRecoLimitedProcess$, as well as $\ancestralRecoProcess$. We will not distinguish these cases in our notation, since in the remainder of the paper, we will use $\ancestralRecoLimitedProcess$.
\end{remark}
The \emph{total tree length} at locus $\aLocus$ is thus given by
\begin{equation}\label{def_total_length_a}
	\totalTreeLength^\aLocus := \treeLength^\aLocus \big( \tMRCA^\aLocus \big),
\end{equation}
and at locus $\bLocus$ by
\begin{equation}\label{def_total_length_b}
	\totalTreeLength^\bLocus := \treeLength^\bLocus \big( \tMRCA^\bLocus \big).
\end{equation}
Here, $\tMRCA^\aLocus$ is the time to the most recent common ancestor at locus $\aLocus$
\begin{equation}
	\tMRCA^\aLocus := \inf \big\{ \time \in \RR_+:  \numLineagesRVLimited_{\aLocus\bLocus}(\time) + \numLineagesRVLimited_{\aLocus}(\time) \leq 1 \big\},
\end{equation}
and thus its distribution is given by
\begin{equation}
	\P \{\tMRCA^\aLocus \leq \time^*\} = \P \big\{ \numLineagesRVLimited_{\aLocus\bLocus}(\time^*) + \numLineagesRVLimited_{\aLocus}(\time^*) \leq 1 \big\}
\end{equation}
for $\time^* \in \R_+$. Similar relations hold for locus $\bLocus$. We will now study the joint distribution of $\totalTreeLength^\aLocus$ and $\totalTreeLength^\bLocus$, and also the marginal $\totalTreeLength$. Note that these quantities are computed under the ancestral process with limited recombination, but we will demonstrate in Section~\ref{sec_empirical} that they give an accurate approximation to the respective quantities under the true ancestral process.

\section{Marginal and Joint Distribution of the Total Tree Length}
\label{sec_cdf}

The main goal of this paper is to present a method to compute the marginal and joint cumulative distribution function (CDF) of the total tree length at two linked loci. Thus, we aim at computing
\begin{equation}\label{eq_marginal_cdf}
	\P \{ \totalTreeLength \leq \treeDimA\}
\end{equation}
and
\begin{equation}\label{eq_joint_cdf}
	\P \{ \totalTreeLength^\aLocus \leq \treeDimA , \totalTreeLength^\bLocus \leq \treeDimB \}
\end{equation}
for $\treeDimA, \treeDimB \in \R_+$.

Note that equation~\eqref{eq_distr_tmrca} can be used to compute \firstRevision{the distribution of the time until the ancestral process reaches the absorbing state, which yields the marginal distribution of the $\tMRCA$}.
\firstRevision{The latter is equal to the sum of the inter-coalescence times, and in a population of constant size, equation~\eqref{eq_distr_tmrca} can also be obtained by convolving the densities of independent exponential variables.}
The total branch length is a more general linear combination of \firstRevision{the inter-coalescence times, but~\cite{Wiuf1999} used a similar convolution approach to derive its marginal density.}
\firstRevision{In a population with variable population size, the inter-coalescence times are not mutually independent. However, \cite{Polanski2003} derived formulas for the density of $\tMRCA$ using a uniform rescaling of time by the coalescent-rate function $\invPopSize(\time)$.}

\firstRevision{The two main difficulties in extending these considerations to the total tree length in a two-locus model with variable population size are as follows:}
\firstRevision{First, in a model that includes recombination, only the coalescence rates scale with $\invPopSize(\time)$, while the recombination rate is constant along each lineage.}
\firstRevision{The approach of~\cite{Polanski2003}, however, relies on a uniform rescaling of all rates,}
\firstRevision{and therefore it cannot be applied.} 
\firstRevision{Second, note that, similar to equation~\eqref{def_accu_length}, we can define
\begin{equation}
	T (\time) := \int_0^\time \1_{\{\ancestralProcess(s) > 1\}} ds = \begin{cases}
		t,	& \text{if $t < \tMRCA$},\\
		\tMRCA,	& \text{otherwise}.
	\end{cases}
\end{equation}
With this definition, the quantity $\time$ is not only the time elapsed in the ancestral process,
but it can also be interpreted as the amount accumulated towards $\tMRCA$.}
\firstRevision{The absorption time of the ancestral process can thus be used in equation~\eqref{eq_distr_tmrca} to compute the distribution of $\tMRCA$.}
\firstRevision{However, when the accumulated tree length $\treeLength (\time)$ defined in equation~\eqref{def_accu_length} is considered, the quantity $\time$ only gives the elapsed time, and it cannot be used as the amount accumulated towards $\totalTreeLength$.}
\firstRevision{Thus, our approach to compute the distribution of $\totalTreeLength$, and the joint distribution of $\totalTreeLength^\aLocus$ and $\totalTreeLength^\bLocus$ has to explicitly account for both, the time that has elapsed in the ancestral process, as well as the amount accumulated towards the total tree length.}



To this end, with $\time \in \R_+$, we introduce the time-dependent cumulative distribution functions
\begin{equation}\label{eq_temp_cdf}
	\eff_\numLineages (\time, \treeDimA) := \P \big\{ \ancestralProcess(\time) = \numLineages, \treeLength(\time) \leq \treeDimA \big\}
\end{equation}
for $\numLineages \in \{1,\ldots,\sampleSize\}$ and
\begin{equation}\label{eq_temp_joint_cdf}
	\eff_\state (\time, \treeDimA, \treeDimB) := \P \big\{ \ancestralRecoLimitedProcess(\time) = \state, \treeLength^\aLocus(\time) \leq \treeDimA, \treeLength^\bLocus(\time) \leq \treeDimB \big\}
\end{equation}
for $\state \in \recoLimitedStates$.

We will show that the CDFs~\eqref{eq_marginal_cdf} and~\eqref{eq_joint_cdf} can be computed from the time-dependent CDFs~\eqref{eq_temp_cdf} and~\eqref{eq_temp_joint_cdf}. Furthermore, we will present numerical schemes, to efficiently and accurately compute the time-dependent CDFs~\eqref{eq_temp_cdf} and~\eqref{eq_temp_joint_cdf}.

\subsection{Distribution of the Total Tree length at a Single Locus}

\label{sec_cdf_marginal}

The following lemma shows that the CDF~\eqref{eq_marginal_cdf} can be computed from the time-dependent CDF~\eqref{eq_temp_cdf}.
\begin{lemma}\label{lem_cdf}
With definition~\eqref{eq_temp_cdf}, the relation
\begin{equation}
	\P \{ \totalTreeLength \leq \treeDimA\} = \P \big\{ \ancestralProcess(\bar\time) = 1, \treeLength(\bar\time) \leq \treeDimA \big\} = \eff_1 (\bar\time, \treeDimA )
\end{equation}
holds for $\treeDimA \in \R_+$ and $\bar\time \geq \treeDimA/2$.
\end{lemma}

\begin{proof}
First, observe that
\begin{equation}
	2 \tMRCA \leq \int_0^{\tMRCA} \1_{\{\ancestralProcess(s) > 1\}} \ancestralProcess(s) \, ds = \totalTreeLength,
\end{equation}
since $\ancestralProcess(s) \geq 2$ holds for $s < \tMRCA$.
Thus, on the event $\big\{ \totalTreeLength \leq \treeDimA \big\}$, the relation $\tMRCA \leq x/2 \leq \bar\time$ holds, which implies $\ancestralProcess(\bar\time) = 1$, and
therefore
\begin{equation}
	\big\{ \totalTreeLength \leq \treeDimA  \big\} = \big\{ \ancestralProcess(\bar\time) = 1, \totalTreeLength \leq \treeDimA \big\}.
\end{equation}
On the event  $\big\{ \ancestralProcess(\bar\time) = 1 \big\}$, $\bar\time \geq \tMRCA$ and $\treeLength(\bar\time) = \totalTreeLength$ hold, and thus
\begin{equation}
	\big\{ \ancestralProcess(\bar\time) = 1, \totalTreeLength \leq \treeDimA \big\} = \big\{ \ancestralProcess(\bar\time) = 1, \treeLength(\bar\time) \leq \treeDimA \big\},
\end{equation}
which proves the statement of the lemma.
\end{proof}

Lemma~\ref{lem_cdf} shows that the CDF of $\totalTreeLength$ can be computed from the time-dependent CDF $\eff_1 (\time, \treeDimA)$. Due to the structure of the underlying Markov chain, it is necessary to compute the time-dependent CDFs for all states in order to compute it for the absorbing state. Thus, in the remainder of this section, we focus on computing the time-dependent CDFs for all $\numLineages \in \{1,\ldots,\sampleSize\}$. Proposition~\ref{prop_app} derived in Appendix~\ref{app_proof} can be applied to show that the time-dependent CDFs solve a certain system of linear hyperbolic PDEs. This yields the following corollary.

\begin{corollary}\label{cor_marginal_pde}
The row-vector
\begin{equation}
	\vecEff (\time,\treeDimA) := \big( \eff_1 (\time, \treeDimA), \ldots, \eff_\sampleSize (\time, \treeDimA) \big)
\end{equation}
can be obtained for all points in ${\mc U}=\Big\{(x,t): 0 < \treeDimA <\sampleSize \time, \, t>0\Big\}$ as the strong solution of
\begin{equation}\label{eq_marginal_pde}
	\begin{split}
		\partial_\time \vecEff (\time,\treeDimA) + \edit{\partial_\treeDimA \vecEff (\time, \treeDimA)\stateFunctionMatrix} = \vecEff (\time, \treeDimA) \ancestralRateMatrix(\time),
	\end{split}
\end{equation}
with
\begin{equation}
	\stateFunctionMatrix = \diag(0,2,3,\ldots,\sampleSize),
\end{equation}
boundary conditions 
\begin{equation}\label{eq_marginal_boundary}
\begin{aligned}
 	\vecEff (\time, \treeDimA) &= \big( \firstRevision{\P} \big\{ \ancestralProcess(\time) = 1 \big\}, \ldots, \P \big\{ \ancestralProcess(\time) = \sampleSize-1\big\}, 0\big), \quad x=nt\\[2pt]
 \vecEff (\time, 0) &= \big(  0, \, 0, \, \dots, \,0 \big), \quad t>0\,,
\end{aligned}
\end{equation}
and matrix $\ancestralRateMatrix(\time)$ as defined in equation~\eqref{eq_ancestral_rates}.
\end{corollary}

\begin{proof}
Define the function
\begin{equation}
	\stateFunction (\numLineages) :=  \numLineages \cdot \1_{\{\numLineages > 1\}}
\end{equation}
on the state space $\ancStates = \{1,\ldots,\sampleSize\}$ of the ancestral process. This function and the generator $\ancestralRateMatrix(\time)$ satisfy the requirements of Proposition~\ref{prop_app}, and thus, the statement of the corollary follows from \firstRevision{Proposition~\ref{prop_app} and Remark~\ref{rem_zero}.}
\end{proof}

\begin{remark}
The $\sampleSize$-th component of the boundary condition~\eqref{eq_marginal_boundary} is equal to 0 and not $\P \big\{ \ancestralProcess(\time) = \sampleSize\big\}$. This holds for technical reasons that will be detailed in the proof of Proposition~\ref{prop_app}.
\end{remark}



\firstRevision{Note that the process $\big(\ancestralProcess(\time), \treeLength(\time) \big)_{\time \in \R_+}$ is a piecewise-deterministic Markov process (see Remark~\ref{rem_generator_d}).} The right-hand side of equation~\eqref{eq_marginal_pde} is essentially equal to the right-hand side of equation~\eqref{eq_anc_proc_ode}, because the only stochastic element in the underlying dynamics is the ancestral process $\ancestralProcess(\time)$. Given a certain number of lineages $\big\{ \ancestralProcess(\time) = \numLineages\big\}$, the accumulation towards the total tree length happens deterministically at rate $\numLineages$, and is captured by the term $\stateFunctionMatrix\partial_\treeDimA \vecEff (\time, \treeDimA)$.

%

\edit{To derive} a numerical scheme for \edit{the efficient and accurate computation of the} time-dependent CDF $\vecEff (\time,\treeDimA)$, \edit{note that} the system of PDEs introduced in Corollary~\ref{cor_marginal_pde} can be solved using the method of characteristics~\cite[Chapter~3]{Renardy2004}.
Due to the triangular structure of the matrix $\ancestralRateMatrix(\time)$, for a given component with $\numLineages \in \{1,\ldots,\sampleSize\}$, the right-side of equation~\eqref{eq_marginal_pde} does only depend on $\eff_\ell$ with $\ell \geq \numLineages$. Thus, the system of PDEs \eqref{eq_marginal_cdf} can be solved separately for each $\numLineages$, starting at $\numLineages = \sampleSize$, and decreasing it step-by-step.

Furthermore, note that \firstRevision{for $\numLineages \in \{ 1,\ldots,\sampleSize \}$,
\begin{equation}\label{eq_eff_regions}
	\P \big\{ \ancestralProcess(\time) = \numLineages, \treeLength(\time) \leq \treeDimA \big\} = 0, \qquad \text{if $\treeDimA < \stateFunction (\numLineages) \time$},
\end{equation}
since if the ancestral process has $\numLineages$ lineages at time $\time$, it must have accumulated at least $\stateFunction (\numLineages) \time$ towards the total tree length. It can be shown that the solution to equation~\eqref{eq_marginal_pde} exhibits this property. Moreover,
\begin{equation}
	\P \big\{ \ancestralProcess(\time) = \numLineages, \treeLength(\time) \leq \treeDimA \big\} = \P \big\{ \ancestralProcess(\time) = \numLineages \big\}, \qquad \text{if $\treeDimA \geq \sampleSize \cdot \time$},
\end{equation}
since the process can have accumulated at most $\sampleSize \time$.
Thus, we only have to use equation~\eqref{eq_marginal_pde} to compute the solution
\begin{equation}
	\eff_\numLineages (\time, \treeDimA) = \P \big\{ \ancestralProcess(\time) = \numLineages, \treeLength(\time) \leq \treeDimA \big\}
\end{equation}
when $\stateFunction (\numLineages) \time \leq \treeDimA  < \sampleSize \cdot \time$.} 
Note that $\stateFunction (1) = 0$.
Moreover, for $\numLineages = \sampleSize$, the region $\stateFunction (\numLineages) \time \leq \treeDimA  < \sampleSize \cdot \time$ is empty, and thus $\eff_\sampleSize (\time, \treeDimA)$ has a discontinuity along the line $\sampleSize \cdot \time$.
See Figure~\ref{fig_1locus_regions} for a visualization of the different regions for different values of $\numLineages$.
To devise an accurate and efficient numerical scheme for computing the time-dependent CDFs in the interior region, we use the method of characteristics to solve the respective PDE
\begin{equation}\label{eq_marginal_pde_comp}
  \partial_\time \eff_\numLineages (\time,\treeDimA) + \stateFunction (\numLineages) \partial_\treeDimA \eff_\numLineages (\time, \treeDimA) = \eff_\numLineages (\time, \treeDimA) \ancestralRateMatrix_{\numLineages,\numLineages} (\time) + \eff_{\numLineages+1} (\time, \treeDimA) \ancestralRateMatrix_{\numLineages+1,\numLineages} (\time).
\end{equation}

\begin{figure}
\newcommand{\xAx}{$\treeDimA$}
\newcommand{\tAx}{$\time$}
\newcommand{\zero}{$0$}
\newcommand{\nT}{$\treeDimA = \sampleSize \cdot \time$}
\newcommand{\kT}{$\treeDimA = \numLineages \cdot \time$}
\newcommand{\twoT}{$\treeDimA=2 \cdot \time$}
\begin{subfigure}[t]{0.32\textwidth}
	\def\svgwidth{\textwidth}
	\begin{center}
	\newcommand{\pT}{$\P \big\{ \ancestralProcess(\time) = \sampleSize \big\}$}
\begingroup%
  \makeatletter%
  \providecommand\color[2][]{%
    \errmessage{(Inkscape) Color is used for the text in Inkscape, but the package 'color.sty' is not loaded}%
    \renewcommand\color[2][]{}%
  }%
  \providecommand\transparent[1]{%
    \errmessage{(Inkscape) Transparency is used (non-zero) for the text in Inkscape, but the package 'transparent.sty' is not loaded}%
    \renewcommand\transparent[1]{}%
  }%
  \providecommand\rotatebox[2]{#2}%
  \ifx\svgwidth\undefined%
    \setlength{\unitlength}{277.05754798bp}%
    \ifx\svgscale\undefined%
      \relax%
    \else%
      \setlength{\unitlength}{\unitlength * \real{\svgscale}}%
    \fi%
  \else%
    \setlength{\unitlength}{\svgwidth}%
  \fi%
  \global\let\svgwidth\undefined%
  \global\let\svgscale\undefined%
  \makeatother%
  \begin{picture}(1,0.89302092)%
    \put(0,0){\includegraphics[width=\unitlength,page=1]{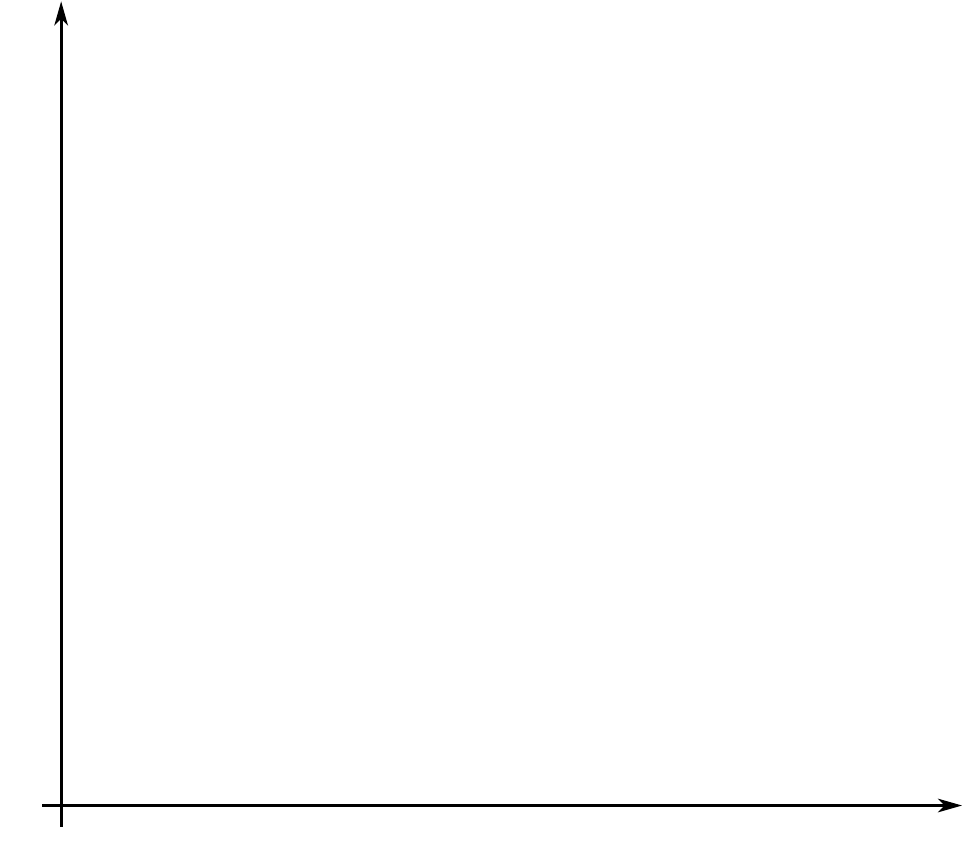}}%
    \put(-0.00148041,0.49164306){\color[rgb]{0,0,0}\makebox(0,0)[lb]{\smash{\tAx}}}%
    \put(0.45089941,0.00642915){\color[rgb]{0,0,0}\makebox(0,0)[lb]{\smash{\xAx}}}%
    \put(0.15539323,0.71600884){\color[rgb]{0,0,0}\makebox(0,0)[lb]{\smash{}}}%
    \put(0.51291918,0.13411713){\color[rgb]{0,0,0}\makebox(0,0)[lb]{\smash{\pT}}}%
    \put(0.27943281,0.54819058){\color[rgb]{0,0,0}\makebox(0,0)[lb]{\smash{\zero}}}%
    \put(0,0){\includegraphics[width=\unitlength,page=2]{F_n.pdf}}%
    \put(0.74458143,0.31287989){\color[rgb]{0,0,0}\makebox(0,0)[lb]{\smash{\nT}}}%
  \end{picture}%
\endgroup%
	\end{center}
	\caption{The two only regions for the initial state $\numLineages = \sampleSize$. The function has a discontinuity at $\treeDimA = \sampleSize \time$.}
\end{subfigure}
\hfill
\begin{subfigure}[t]{0.32\textwidth}
	\def\svgwidth{\textwidth}
	\begin{center}
	\newcommand{\pT}{$\P \big\{ \ancestralProcess(\time) = \numLineages \big\}$}
\begingroup%
  \makeatletter%
  \providecommand\color[2][]{%
    \errmessage{(Inkscape) Color is used for the text in Inkscape, but the package 'color.sty' is not loaded}%
    \renewcommand\color[2][]{}%
  }%
  \providecommand\transparent[1]{%
    \errmessage{(Inkscape) Transparency is used (non-zero) for the text in Inkscape, but the package 'transparent.sty' is not loaded}%
    \renewcommand\transparent[1]{}%
  }%
  \providecommand\rotatebox[2]{#2}%
  \ifx\svgwidth\undefined%
    \setlength{\unitlength}{277.05754798bp}%
    \ifx\svgscale\undefined%
      \relax%
    \else%
      \setlength{\unitlength}{\unitlength * \real{\svgscale}}%
    \fi%
  \else%
    \setlength{\unitlength}{\svgwidth}%
  \fi%
  \global\let\svgwidth\undefined%
  \global\let\svgscale\undefined%
  \makeatother%
  \begin{picture}(1,0.89302092)%
    \put(0,0){\includegraphics[width=\unitlength,page=1]{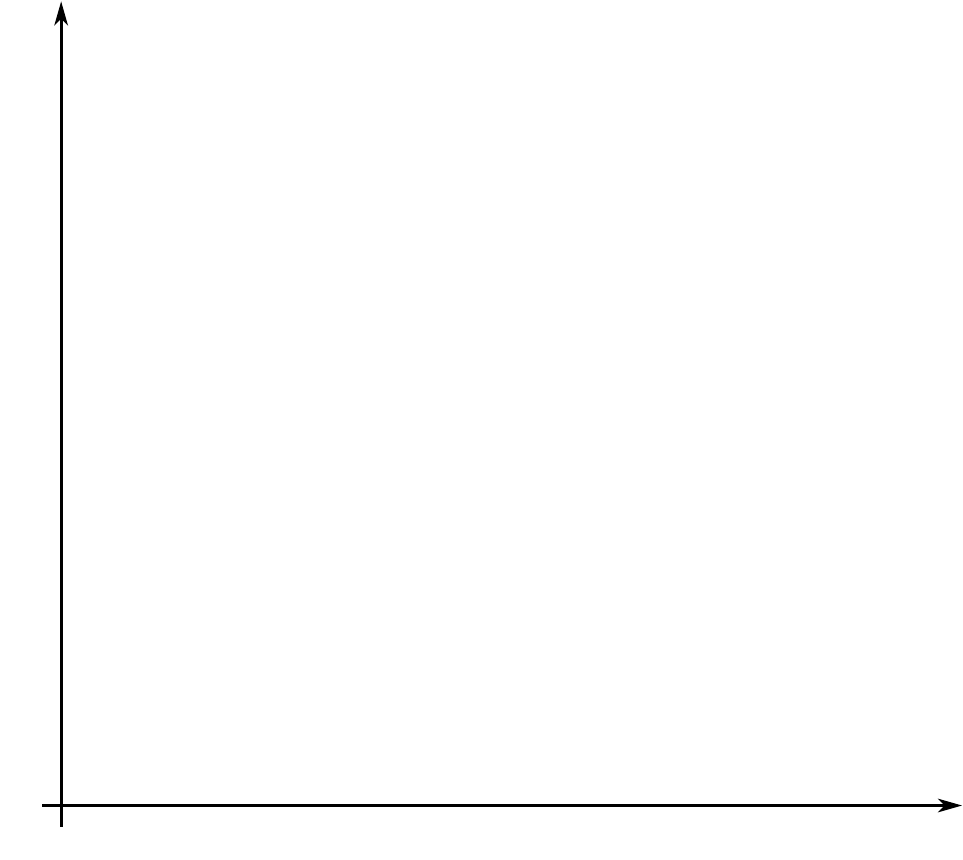}}%
    \put(-0.00148041,0.49164306){\color[rgb]{0,0,0}\makebox(0,0)[lb]{\smash{\tAx}}}%
    \put(0.45089941,0.00642915){\color[rgb]{0,0,0}\makebox(0,0)[lb]{\smash{\xAx}}}%
    \put(0.15539323,0.71600884){\color[rgb]{0,0,0}\makebox(0,0)[lb]{\smash{}}}%
    \put(0.51291918,0.13411713){\color[rgb]{0,0,0}\makebox(0,0)[lb]{\smash{\pT}}}%
    \put(0.27943281,0.54819058){\color[rgb]{0,0,0}\makebox(0,0)[lb]{\smash{\zero}}}%
    \put(0,0){\includegraphics[width=\unitlength,page=2]{F_k.pdf}}%
    \put(0.74458143,0.31287989){\color[rgb]{0,0,0}\makebox(0,0)[lb]{\smash{\nT}}}%
    \put(0,0){\includegraphics[width=\unitlength,page=3]{F_k.pdf}}%
    \put(0.737285,0.86011372){\color[rgb]{0,0,0}\makebox(0,0)[lb]{\smash{\kT}}}%
    \put(0,0){\includegraphics[width=\unitlength,page=4]{F_k.pdf}}%
  \end{picture}%
\endgroup%
	\end{center}
	\caption{The three regions and the characteristics in the interior for an intermediate state with $1 < \numLineages < \sampleSize$.}
\end{subfigure}
\hfill
\begin{subfigure}[t]{0.32\textwidth}
	\def\svgwidth{\textwidth}
	\begin{center}
	\newcommand{\pT}{$\P \big\{ \ancestralProcess(\time) = 1 \big\}$}
\begingroup%
  \makeatletter%
  \providecommand\color[2][]{%
    \errmessage{(Inkscape) Color is used for the text in Inkscape, but the package 'color.sty' is not loaded}%
    \renewcommand\color[2][]{}%
  }%
  \providecommand\transparent[1]{%
    \errmessage{(Inkscape) Transparency is used (non-zero) for the text in Inkscape, but the package 'transparent.sty' is not loaded}%
    \renewcommand\transparent[1]{}%
  }%
  \providecommand\rotatebox[2]{#2}%
  \ifx\svgwidth\undefined%
    \setlength{\unitlength}{277.05754798bp}%
    \ifx\svgscale\undefined%
      \relax%
    \else%
      \setlength{\unitlength}{\unitlength * \real{\svgscale}}%
    \fi%
  \else%
    \setlength{\unitlength}{\svgwidth}%
  \fi%
  \global\let\svgwidth\undefined%
  \global\let\svgscale\undefined%
  \makeatother%
  \begin{picture}(1,0.89302092)%
    \put(0,0){\includegraphics[width=\unitlength,page=1]{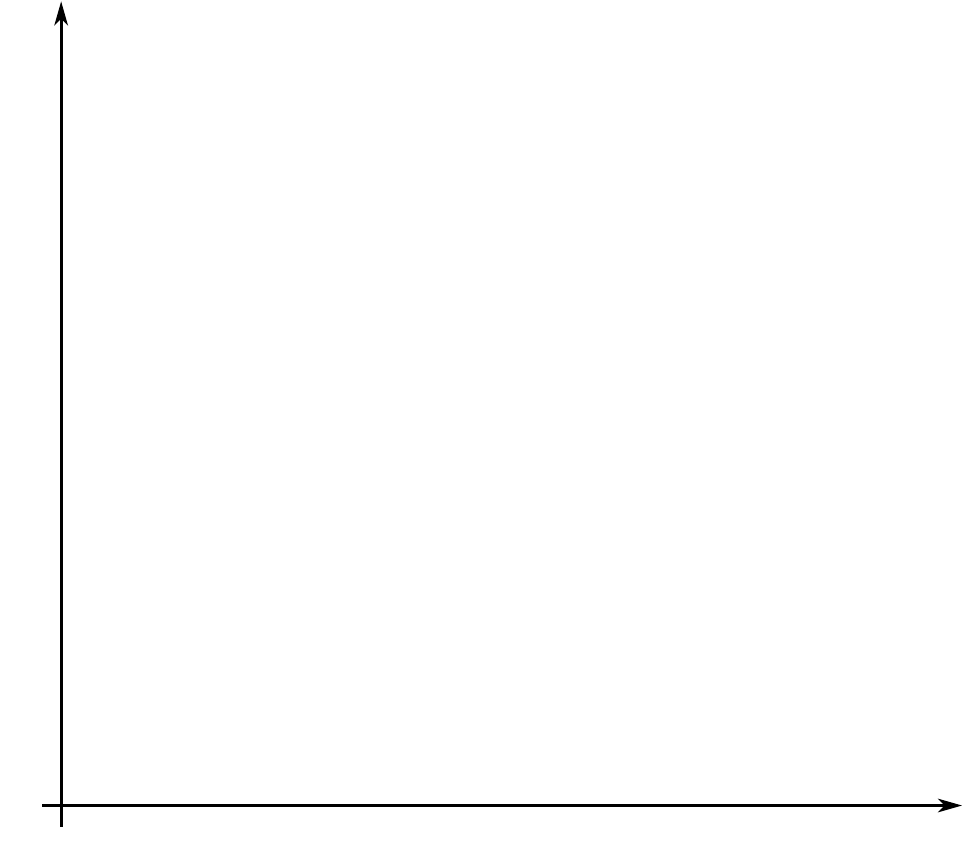}}%
    \put(-0.00148041,0.49164306){\color[rgb]{0,0,0}\makebox(0,0)[lb]{\smash{\tAx}}}%
    \put(0.45089941,0.00642915){\color[rgb]{0,0,0}\makebox(0,0)[lb]{\smash{\xAx}}}%
    \put(0.15539323,0.71600884){\color[rgb]{0,0,0}\makebox(0,0)[lb]{\smash{}}}%
    \put(0.51291918,0.13411713){\color[rgb]{0,0,0}\makebox(0,0)[lb]{\smash{\pT}}}%
    \put(0,0){\includegraphics[width=\unitlength,page=2]{F_1.pdf}}%
    \put(0.74458143,0.31287989){\color[rgb]{0,0,0}\makebox(0,0)[lb]{\smash{\nT}}}%
    \put(0,0){\includegraphics[width=\unitlength,page=3]{F_1.pdf}}%
    \put(0.53298443,0.84916905){\color[rgb]{0,0,0}\makebox(0,0)[lb]{\smash{\twoT}}}%
    \put(0,0){\includegraphics[width=\unitlength,page=4]{F_1.pdf}}%
  \end{picture}%
\endgroup%
	\end{center}
	\caption{Regions and characteristics for the absorbing state $\numLineages = 1$. The characteristics are parallel to the $\time$-axis.}
\end{subfigure}
\caption{\label{fig_1locus_regions} The different regions and characteristics of $\eff_\numLineages (\time, \treeDimA)$ \edit{(defined in equation~\eqref{eq_temp_cdf})} for different values of $\numLineages$. In (c), according to Lemma~\ref{lem_cdf}, $\eff_\numLineages (\time, \treeDimA)$ does not depend on $\time$ beyond the dashed line $\treeDimA = 2 \time$.}
\end{figure}

Since for $k=n$, the interior region is empty, we consider $k\neq n$ and introduce the family of characteristics
\begin{equation}\label{def_char}
\argChar \to \Big( \time_0 + \argChar, \treeDimA_0 + \stateFunction (\numLineages) \argChar \Big)^{\T}\quad \text{with} \quad \time_0 = \frac{\treeDimA_0}{\sampleSize}
\end{equation}
Taking the derivative of $\eff_\numLineages (\time,\treeDimA)$ along such a characteristic yields
\begin{equation}\label{eq_chain_rule}
	\begin{split}
		\frac{d}{d\argChar} \eff_\numLineages & \Big( \time_0 + \argChar, \treeDimA_0 + \stateFunction (\numLineages) \argChar \Big)\\
			& =  \bigg( \frac{d}{d\argChar} \Big[ \time_0 + \argChar \Big] \cdot \partial_\time \eff_\numLineages (\time,\treeDimA) + \frac{d}{d\argChar} \Big[ \treeDimA_0 + \stateFunction (\numLineages) \argChar \Big] \cdot \partial_\treeDimA \eff_\numLineages (\time,\treeDimA) \bigg)\Bigg|_{(\time,\treeDimA) = ( \frac{\treeDimA_0}{\sampleSize} + \argChar, \treeDimA_0 + \stateFunction (\numLineages) \argChar )}\\
			& = \bigg(  \partial_\time \eff_\numLineages (\time,\treeDimA) + \stateFunction (\numLineages) \partial_\treeDimA \eff_\numLineages (\time, \treeDimA)\bigg)\Bigg|_{(\time,\treeDimA) = ( \frac{\treeDimA_0}{\sampleSize} + \argChar, \treeDimA_0 + \stateFunction (\numLineages) \argChar )}\\
			& = \eff_\numLineages \Big( \time_0 + \argChar, \treeDimA_0 + \stateFunction (\numLineages) \argChar \Big) \ancestralRateMatrix_{\numLineages,\numLineages} (\time_0 + \argChar) + \eff_{\numLineages+1} \Big( \time_0 + \argChar, \treeDimA_0 + \stateFunction (\numLineages) \argChar \Big) \ancestralRateMatrix_{\numLineages+1,\numLineages} (\time_0 + \argChar).
	\end{split}
\end{equation}
Here we used the chain rule and the fact that the third line is equal to the left-hand side of equation~\eqref{eq_marginal_pde_comp}. Formally, the derivations~\eqref{eq_chain_rule} do not hold for all $\argChar$. It can be shown, however, that the equality holds for almost all $\tau$; we omit the technical details here for readability. Thus, for given $\treeDimA_0$, as a function of $\argChar$, the function $\tau \to \eff_\numLineages ( \time_0 + \argChar, \treeDimA_0 + \stateFunction (\numLineages) \argChar )$ solves the equation
\begin{equation}
	\frac{d}{d\argChar} \eff_\numLineages \Big( \time_0 + \argChar, \treeDimA_0 + \stateFunction (\numLineages) \argChar \Big) = - \rateODE^{(1)}_\numLineages (\argChar) \eff_\numLineages \Big( \time_0 + \argChar, \treeDimA_0 + \stateFunction (\numLineages) \argChar \Big) + \inhODE^{(1)}_\numLineages (\argChar),
\end{equation}
with
\begin{equation}
	\rateODE^{(1)}_\numLineages (\argChar) :=  - \ancestralRateMatrix_{\numLineages,\numLineages} (\time_0 + \argChar) = \frac{\numLineages (\numLineages - 1)}{2} \invPopSize (\time_0 + \argChar)
\end{equation}
and
\begin{equation}
	\begin{split}
		\inhODE^{(1)}_\numLineages (\argChar) := & \eff_{\numLineages + 1} \Big( \time_0 + \argChar, \treeDimA_0 + \stateFunction (\numLineages) \argChar \Big) \ancestralRateMatrix_{\numLineages + 1,\numLineages} (\time_0 + \argChar)\\
			= & \eff_{\numLineages + 1} \Big( \time_0 + \argChar, \treeDimA_0 + \stateFunction (\numLineages) \argChar \Big) \frac{(\numLineages+1) \numLineages}{2} \invPopSize (\time_0 + \argChar),
	\end{split}
\end{equation}
Since this is a non-homogeneous linear first-order ODE, the solution can be readily obtained as 
\begin{equation}\label{eq_sol_marginal_ode}
	\eff_\numLineages \Big( \time_0 + \argChar, \treeDimA_0 + \stateFunction (\numLineages) \argChar \Big) = e^{-\rateInt^{(1)}_\numLineages(\argChar)} \Bigg( \int_0^\argChar \inhODE^{(1)}_\numLineages (\alpha) e^{\rateInt^{(1)}_\numLineages(\alpha)}  d\alpha  + \eff_\numLineages \big( \time_0, \treeDimA_0  \big)  \Bigg),
\end{equation}
with
\begin{equation}\label{eq_marginal_rate_int}
	\rateInt^{(1)}_\numLineages(\argChar) := \int_0^{\argChar} \rateODE^{(1)}_\numLineages (\alpha) d\alpha = \frac{\numLineages (\numLineages - 1)}{2} \big( \cumInvPopSize(u) - \cumInvPopSize(\time_0)\big).
\end{equation}
The initial conditions for $\argChar = 0$ are given by the boundary values of the associated PDE as
\begin{equation}
	\eff_\numLineages ( \time_0, \treeDimA_0) = \P \big\{ \ancestralProcess(\time_0) = \numLineages \big\}.
\end{equation}

Now, to obtain the value of the function $\eff_\numLineages ( \time, \treeDimA)$, for given $\time$ and $\treeDimA$, one just needs to identify the right characteristic and the parameters $\treeDimA_0$ and $\argChar$ such that $(\time_0 + \argChar, \treeDimA_0 + \stateFunction (\numLineages) \argChar )^{\T} = ( \time, \treeDimA)^{\T}$. Since the characteristics are parallel, it can be uniquely identified. Using these values of $\treeDimA_0$ and $\argChar$ in the solution~\eqref{eq_sol_marginal_ode} yields $\eff_\numLineages ( \time, \treeDimA)$. However, we will not pursue this strategy to compute the requisite values of $\eff_\numLineages ( \time, \treeDimA)$. Instead, we present a numerical upstream scheme in Appendix~\ref{sec_num_alg_marginal} that can be used to compute $\eff_\numLineages ( \time, \treeDimA)$ efficiently on a suitable grid to ultimately obtain values for the CDF $\P \{ \totalTreeLength \leq \treeDimA\}$.



\subsection{Joint Distribution of the Total Tree Length}

In this section we present a method to compute the joint CDF of the total tree length
\begin{equation}
	\P \{ \totalTreeLength^\aLocus \leq \treeDimA , \totalTreeLength^\bLocus \leq \treeDimB \}
\end{equation}
at two loci $\aLocus$ and $\bLocus$ separated by a given recombination distance $\recoRate$. Again, we approach this problem by first computing the time-dependent joint CDF
\begin{equation}
	\eff_\state (\time, \treeDimA, \treeDimB) = \P \big\{ \ancestralRecoLimitedProcess(\time) = \state, \treeLength^\aLocus(\time) \leq \treeDimA, \treeLength^\bLocus(\time) \leq \treeDimB \big\}.
\end{equation}
We will follow closely along the lines of the method presented in Section~\ref{sec_cdf_marginal}, where we replace the ancestral process $\ancestralProcess$ by the ancestral process with limited recombination $\ancestralRecoLimitedProcess$, and compute the integrals~\eqref{def_total_length_a} and~\eqref{def_total_length_b}, to ultimately compute the joint CDF.

The analog to Lemma~\ref{lem_cdf} is as follows.
\begin{lemma}\label{lem_joint_cdf}
With definition~\eqref{eq_temp_joint_cdf}, the relation
\begin{equation}
	\begin{split}
		\P \big\{ \totalTreeLength^\aLocus \leq \treeDimA, \totalTreeLength^\bLocus \leq \treeDimB \big\} & = \P \big\{ \ancestralRecoLimitedProcess(\bar\time) \in \absStates, \treeLength^\aLocus(\bar\time) \leq \treeDimA, \treeLength^\bLocus(\bar\time) \leq \treeDimB \big\}\\
			& = \eff_{(1,0,0,0)} (\bar\time, \treeDimA, \treeDimB ) + \eff_{(1,0,0,1)} (\bar\time, \treeDimA, \treeDimB )
	\end{split}
\end{equation}
holds for $\treeDimA,\treeDimB \in \R_+$, $\bar\time \geq \max\{\treeDimA, \treeDimB\}/2$, and $\absStates = \big\{ (1,0,0,0), (1,0,0,1)\big\}$, the absorbing states of $\ancestralRecoLimitedProcess$.
\end{lemma}

\begin{proof}
The proof is similar to the proof of lemma~\ref{lem_cdf}. With
\begin{equation}
	\ancestralRecoLimitedProcess(\time) = \big( \numLineagesRVLimited_{\aLocus\bLocus} (\time), \numLineagesRVLimited_{\aLocus} (\time), \numLineagesRVLimited_{\bLocus} (\time), \numRecoRV(\time) \big),
\end{equation}
note that
\begin{equation}
	2\tMRCA^\aLocus \leq \int_0^{\tMRCA^\aLocus} \1_{\{\numLineagesRVLimited_{\aLocus\bLocus}(s) + \numLineagesRVLimited_{\aLocus}(s) > 1\}} \big( \numLineagesRVLimited_{\aLocus\bLocus}(s) + \numLineagesRVLimited_{\aLocus}(s) \big) ds = \totalTreeLength^\aLocus,
\end{equation}
and similarly $	2\tMRCA^\bLocus \leq \totalTreeLength^\bLocus$. Thus, on the event $\big\{ \totalTreeLength^\aLocus \leq \treeDimA, \totalTreeLength^\bLocus \leq \treeDimB \big\}$, the relations $\tMRCA^\aLocus \leq \max\{\treeDimA, \treeDimB\}/2 \leq \bar\time$ and $\tMRCA^\bLocus \leq \bar\time$ hold. This implies $\numLineagesRVLimited_{\aLocus\bLocus}(\bar\time) + \numLineagesRVLimited_{\aLocus}(\bar\time) = 1$ and $\numLineagesRVLimited_{\aLocus\bLocus}(\bar\time) + \numLineagesRVLimited_{\aLocus}(\bar\time) = 1$, which in turn implies $\ancestralRecoLimitedProcess(\bar\time) \in \absStates = \big\{ (1,0,0,0), (1,0,0,1)\big\}$, since these two states are the only admissible states that can satisfy these conditions. Incidentally, these are also the absorbing states of $\ancestralRecoLimitedProcess$. Thus,
\begin{equation}
	\big\{ \totalTreeLength^\aLocus \leq \treeDimA, \totalTreeLength^\bLocus \leq \treeDimB \big\} = \big\{ \ancestralRecoLimitedProcess(\bar\time) \in \absStates, \totalTreeLength^\aLocus \leq \treeDimA, \totalTreeLength^\bLocus \leq \treeDimB \big\}
\end{equation}
holds. Furthermore, on the event $\big\{ \ancestralRecoLimitedProcess(\bar\time) \in \absStates \big\}$, $\tMRCA^\aLocus \leq \bar\time$ and $\tMRCA^\bLocus \leq \bar\time$ hold, which imply $\treeLength^\aLocus(\bar\time) = \totalTreeLength^\aLocus$ and $\treeLength^\bLocus(\bar\time) = \totalTreeLength^\bLocus$. This in turn implies
\begin{equation}
	\big\{ \ancestralRecoLimitedProcess(\bar\time) \in \absStates, \totalTreeLength^\aLocus \leq \treeDimA, \totalTreeLength^\bLocus \leq \treeDimB \big\} = \big\{ \ancestralRecoLimitedProcess(\bar\time) \in \absStates, \treeLength^\aLocus(\bar\time) \leq \treeDimA, \treeLength^\bLocus(\bar\time) \leq \treeDimB \big\}.
\end{equation}
Finally, note that
\begin{equation}
	\begin{split}
		\big\{ \ancestralRecoLimitedProcess(\bar\time) = (1,0,0,1) \big\} \cap \big\{ \ancestralRecoLimitedProcess(\bar\time) = (1,0,0,0) \big\} = \emptyset,
	\end{split}
\end{equation}
which proves the statement of the lemma.
\end{proof}

Again, Lemma~\ref{lem_joint_cdf} shows that the joint CDF of $\totalTreeLength^\aLocus$ and $\totalTreeLength^\bLocus$ can be computed from the time-dependent CDFs $\eff_{(1,0,0,0)} (\time, \treeDimA, \treeDimB)$, and $\eff_{(1,0,0,1)} (\time, \treeDimA, \treeDimB)$.
\edit{In order to derive a system of PDEs like~\eqref{eq_marginal_pde} for the time-dependent joint CDF of the tree length at two loci, we again apply Proposition~\ref{prop_app}, for dimension $d=2$.}
To this end, define the functions
\begin{equation}
	\stateFunction^\aLocus (\numLineages_{\aLocus\bLocus}, \numLineages_{\aLocus}, \numLineages_{\bLocus}, \numReco) := \1_{\{\numLineages_{\aLocus\bLocus} + \numLineages_{\aLocus} > 1\}} (\numLineages_{\aLocus\bLocus} + \numLineages_{\aLocus})
\end{equation}
and
\begin{equation}
	\stateFunction^\bLocus (\numLineages_{\aLocus\bLocus}, \numLineages_{\aLocus}, \numLineages_{\bLocus}, \numReco) := \1_{\{\numLineages_{\aLocus\bLocus} + \numLineages_{\bLocus} > 1\}} (\numLineages_{\aLocus\bLocus} + \numLineages_{\bLocus})
\end{equation}
that yield for $(\numLineages_{\aLocus\bLocus}, \numLineages_{\aLocus}, \numLineages_{\bLocus}, \numReco) \in \recoLimitedStates$ the number of lineages ancestral to locus $\aLocus$ and $\bLocus$, respectively, and define
\begin{equation}
	\stateFunctionMatrix^\aLocus := \diag \Big( \big( \stateFunction^\aLocus (\state) \big)_{\state \in \recoLimitedStates} \Big) \qquad \text{and} \qquad \stateFunctionMatrix^\bLocus := \diag \Big( \big( \stateFunction^\bLocus (\state) \big)_{\state \in \recoLimitedStates} \Big)\,.
\end{equation}
We then have the following \firstRevision{corollary}.

\begin{corollary}
The time-dependent joint CDF of the tree lengths
\begin{equation}
	\vecEff (\time,\treeDimA, \treeDimB) = \big( \vecEff_\state (\time,\treeDimA, \treeDimB) \big)_{\state \in \recoLimitedStates}
\end{equation}
can be obtained for all points in $U = \Big\{ (t,x,y): \, 0 < \treeDimA < \sampleSize \time, 0 < \treeDimB < \sampleSize \time, \, t > 0 \Big\} $ as the strong solution of
\begin{equation}\label{eq_joint_pde}
		\partial_\time \vecEff (\time,\treeDimA, \treeDimB) + \edit{\partial_\treeDimA \vecEff (\time,\treeDimA, \treeDimB) \stateFunctionMatrix^\aLocus} + \edit{\partial_\treeDimB \vecEff (\time,\treeDimA, \treeDimB) \stateFunctionMatrix^\bLocus} = \vecEff (\time,\treeDimA, \treeDimB) \ancestralRecoLimitedMatrix(\time),
\end{equation}
with boundary conditions
\begin{equation}
	\vecEff (\time,\treeDimA, \treeDimB) = \begin{cases}
			\bigg( \P \big\{ \ancestralRecoLimitedProcess(\time) = \state, \treeLength^\bLocus(\time) \leq \treeDimB \big\} \bigg)_{\state \in \recoLimitedStates} \cdot \1_{\{\stateFunction^\aLocus (\state) \neq \sampleSize\}},		& \text{if $\treeDimA = \sampleSize \time$},\\
			\bigg( \P \big\{ \ancestralRecoLimitedProcess(\time) = \state, \treeLength^\aLocus(\time) \leq \treeDimA \big\} \bigg)_{\state \in \recoLimitedStates} \cdot \1_{\{\stateFunction^\bLocus (\state) \neq \sampleSize\}},		& \text{if $\treeDimB = \sampleSize \time$},\\
			0,			&\text{if $\treeDimA = 0$ or $\treeDimB = 0$},
		\end{cases}
\end{equation}
for $(x,y,t) \in \del U$ and $\ancestralRecoLimitedMatrix(\time)$ as defined in~\eqref{def_limited_reco_matrix}.
\end{corollary}

\firstRevision{
\begin{proof}
Define the function ${\bf \stateFunction} : \recoLimitedStates \to \R^2$ as
\begin{equation}
	{\bf \stateFunction} (\numLineages_{\aLocus\bLocus}, \numLineages_{\aLocus}, \numLineages_{\bLocus}, \numReco) := \big( \1_{\{\numLineages_{\aLocus\bLocus} + \numLineages_{\aLocus} > 1\}} (\numLineages_{\aLocus\bLocus} + \numLineages_{\aLocus}), \1_{\{\numLineages_{\aLocus\bLocus} + \numLineages_{\bLocus} > 1\}} (\numLineages_{\aLocus\bLocus} + \numLineages_{\bLocus}) \big).
\end{equation}
This function and the generator $\ancestralRecoLimitedMatrix(\time)$ satisfy the requirements of Proposition~\ref{prop_app}, and thus, the statement of the corollary follows from Proposition~\ref{prop_app} and Remark~\ref{rem_zero}.
\end{proof}
}

\begin{remark}
Note that due to symmetry of $\ancestralRecoLimitedProcess$,
\begin{equation}
	\P \big\{ \ancestralRecoLimitedProcess(\time) = \state, \treeLength^\aLocus(\time) \leq \treeDimA \big\} = \P \big\{ \ancestralRecoLimitedProcess(\time) = \state, \treeLength^\bLocus(\time) \leq \treeDimA \big\} 
\end{equation}
holds.
\end{remark}

\firstRevision{\edit{The process} $\big(\ancestralProcess(\time), \treeLength^\aLocus(\time), \treeLength^\bLocus(\time) \big)_{\time \in \R_+}$ is a piecewise-deterministic Markov process \edit{as well} (see Remark~\ref{rem_generator_d}), where $\ancestralRecoLimitedMatrix(\time)$ captures the stochastic dynamics, and $\partial_\treeDimA$ and $\partial_\treeDimB$ the deterministic dynamics.}
%
The numerical scheme to compute the time-dependent joint CDF is again an upstream scheme based on the method of characteristics and follows essentially along the lines of the scheme presented for the marginal case. The relation $\prec$ defined in~\eqref{def_rel} implies a partial ordering on the state space $\recoLimitedStates$, and the matrix $\ancestralRecoLimitedMatrix(t)$ is triangular with respect to this ordering. Thus, again, the values of $\eff_\state$ only depend on $\eff_{\state'}$ with $\state \prec \state'$, and they can be computed for each $\state$ separately. For given $\state \in \recoLimitedStates$,
\begin{equation}\label{eq_joint_boundaries}
	\begin{split}
	\eff_\state (\time, \treeDimA, \treeDimB) & = \P \big\{ \ancestralRecoLimitedProcess(\time) = \state, \treeLength^\aLocus(\time) \leq \treeDimA, \treeLength^\bLocus(\time) \leq \treeDimB \big\}\\
		&  = \begin{cases}
							0,															& \text{if $\treeDimA < \stateFunction^{\aLocus} (\state) \cdot \time$ or $\treeDimB < \stateFunction^{\bLocus} (\state) \cdot \time$},\\
							\text{solution to~\eqref{eq_joint_pde}},								& \text{if $\stateFunction^{\aLocus} (\state) \cdot \time \leq \treeDimA  < \sampleSize \cdot \time$ and $\stateFunction^{\bLocus} (\state) \cdot \time \leq \treeDimB  < \sampleSize \cdot \time$},\\
							\P \big\{ \ancestralRecoLimitedProcess(\time) = \state, \treeLength^\aLocus(\time) \leq \treeDimA \big\},								& \text{if $\stateFunction^{\aLocus} (\state) \cdot \time \leq \treeDimA  < \sampleSize \cdot \time$ and $\sampleSize \cdot \time \leq \treeDimB$},\\
							\P \big\{ \ancestralRecoLimitedProcess(\time) = \state,  \treeLength^\bLocus(\time) \leq \treeDimB \big\},								& \text{if $\sampleSize \cdot \time \leq \treeDimA $ and $\stateFunction^{\bLocus} (\state) \cdot \time \leq \treeDimB  < \sampleSize \cdot \time$},\\
							\P \big\{ \ancestralRecoLimitedProcess(\time) = \state \},	& \text{if $\sampleSize \cdot \time \leq \treeDimA$ and $\sampleSize \cdot \time \leq \treeDimB$}
						\end{cases}
	\end{split}
\end{equation}
holds. Figure~\ref{fig_region3d} shows the different regions of $\eff_\state (\time, \treeDimA, \treeDimB)$ for a fixed $\time$. Moreover, for each $\state \in \recoLimitedStates$, the PDE that has to be satisfied in the region $\stateFunction^{\aLocus} (\state) \cdot \time \leq \treeDimA  < \sampleSize \cdot \time$ and $\stateFunction^{\bLocus} (\state) \cdot \time \leq \treeDimB  < \sampleSize \cdot \time$ can be re-written as
\begin{equation}\label{eq_joint_pde_comp}
	\partial_\time \eff_\state (\time,\treeDimA, \treeDimB) + \big(\stateFunction^\aLocus (\state), \stateFunction^\bLocus (\state) \big) \nabla \eff_\state (\time,\treeDimA, \treeDimB)  = \eff_\state (\time,\treeDimA, \treeDimB) \ancestralRecoLimitedMatrix_{\state,\state}(\time) + \sum_{\state \prec \state'} \eff_{\state'} (\time,\treeDimA, \treeDimB) \ancestralRecoLimitedMatrix_{\state',\state}(\time),
\end{equation}
where $\nabla f = (\partial_x f, \partial_y f)^{\T}$.
Again, taking the derivative of $\eff_\state (\time,\treeDimA, \treeDimB)$ along the characteristics
\begin{equation}
	\tau \to \Big( \time_0 + \argChar, \vecEks_0  + \argChar \vecVee (\state) \Big)^{\T},
\end{equation}
with $\time_0 := \frac{1}{\sampleSize} \max\{\treeDimA_0, \treeDimB_0\}$, $\vecEks_0 := (\treeDimA_0, \treeDimB_0)$, and $\vecVee (\state) :=  \big(\stateFunction^\aLocus (\state), \stateFunction^\bLocus (\state) \big)$, yields the right-hand side of equation~\eqref{eq_joint_pde_comp}. Thus, $\eff_\state(\cdot,\cdot,\cdot)$ satisfies the ODE
\begin{equation}
	\frac{d}{d\argChar} \eff_\state \Big( \time_0 + \argChar, \vecEks_0 + \argChar \vecVee (\state) \Big) = - \rateODE^{(2)}_\state (\argChar) \eff_\state \Big( \time_0 + \argChar, \vecEks_0 + \argChar \vecVee (\state) \Big) + \inhODE^{(2)}_\state (\argChar),
\end{equation}
with
\begin{equation}
	\rateODE^{(2)}_\state (\argChar) =  - \ancestralRecoLimitedMatrix_{\state,\state}(\time_0 + \argChar)
\end{equation}
and
\begin{equation}
	\inhODE^{(2)}_\state (\argChar) = \sum_{\state \prec \state'} \eff_{\state'} \Big( \time_0 + \argChar, \vecEks_0 + \argChar \vecVee (\state) \Big) \ancestralRecoLimitedMatrix_{\state',\state} (\time_0 + \argChar).
\end{equation}

The characteristics for $\eff_\state (\time,\treeDimA, \treeDimB)$ are depicted in Figure~\ref{fig_region3d}. Like in the marginal case, this is a non-homogeneous linear first-order ODE and can be readily solved. The solution involves integrating $\rateODE^{(2)}_\state (\argChar)$, which leads to
\begin{equation}\label{eq_sol_joint_ode}
	\eff_\state \Big( \time_0 + \argChar, \treeDimA_0 + \vecVee (\state) \argChar \Big) = e^{-\rateInt^{(2)}_\numLineages(\argChar)} \Bigg( \int_0^\argChar \inhODE^{(2)}_\state (\alpha) e^{\rateInt^{(2)}_\numLineages(\alpha)} d\alpha+ \eff_\state \big( \time_0, \treeDimA_0 \big) \Bigg),
\end{equation}
with
\begin{equation}\label{eq_rate_int_joint}
	\rateInt^{(2)}_\state (\argChar) = \int_0^{\argChar} \rateODE^{(2)}_\state (\alpha) d\alpha = - \ancestralRecoOnlyLimitedMatrix_{\state,\state} (u - \time_0) - \ancestralRecoCoalLimitedMatrix_{\state,\state}  \big( \cumInvPopSize(u) - \cumInvPopSize(\time_0)\big).
\end{equation}
We provide the details of our numerical upstream scheme to efficiently and accurately compute solutions to equation~\eqref{eq_sol_joint_ode} in Appendix~\ref{sec_algo_joint}.

\begin{figure}
\newcommand{\xA}{$\treeDimA$}
\newcommand{\xB}{$\treeDimB$}
\newcommand{\nt}{$\sampleSize \cdot \time$}
\newcommand{\KA}{$\stateFunction^{\aLocus}(\state) \cdot \time$}
\newcommand{\KB}{$\stateFunction^{\bLocus}(\state) \cdot \time$}
\newcommand{\zero}{$0$}
\newcommand{\PALa}{$\P \big\{ \ancestralRecoLimitedProcess(\time) = \state, \treeLength^\aLocus(\time) \leq \treeDimA\big\}$}
\newcommand{\PALb}{$\P \big\{ \ancestralRecoLimitedProcess(\time) = \state, \treeLength^\bLocus(\time) \leq \treeDimB \big\}$}
\newcommand{\PA}{$\P \big\{ \ancestralRecoLimitedProcess(\time) = \state \big\}$}
\def\svgwidth{.7\textwidth}
\begin{center}
\begingroup%
  \makeatletter%
  \providecommand\color[2][]{%
    \errmessage{(Inkscape) Color is used for the text in Inkscape, but the package 'color.sty' is not loaded}%
    \renewcommand\color[2][]{}%
  }%
  \providecommand\transparent[1]{%
    \errmessage{(Inkscape) Transparency is used (non-zero) for the text in Inkscape, but the package 'transparent.sty' is not loaded}%
    \renewcommand\transparent[1]{}%
  }%
  \providecommand\rotatebox[2]{#2}%
  \ifx\svgwidth\undefined%
    \setlength{\unitlength}{251.18775135bp}%
    \ifx\svgscale\undefined%
      \relax%
    \else%
      \setlength{\unitlength}{\unitlength * \real{\svgscale}}%
    \fi%
  \else%
    \setlength{\unitlength}{\svgwidth}%
  \fi%
  \global\let\svgwidth\undefined%
  \global\let\svgscale\undefined%
  \makeatother%
  \begin{picture}(1,0.62860379)%
    \put(0,0){\includegraphics[width=\unitlength,page=1]{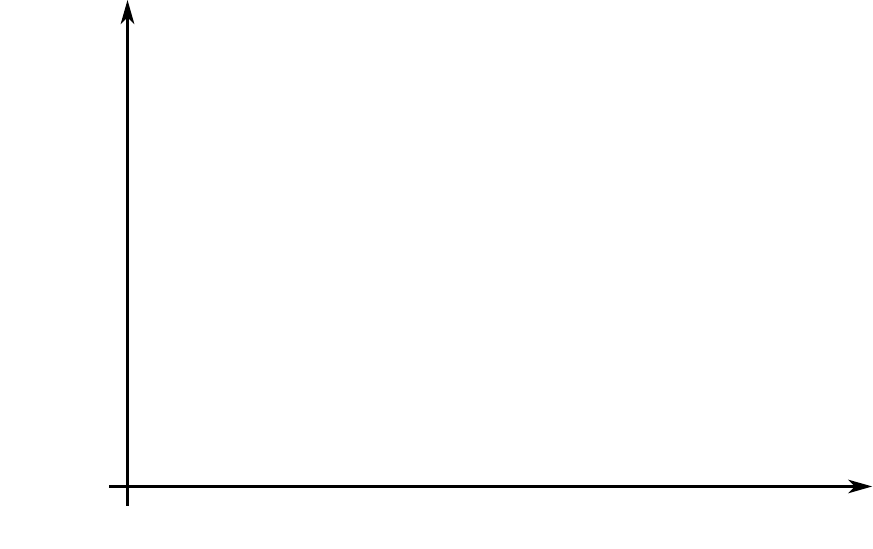}}%
    \put(0.10261086,0.35049225){\color[rgb]{0,0,0}\makebox(0,0)[lb]{\smash{\xB}}}%
    \put(0.78869512,0.00443209){\color[rgb]{0,0,0}\makebox(0,0)[lb]{\smash{\xA}}}%
    \put(0.2474732,0.59796543){\color[rgb]{0,0,0}\makebox(0,0)[lb]{\smash{}}}%
    \put(0.60962906,0.00443209){\color[rgb]{0,0,0}\makebox(0,0)[lb]{\smash{\nt}}}%
    \put(0.05633538,0.51547443){\color[rgb]{0,0,0}\makebox(0,0)[lb]{\smash{\nt}}}%
    \put(0.01408384,0.23379764){\color[rgb]{0,0,0}\makebox(0,0)[lb]{\smash{\KB}}}%
    \put(0.20320972,0.00443209){\color[rgb]{0,0,0}\makebox(0,0)[lb]{\smash{\KA}}}%
    \put(0.18711388,0.55168997){\color[rgb]{0,0,0}\makebox(0,0)[lb]{\smash{\zero}}}%
    \put(0.18711388,0.35652822){\color[rgb]{0,0,0}\makebox(0,0)[lb]{\smash{\zero}}}%
    \put(0.18711388,0.12917459){\color[rgb]{0,0,0}\makebox(0,0)[lb]{\smash{\zero}}}%
    \put(0.41245539,0.12917459){\color[rgb]{0,0,0}\makebox(0,0)[lb]{\smash{\zero}}}%
    \put(0.82491075,0.12917459){\color[rgb]{0,0,0}\makebox(0,0)[lb]{\smash{\zero}}}%
    \put(0.58711834,0.6472095){\color[rgb]{0,0,0}\makebox(0,0)[lb]{\smash{}}}%
    \put(0.26155712,0.55168997){\color[rgb]{0,0,0}\makebox(0,0)[lb]{\smash{\PALa}}}%
    \put(0.72431184,0.55168997){\color[rgb]{0,0,0}\makebox(0,0)[lb]{\smash{\PA}}}%
    \put(0.6538926,0.35652822){\color[rgb]{0,0,0}\makebox(0,0)[lb]{\smash{\PALb}}}%
    \put(0,0){\includegraphics[width=\unitlength,page=2]{regions3D.pdf}}%
  \end{picture}%
\endgroup%
\end{center}
\caption{The different regions and (projected) characteristics of $\eff_\state (\time, \treeDimA, \treeDimB)$ \edit{(defined in equation~\eqref{eq_temp_joint_cdf})} for an intermediate state $\state \in \recoLimitedStates$ at a given time $\time$. The characteristics also extend in the $\time$-direction at unit speed. Note that for the states $\state$ with $\stateFunction^\aLocus (\state) = \sampleSize$ or $\stateFunction^\bLocus (\state) = \sampleSize$ the interior region is empty.}
\label{fig_region3d}
\end{figure}

\section{Empirical evaluation}
\label{sec_empirical}

In this section, we demonstrate that the numerical algorithms presented in Section~\ref{sec_num_alg_marginal} and~\ref{sec_algo_joint} can be used to accurately and efficiently compute the time-dependent marginal CDF~\eqref{eq_temp_cdf} and joint CDF~\eqref{eq_temp_joint_cdf}, as well as the regular marginal CDF~\eqref{eq_marginal_cdf} and joint CDF~\eqref{eq_joint_cdf}, for different population size histories and different recombination rates. Furthermore, we show how our method can be used to study properties of the marginal and joint distributions, and compute their moments. We implemented the numerical algorithms in \textsc{Matlab}, and the code is available upon request.

For ease of exposition, we use a sample size of $\sampleSize = 10$ in the remainder of this paper, \firstRevision{unless mentioned otherwise}.
We mainly focus on three population size histories, depicted in Figure~\ref{fig_pop_sizes}. \firstRevision{The first is a history} of constant size $1$, \firstRevision{and we refer to the corresponding rate function as $\invPopSize_c$.}
\firstRevision{Second, we consider} a history with an ancient bottleneck, followed by exponential growth up to the present. Specifically, for $\time > 0.15$, the relative population size is set to $2$, and for $0.025 < \time < 0.15$, it is set to $0.25$. Then, the population grows exponentially from size $0.25$ at $\time = 0.025$ up to $\time = 0$ (the present), at an exponential rate of $g$. We refer to this population size history by \firstRevision{$\invPopSize_e$}, and if not mentioned otherwise, the growth rate is set to $g = 200$. This size history is a rough sketch of the human population size history, with an out-of-Africa bottleneck, followed by recent exponential growth at a rate of $1\%$ per generation. In addition, we consider a pure bottleneck, where the relative ancestral size is 2 until time $\time = 0.05$, and $N_\text{B}$ from $\time = 0.05$ until the present. We refer to this size history by \firstRevision{$\invPopSize_b$}, and if not otherwise mentioned, we set $N_\text{B} = 0.2$.

\begin{figure}
\begin{center}
\includegraphics[width=.5\textwidth]{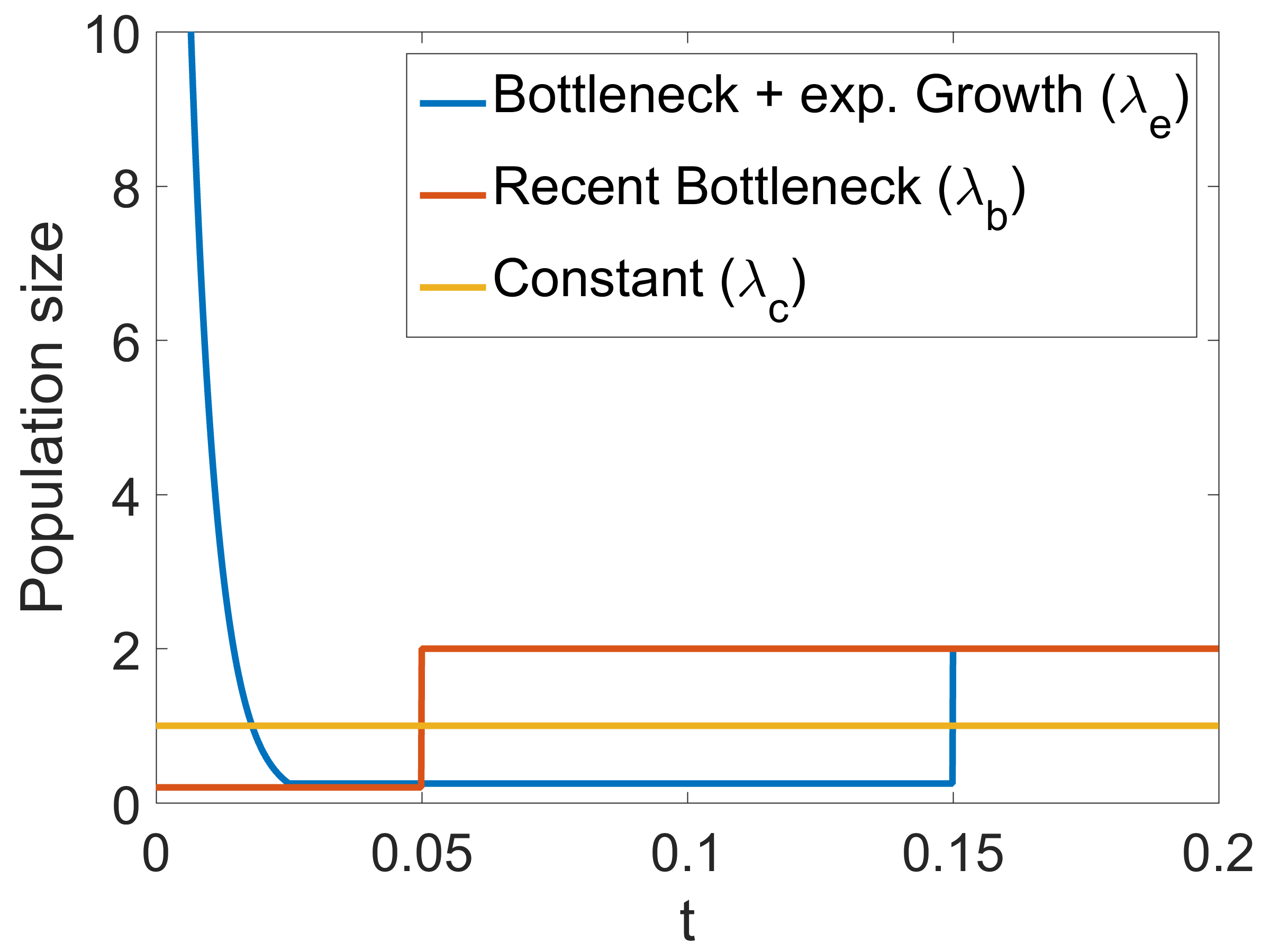}
\caption{The three population size histories we will mainly consider in this paper: \firstRevision{A constant population size ($\invPopSize_c$), an ancient bottleneck followed by exponential growth ($\invPopSize_e$), and a recent bottleneck ($\invPopSize_b$).}}
\label{fig_pop_sizes}
\end{center}
\end{figure}

\subsection{Accuracy}

In this section we demonstrate that the numerical algorithms presented in this paper can be used to compute the requisite CDFs accurately. Naturally, the accuracy will depend on the exact choice of the grid for the numerical algorithm. We will present the results for a particular grid here, and discuss the issues for choosing an adequate grid in Section~\ref{sec_discussion}. We set $\sampleSize = 5$ and compute the time-dependent marginal CDF
\begin{equation}
	\P \big\{ \ancestralProcess(\time) = \numLineages, \treeLength(\time) \leq \treeDimA \big\}
\end{equation}
for $\numLineages = 5$, $3$, and the absorbing state $1$, and show the respective surfaces as functions of $\time$ and $\treeDimA$ in Figure~\ref{fig_heatmaps_cdf_marginal}. Here we used the population size history with exponential growth \firstRevision{$\invPopSize_e$}. These surfaces exhibit the properties sketched in Figure~\ref{fig_1locus_regions}, and the different regions can be observed. Below the line $\treeDimA = \sampleSize \time$, the functions are independent of $\treeDimA$. Furthermore, the functions are zero above the line $\time = \numLineages \time$, except for $\numLineages = 1$, where the function is independent of $\time$ above the line $\treeDimA = 2 \time$.
 
\begin{figure}
\begin{subfigure}[b]{0.33\textwidth}
	\begin{center}
	\includegraphics[width=\textwidth]{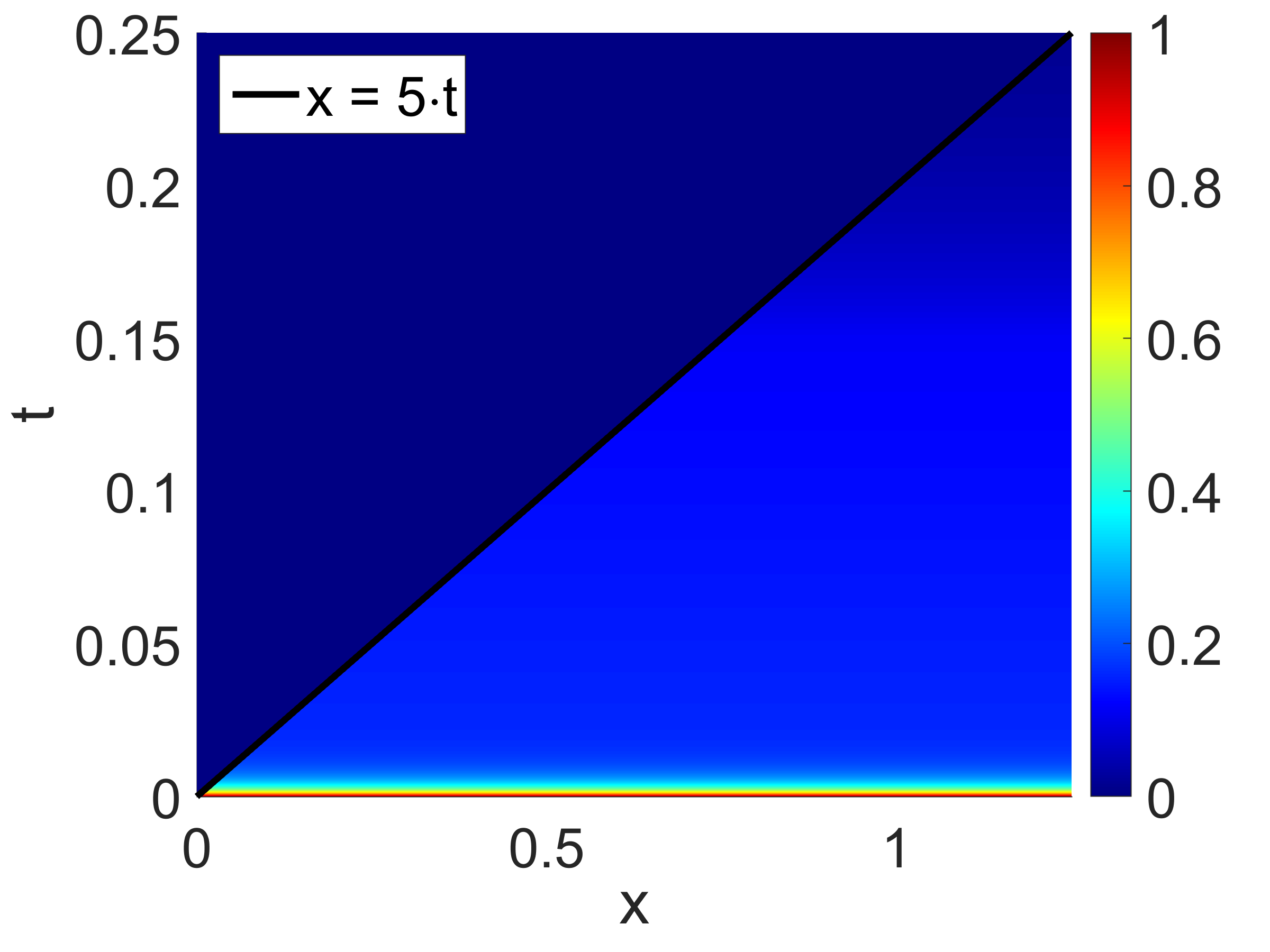}
	\end{center}
	\caption{$\P \big\{ \ancestralProcess(\time) = 5, \treeLength(\time) \leq \treeDimA \big\}$}
\end{subfigure}
\begin{subfigure}[b]{0.33\textwidth}
	\begin{center}
	\includegraphics[width=\textwidth]{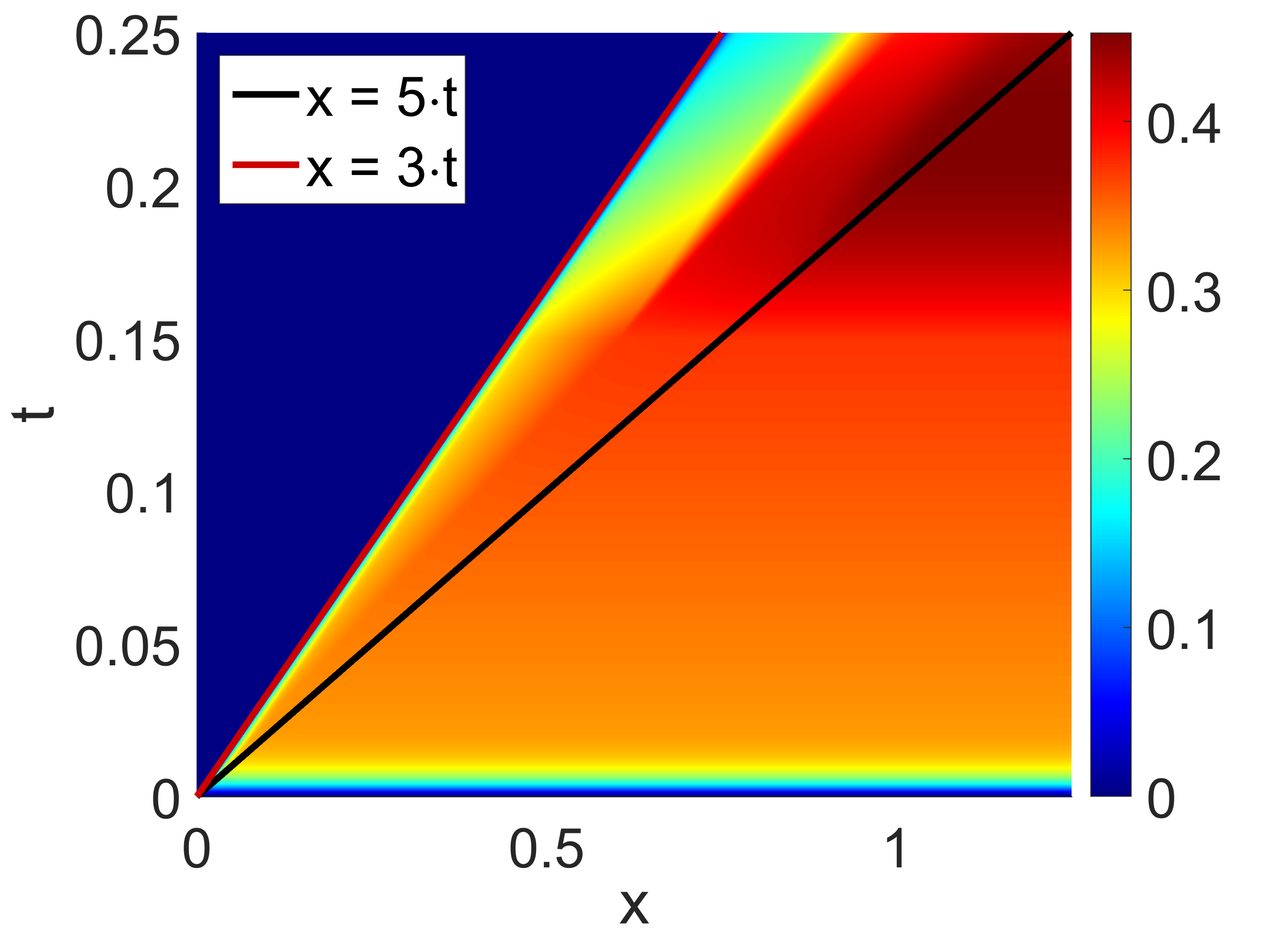}
	\end{center}
	\caption{$\P \big\{ \ancestralProcess(\time) = 3, \treeLength(\time) \leq \treeDimA \big\}$}
\end{subfigure}
\begin{subfigure}[b]{0.33\textwidth}
	\begin{center}
	\includegraphics[width=\textwidth]{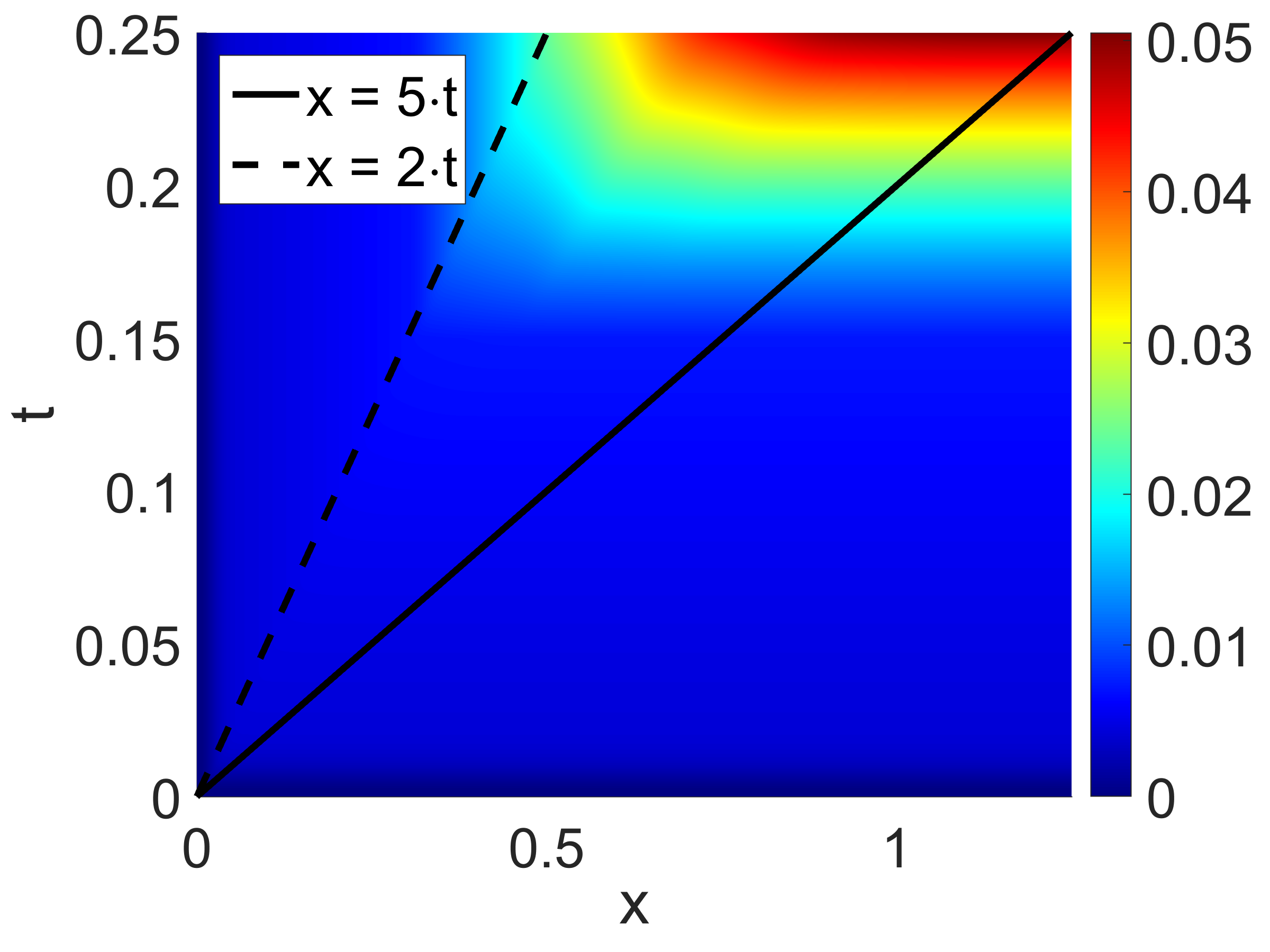}
	\end{center}
	\caption{$\P \big\{ \ancestralProcess(\time) = 1, \treeLength(\time) \leq \treeDimA \big\}$}
\end{subfigure}
\caption{Heatmaps of $\P \big\{ \ancestralProcess(\time) = \numLineages, \treeLength(\time) \leq \treeDimA \big\}$ \edit{(defined in equation~\eqref{eq_temp_cdf})} as a function of $\time$ and $\treeDimA$, for different $\numLineages$, computed using our numerical algorithm.}
\label{fig_heatmaps_cdf_marginal}
\end{figure}

\firstRevision{As shown in Section~\ref{sec_cdf}, the} marginal \firstRevision{CDF} of the total tree length
\begin{equation}
	\P \{ \totalTreeLength \leq \treeDimA\}
\end{equation}
and \firstRevision{the joint CDF}
\begin{equation}
	\P \{ \totalTreeLength^\aLocus \leq \treeDimA , \totalTreeLength^\bLocus \leq \treeDimB \}
\end{equation}
can be computed from the respective time-dependent CDFs.
To demonstrate the accuracy of our numerical algorithm, we compared the numerical values from the algorithm to simulations \firstRevision{under the respective ancestral processes $\ancestralProcess$ and $\ancestralRecoLimitedProcess$.}
\firstRevision{To this end, we simulated a certain number $N$ of trajectories from these processes,} and estimated the respective probabilities.
Figure~\ref{fig_marginal_cdf} shows the marginal CDFs for $\sampleSize = 10$ under exponential growth \firstRevision{($\invPopSize_e$)} and the bottleneck scenario \firstRevision{($\invPopSize_b$)}. The simulations can \firstRevision{also} be used to bound the difference $d(P_\text{pde}, P_\text{T})$ between the values computed using the numerical scheme $P_\text{pde}$ and the true value $P_\text{T}$. \firstRevision{These bounds are} indicated in Figure~\ref{fig_marginal_cdf} for different values of $N$ and decrease as $N$ gets larger, as expected.
\firstRevision{For the joint CDF, we present the numerical values for different $\treeDimA$ and $\treeDimB$, and compare them to the respective estimates from the simulations, including the confidence bounds for these estimates.}
We set $\sampleSize = 10$, and used $\recoRate = 0.001$. The values for the model with exponential growth \firstRevision{($\invPopSize_e$)} are shown in Table~\ref{tab_PLLlambda1}, and for the bottleneck scenario \firstRevision{($\invPopSize_b$)} in Table~\ref{tab_PLLlambda2}.
\firstRevision{The} values computed using the numeric algorithm always fall into the confidence bounds, \firstRevision{demonstrating that our algorithm computes the respective values accurately}. In these tables, it becomes particularly apparent that in order to guarantee a high accuracy using simulations, a very large number of trajectories should be simulated, which is time-consuming. Our numerical scheme yields a high accuracy, and does not suffer from these issues. 

\begin{figure}
\begin{subfigure}[b]{0.5\textwidth}
	\begin{center}
	\includegraphics[width=\textwidth]{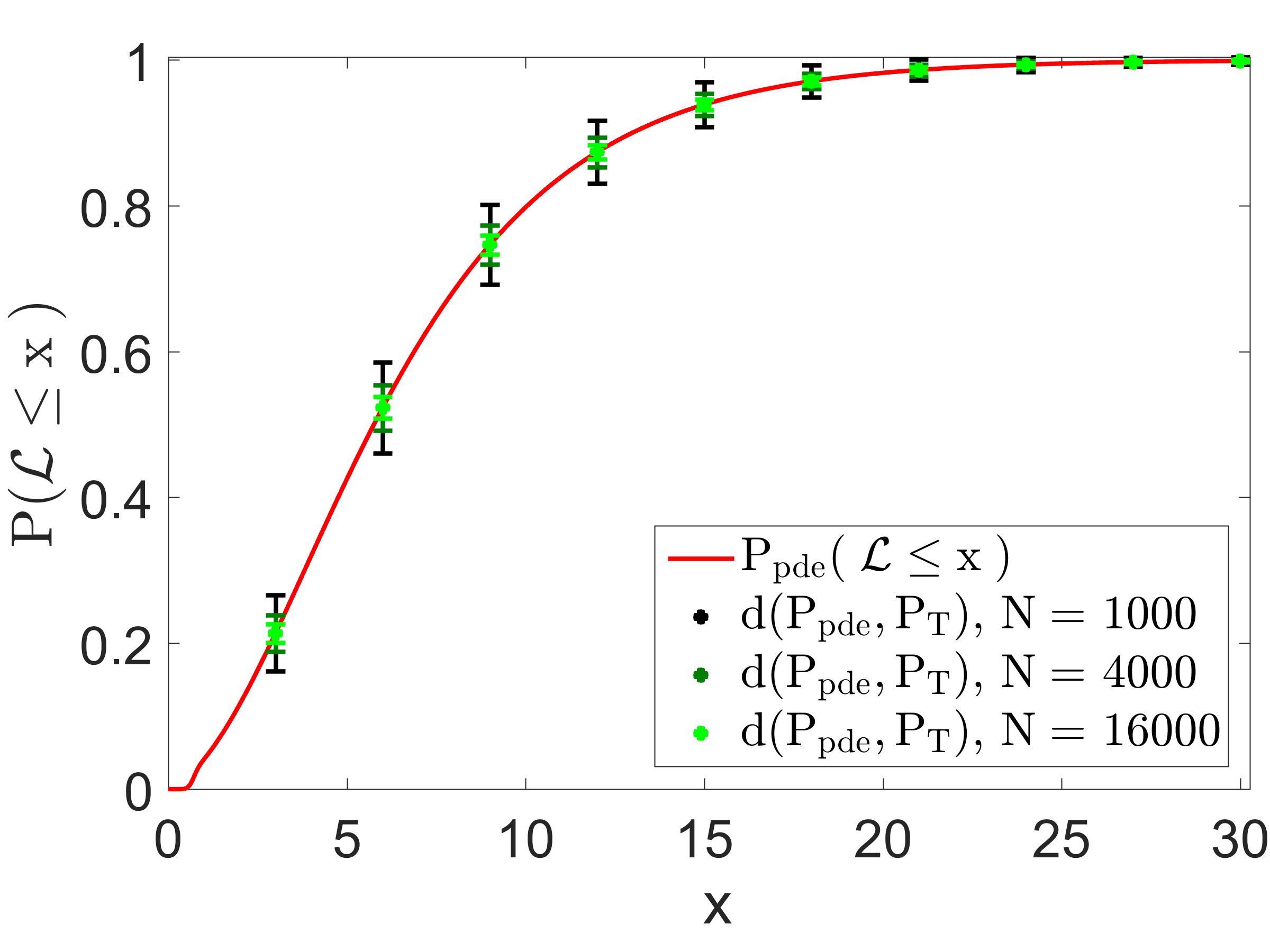}
	\end{center}
	\caption{$\P \{ \totalTreeLength \leq \treeDimA\}$ under exponential population growth \firstRevision{($\invPopSize_e$)}.}
\end{subfigure}
\begin{subfigure}[b]{0.5\textwidth}
	\begin{center} 
	\includegraphics[width=\textwidth]{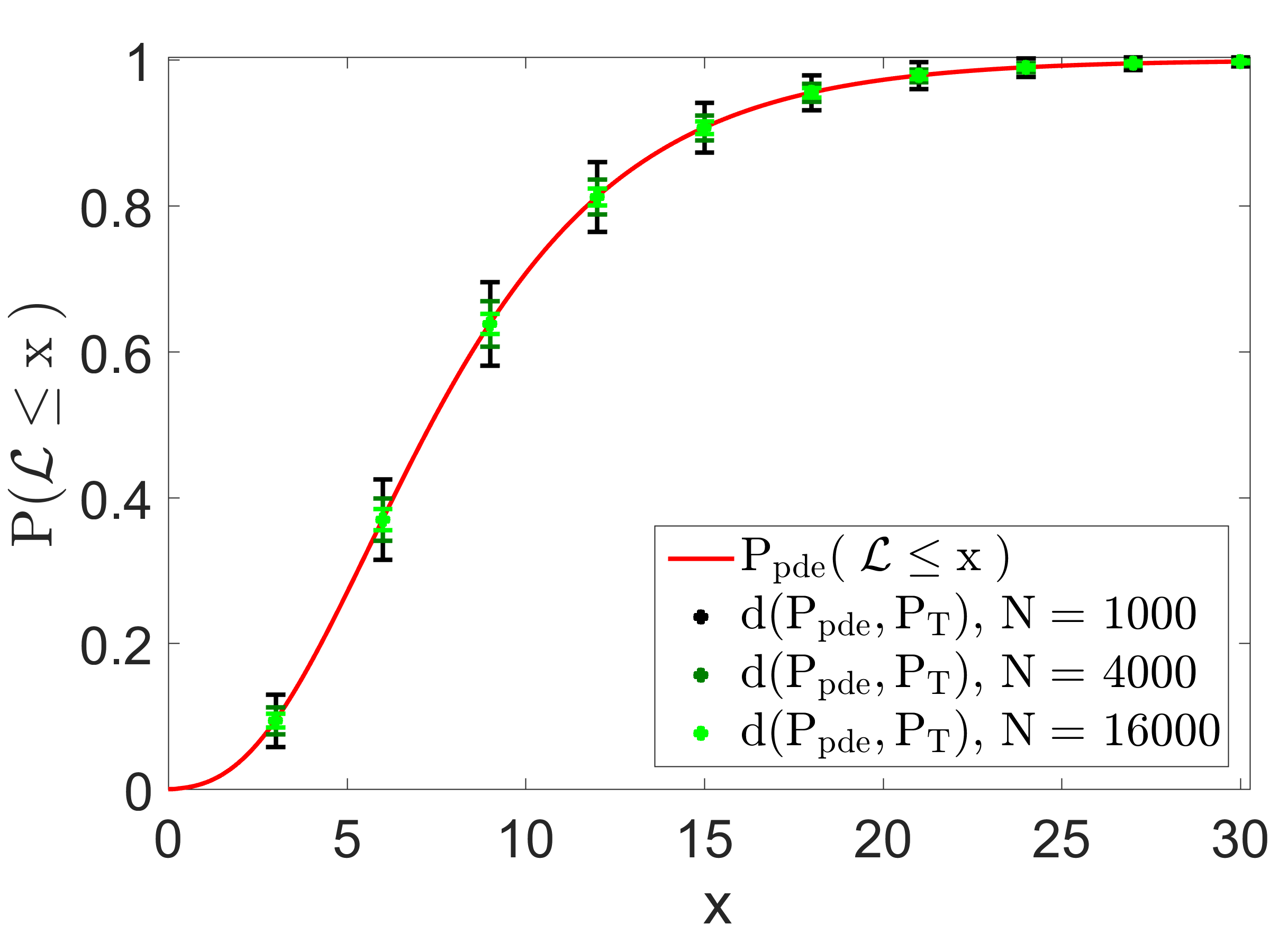}
	\end{center}
	\caption{$\P \{ \totalTreeLength \leq \treeDimA\}$ in the bottleneck scenario \firstRevision{($\invPopSize_b$)}.}
\end{subfigure}
\caption{The CDF $\P \{ \totalTreeLength \leq \treeDimA\}$ \edit{(defined in equation~\eqref{eq_marginal_cdf})} as a function of $\treeDimA$ is depicted by the red line. Additionally, the green bars indicate the bound on the distance between the numerical value $P_\text{pde}$ and the true value $P_\text{T}$ for different $N$, thus the true value is guaranteed to fall within these bounds.}
\label{fig_marginal_cdf}
\end{figure}

\begin{table}
\begin{center}
\begin{tabular}{|c|c|c|c|c|}
\hline $\treeDimA$	& $\treeDimB$	& $p$	& $\hat{p}$ ($N=256,000$)	& $\hat{p}$ ($N=16,384,000$)\\\hline\hline
1.5	& 3.0	& 0.075326	& 0.074914 ($\pm$ 0.002)	& 0.075030 ($\pm$ 0.0002) \\\hline
3.0	& 6.0	& 0.213703	& 0.213324 ($\pm$ 0.002)	& 0.213565 ($\pm$ 0.0002) \\\hline
6.0	& 6.0	& 0.522821	& 0.521578 ($\pm$ 0.002)	& 0.522707 ($\pm$ 0.0003) \\\hline
12.0	& 18.0	& 0.873357	& 0.872840 ($\pm$ 0.002)	& 0.873319 ($\pm$ 0.0002) \\\hline
30.0	& 30.0	& 0.998499	& 0.998504 ($\pm$ 0.0002)	& 0.998516 ($\pm$ 0.00002) \\\hline
\end{tabular}
\end{center}
\caption{The CDF $\P \{ \totalTreeLength^\aLocus \leq \treeDimA , \totalTreeLength^\bLocus \leq \treeDimB \}$ \edit{(defined in equation~\eqref{eq_joint_cdf})} for different values of $\treeDimA$ and $\treeDimB$ under $\invPopSize_1$, with $\sampleSize = 10$ and $\recoRate = 0.001$. $p$ is computed using the numeric algorithm, and $\hat{p}$  is estimated from simulations for different $N$. The confidence bounds are indicated in parentheses.}
\label{tab_PLLlambda1}
\end{table}

\begin{table}
\begin{center}
\begin{tabular}{|c|c|c|c|c|}
\hline $\treeDimA$	& $\treeDimB$	& $p$	& $\hat{p}$ ($N=256,000$)	& $\hat{p}$ ($N=16,384,000$)\\\hline\hline
1.5	& 3.0	& 0.019794	& 0.019238 ($\pm$ 0.0006)	& 0.019579 ($\pm$ 0.00007) \\\hline
3.0	& 6.0	& 0.094393	& 0.094414 ($\pm$ 0.002)	& 0.094172 ($\pm$ 0.0002) \\\hline
6.0	& 6.0	& 0.369581	& 0.369059 ($\pm$ 0.002)	& 0.369544 ($\pm$ 0.0003) \\\hline
12.0	& 18.0	& 0.812236	& 0.812328 ($\pm$ 0.002)	& 0.812109 ($\pm$ 0.0002) \\\hline
30.0	& 30.0	& 0.997696	& 0.997922 ($\pm$ 0.0002)	& 0.997721 ($\pm$ 0.00003) \\\hline
\end{tabular}
\end{center}
\caption{The CDF $\P \{ \totalTreeLength^\aLocus \leq \treeDimA , \totalTreeLength^\bLocus \leq \treeDimB \}$ \edit{(defined in equation~\eqref{eq_joint_cdf})} for different values of $\treeDimA$ and $\treeDimB$ under $\invPopSize_2$, with $\sampleSize = 10$ and $\recoRate = 0.001$. $p$ is computed using the numeric algorithm, and $\hat{p}$  is estimated from simulations for different $N$. The confidence bounds are indicated in parentheses.}
\label{tab_PLLlambda2}
\end{table}

\subsection{Properties of the Distributions}

The results provided in the previous section show that our numerical algorithm can be used to accurately and efficiently compute the marginal and joint CDF of the total tree length in populations with variable size. We will now demonstrate the utility of our numerical method to study the properties of the respective distributions.

The numerical values of the marginal CDF $\P \{ \totalTreeLength \leq \treeDimA\}$ can be readily applied to compute the approximations of the expected value and the variance of the total tree length $\totalTreeLength$. Figure~\ref{fig_expc_var_lambda1} shows the different values of the expectation and the variance under exponential growth \firstRevision{($\invPopSize_e$)}, with varying growth-rates $g$. Recall that a rate of $g=0$ corresponds to no growth.
\firstRevision{Figure~\ref{fig_expc_var_lambda2} shows the expected value and the variance under the bottleneck model ($\invPopSize_b$) for different values of the bottleneck size $N_B$.}
\firstRevision{In both scenarios, the expected value and the variance are smallest in the models with the smallest contemporary population size, corresponding to the largest recent coalescent rate.}
\firstRevision{They increase as $g$, respectively $N_B$, increases, but level off, indicating that increasing the population size has diminishing effects for large values.}
\firstRevision{The absolute value of the expectation is higher in the bottleneck scenario, because, independent of the growth parameter, there is a substantial bottleneck in the growth-scenario.}

\begin{figure}
\begin{subfigure}[b]{0.5\textwidth}
	\begin{center}
	\includegraphics[width=\textwidth]{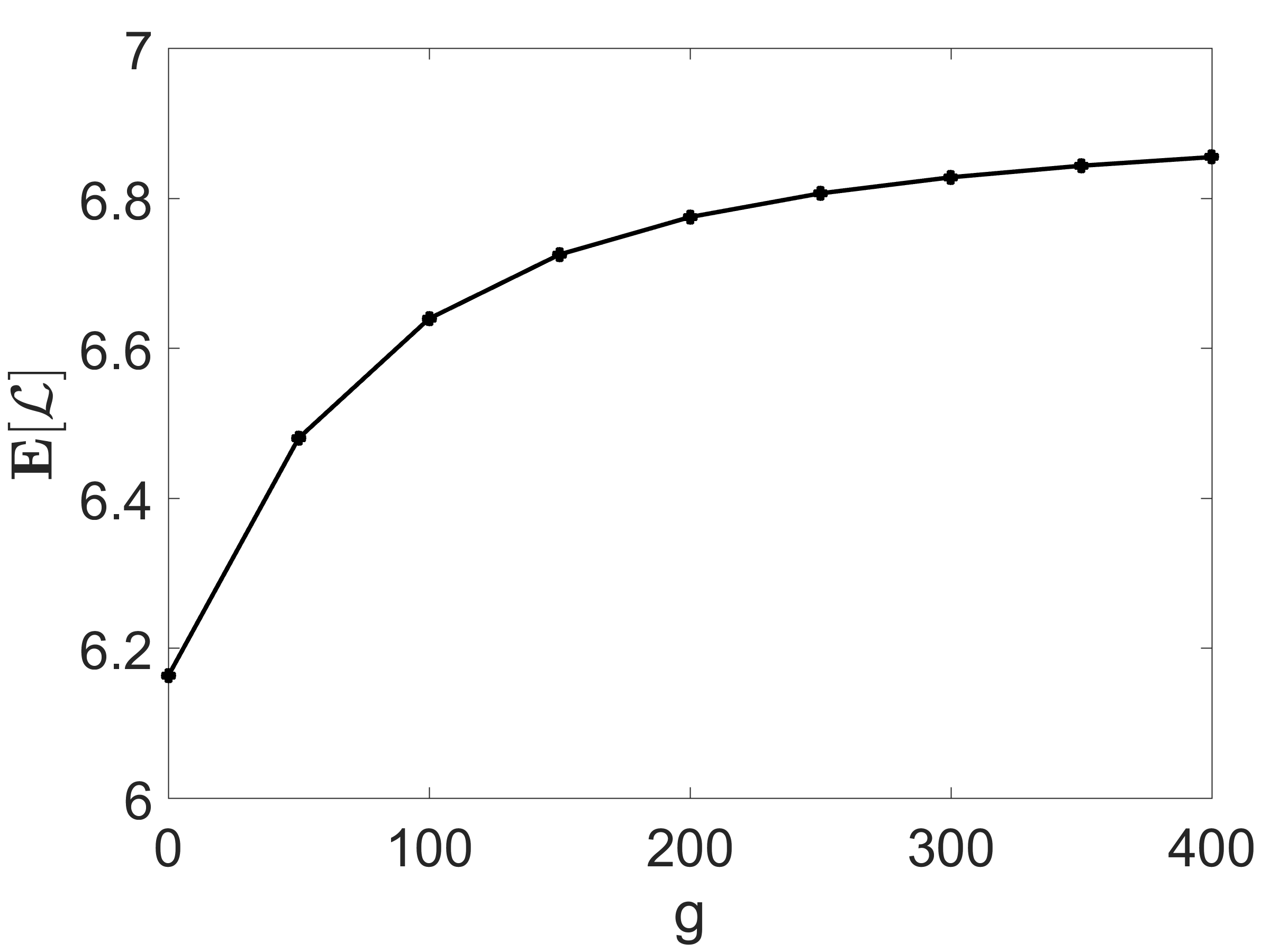}
	\end{center}
	\caption{The expected value of $\totalTreeLength$.}
\end{subfigure}
\begin{subfigure}[b]{0.5\textwidth}
	\begin{center}
	\includegraphics[width=\textwidth]{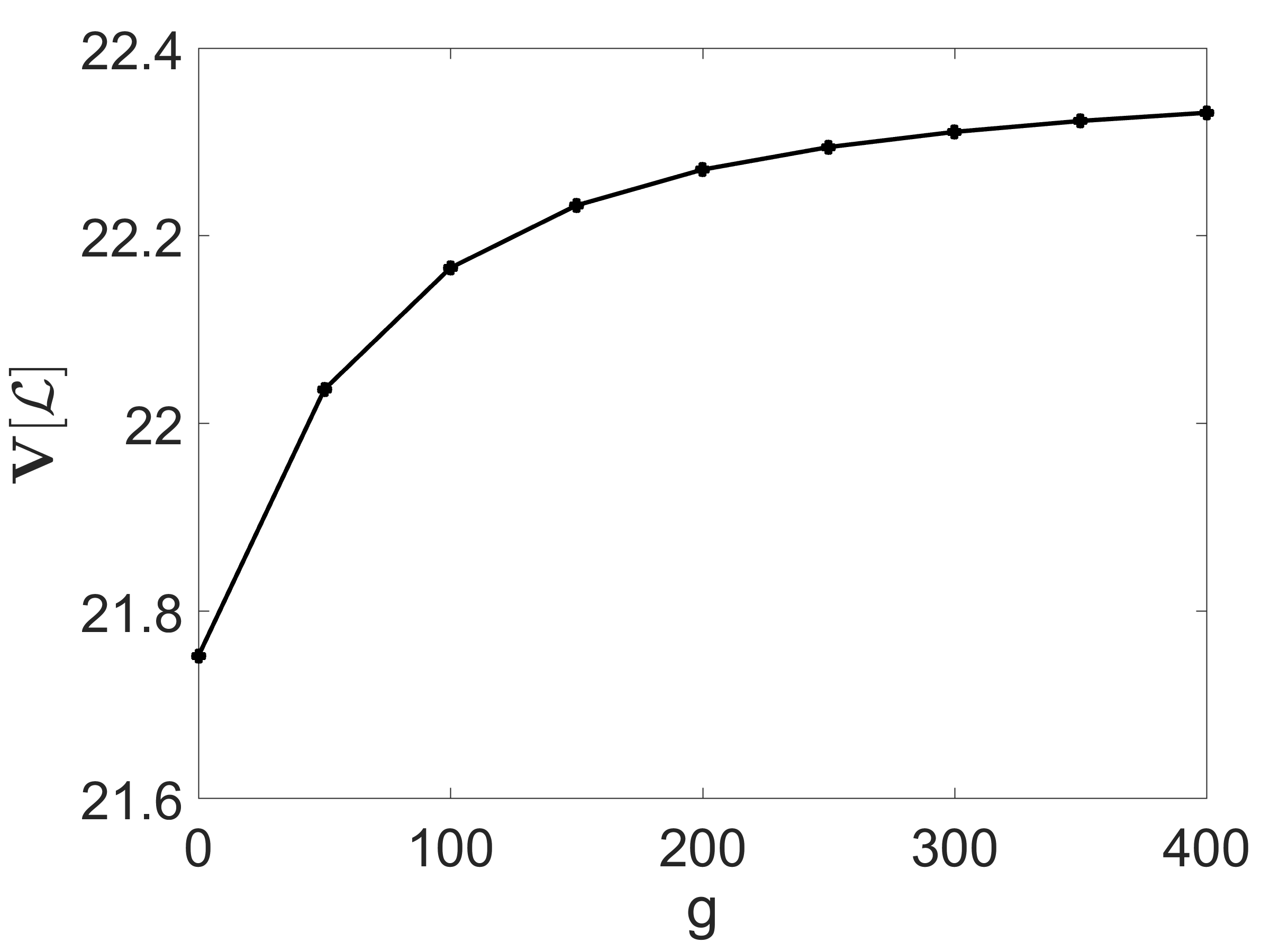}
	\end{center}
	\caption{The variance of $\totalTreeLength$.}
\end{subfigure}
\caption{Approximations to the expected value and the variance of \edit{the total tree length} $\totalTreeLength$ \edit{(defined in equation~\eqref{eq_def_total_tree_length})} computed using our numerical procedure, \firstRevision{under the model for} exponential growth \firstRevision{($\invPopSize_e$)}, with different values for the growth-rate $g$.}
\label{fig_expc_var_lambda1}
\end{figure}

\begin{figure}
\begin{subfigure}[b]{0.5\textwidth}
	\begin{center}
	\includegraphics[width=\textwidth]{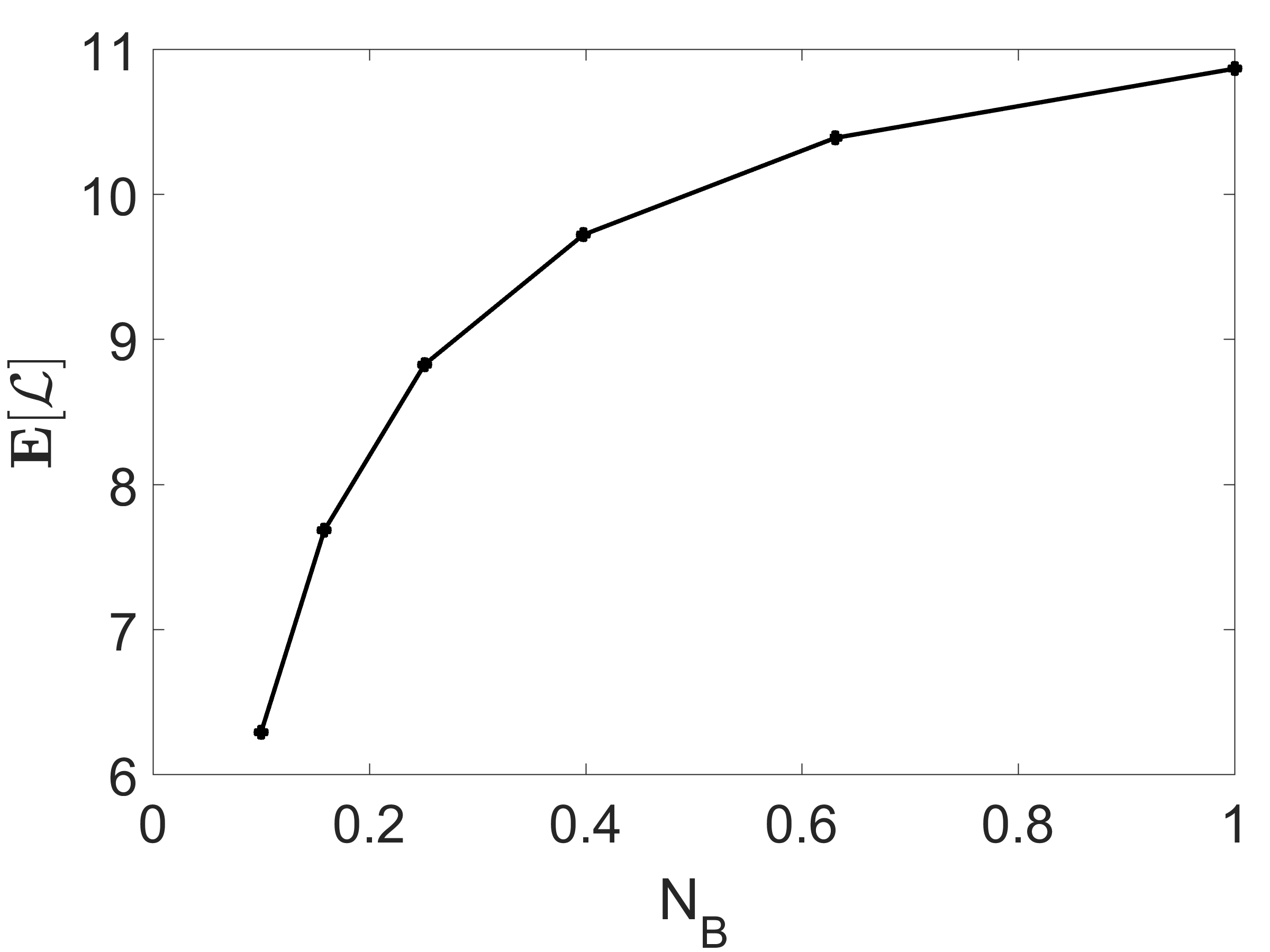}
	\end{center}
	\caption{The expected value of $\totalTreeLength$.}
\end{subfigure}
\begin{subfigure}[b]{0.5\textwidth}
	\begin{center}
	\includegraphics[width=\textwidth]{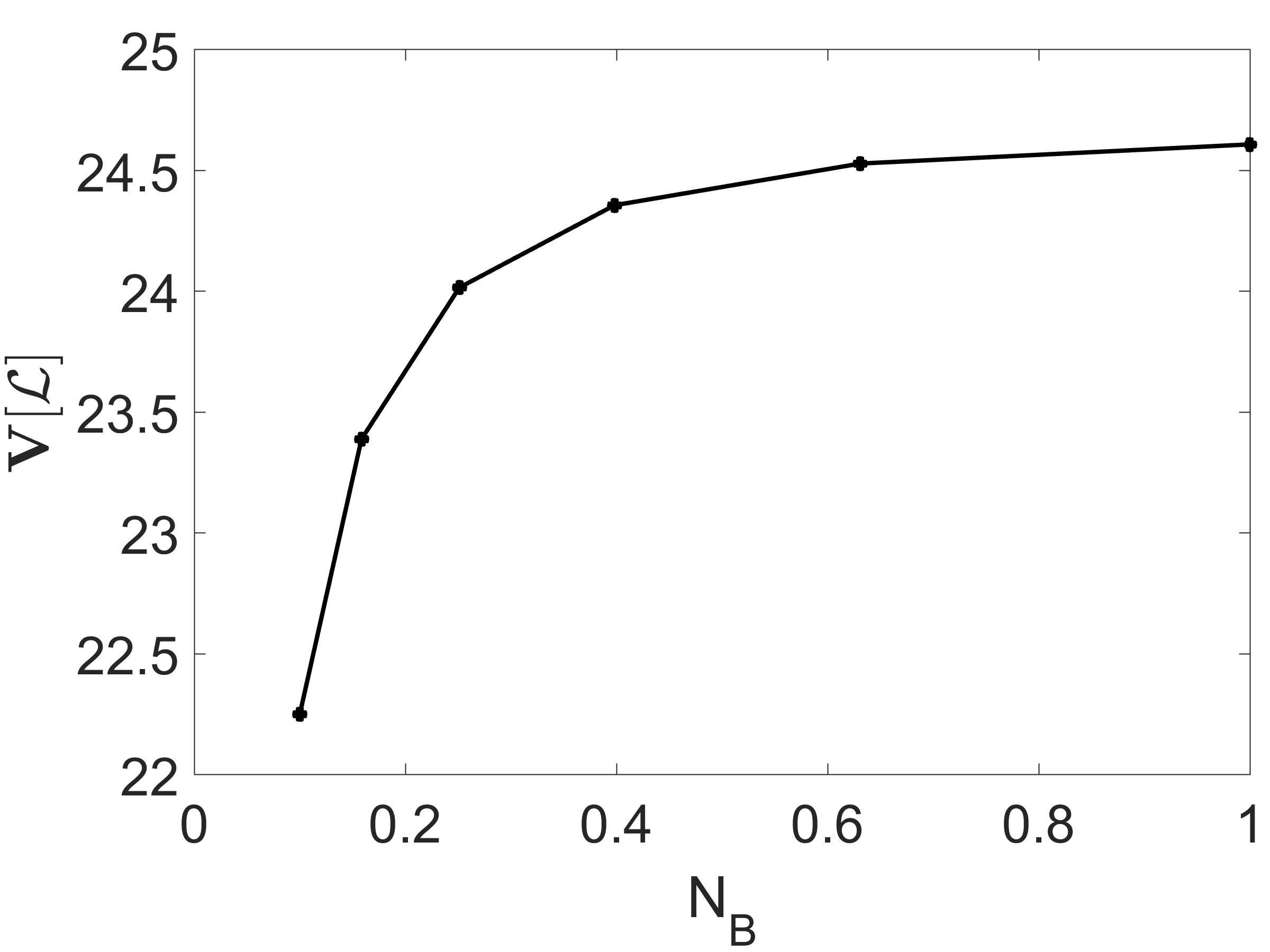}
	\end{center}
	\caption{The variance of $\totalTreeLength$.}
\end{subfigure}
\caption{Approximations to the expected value and the variance of \edit{the total tree length} $\totalTreeLength$ \edit{(defined in equation~\eqref{eq_def_total_tree_length})}, under the bottleneck model \firstRevision{($\invPopSize_b$)}, with different values for the bottleneck size $N_B$.}
\label{fig_expc_var_lambda2}
\end{figure}

Figure~\ref{fig_joint_cdf_surf} shows the joint CDF $\P \{ \totalTreeLength^\aLocus \leq \treeDimA , \totalTreeLength^\bLocus \leq \treeDimB \}$ as a function of $\treeDimA$ and $\treeDimB$ for different population size scenarios and different recombination rates $\recoRate$, computed on a suitable grid using our numerical algorithm. Naturally, the CDF converges towards $1$ as $\treeDimA$ and $\treeDimB$ increase, and due to the symmetry of the ancestral process $\ancestralRecoLimitedProcess$ the CDF is symmetric when interchanging $\treeDimA$ and $\treeDimB$. Furthermore, note that the isolines in the plots for $\recoRate=0.0001$ show pronounced right angles along the line $\treeDimA = \treeDimB$,
\firstRevision{because for small $\recoRate$ the trees at the two loci are highly correlated.}
As the recombination rate increases, the two tree lengths become increasingly uncorrelated, and \firstRevision{these angles soften.}
In all plots, the isoline for $0.2$ is around $\treeDimA = \treeDimB = 5$, for the case \firstRevision{$\invPopSize_e$} even lower. Thus, under \firstRevision{$\invPopSize_e$}, there is an elevated probability for very short trees,
\firstRevision{likely due to the strong bottleneck, which favors short trees.}
Under the constant population size model \firstRevision{$\invPopSize_c$}, the CDF increases rapidly as $\treeDimA$ and $\treeDimB$ increase, whereas the function is less steep for \firstRevision{$\invPopSize_e$ and $\invPopSize_b$}. This behavior seems to be dominated by the ancient population sizes.

\begin{figure}
\begin{center}
\captionsetup[subfigure]{aboveskip=-1.9ex,belowskip=1.5ex}
\begin{subfigure}[b]{0.49\textwidth}
	\begin{center}
	\includegraphics[width=\textwidth]{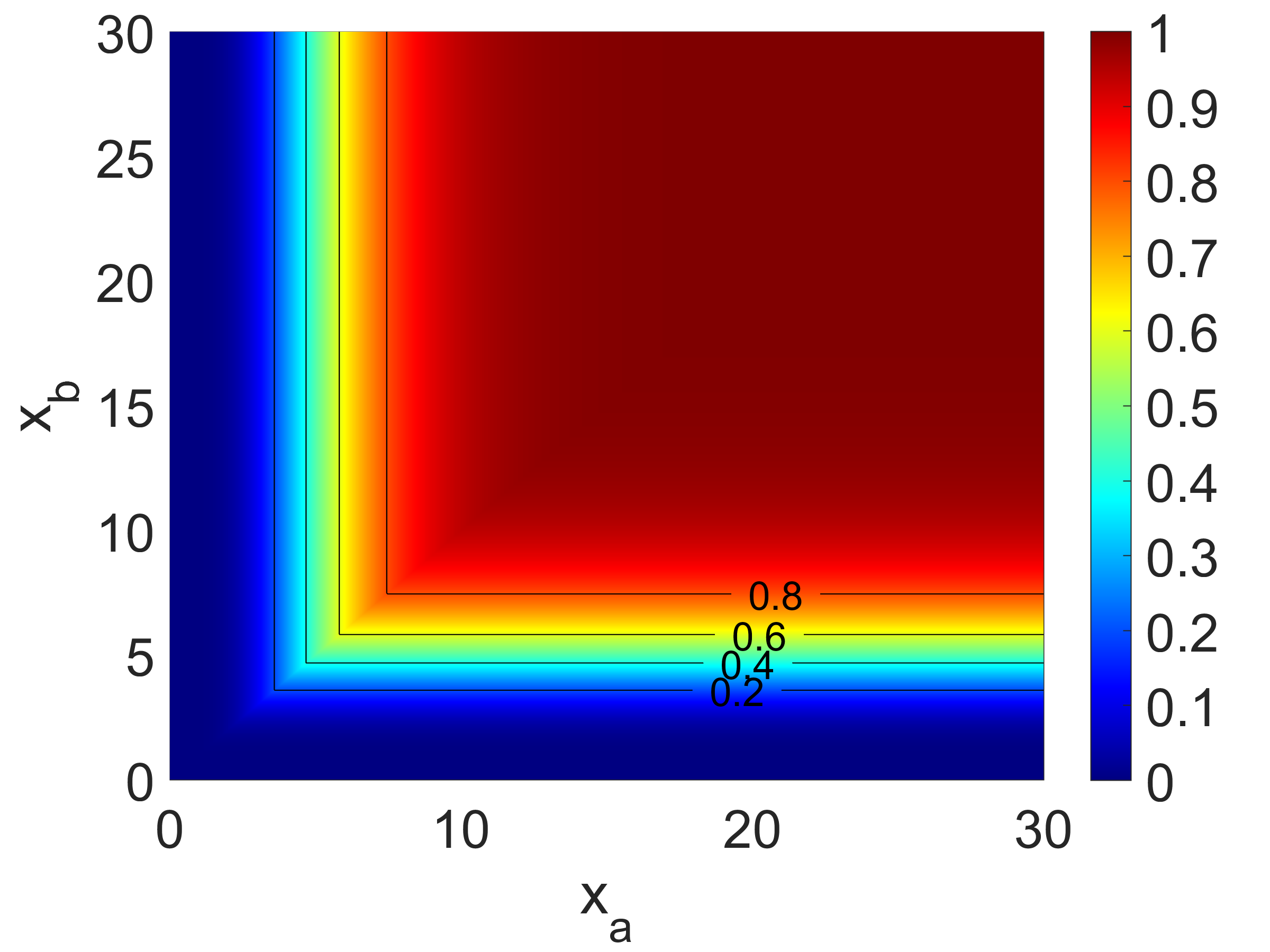}
	\end{center}
	\caption{\firstRevision{Constant population size ($\lambda_c$), $\rho=0.0001$.}}
\end{subfigure}
\begin{subfigure}[b]{0.49\textwidth}
	\begin{center}
	\includegraphics[width=\textwidth]{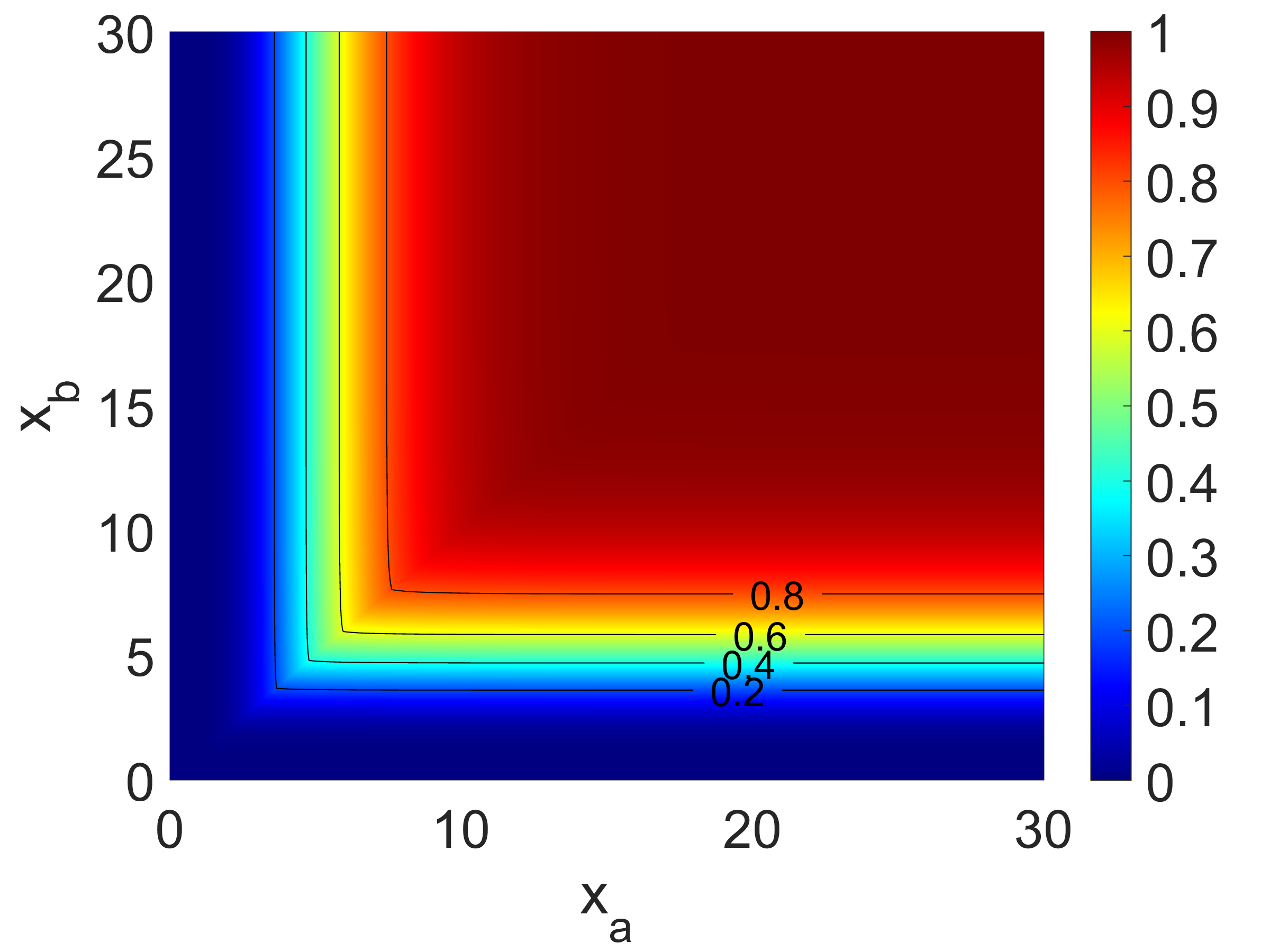}
	\end{center}
	\caption{\firstRevision{Constant population size ($\lambda_c$), $\rho=0.1$.}}
\end{subfigure}
\begin{subfigure}[b]{0.49\textwidth}
	\begin{center}
	\includegraphics[width=\textwidth]{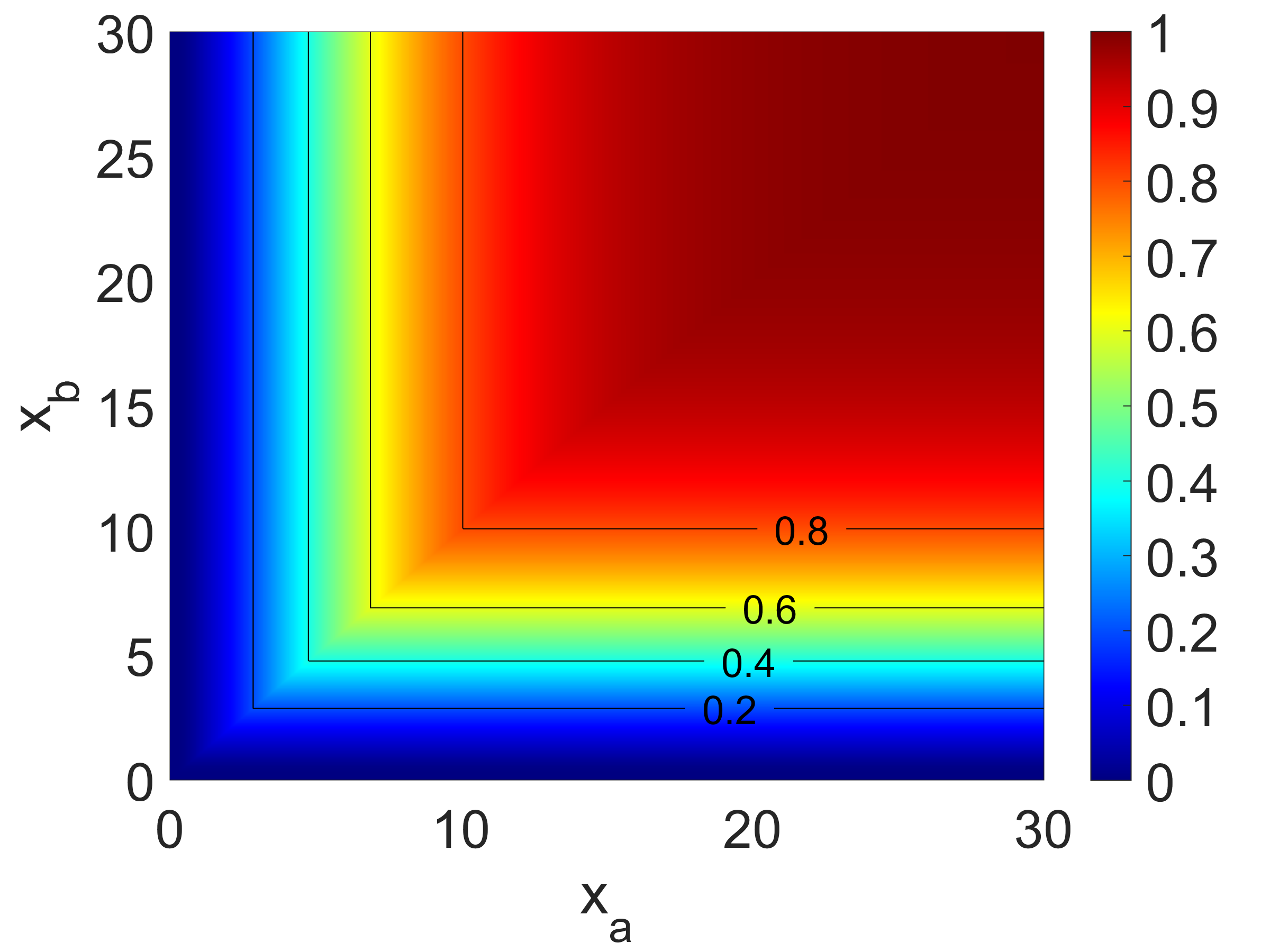}
	\end{center}
	\caption{\firstRevision{Exponential growth ($\lambda_e$), $\rho=0.0001$.}}
\end{subfigure}
\begin{subfigure}[b]{0.49\textwidth}
	\begin{center}
	\includegraphics[width=\textwidth]{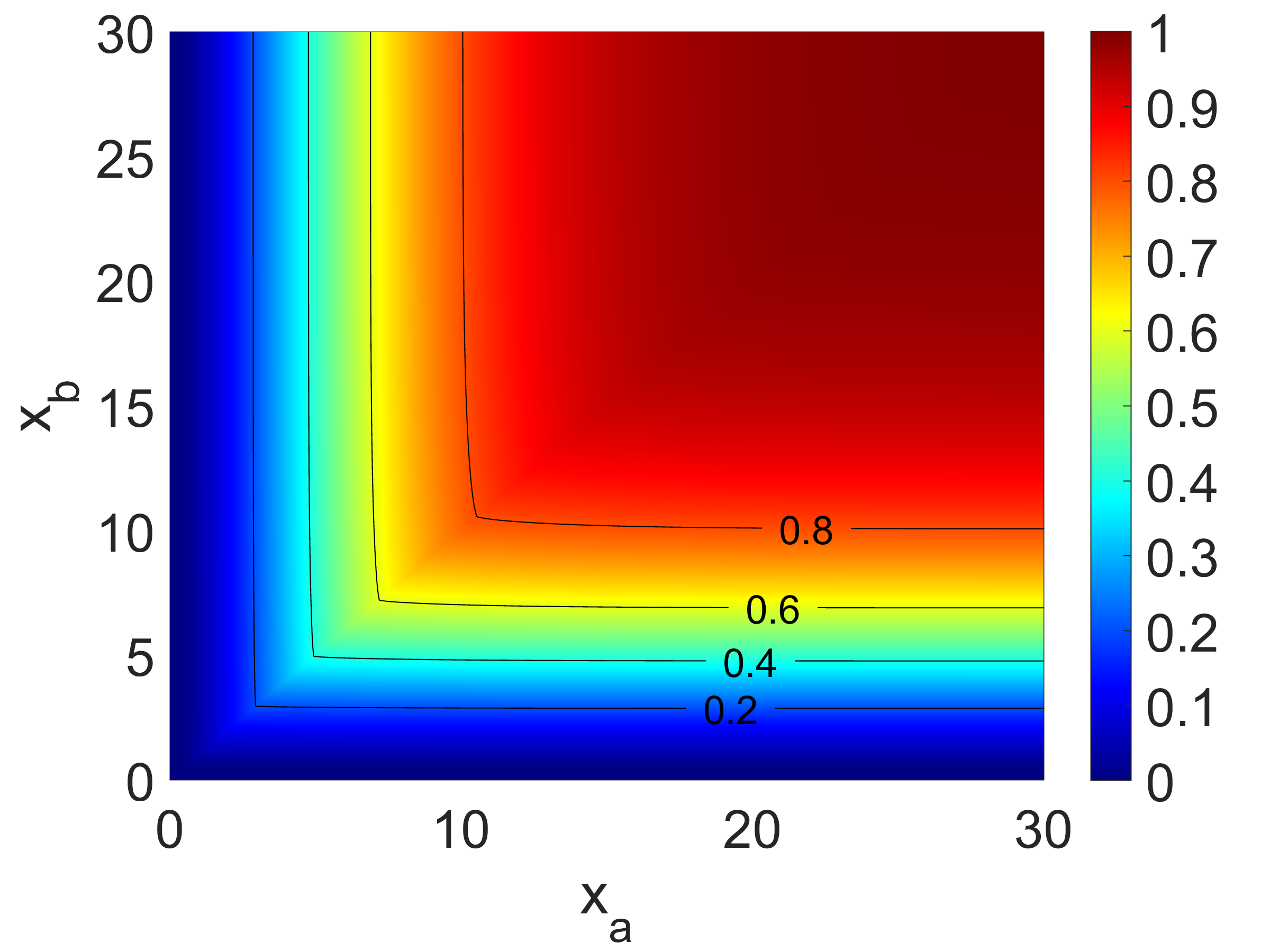}
	\end{center}
	\caption{\firstRevision{Exponential growth ($\lambda_e$), $\rho=0.1$.}}
\end{subfigure}
\begin{subfigure}[b]{0.49\textwidth}
	\begin{center}
	\includegraphics[width=\textwidth]{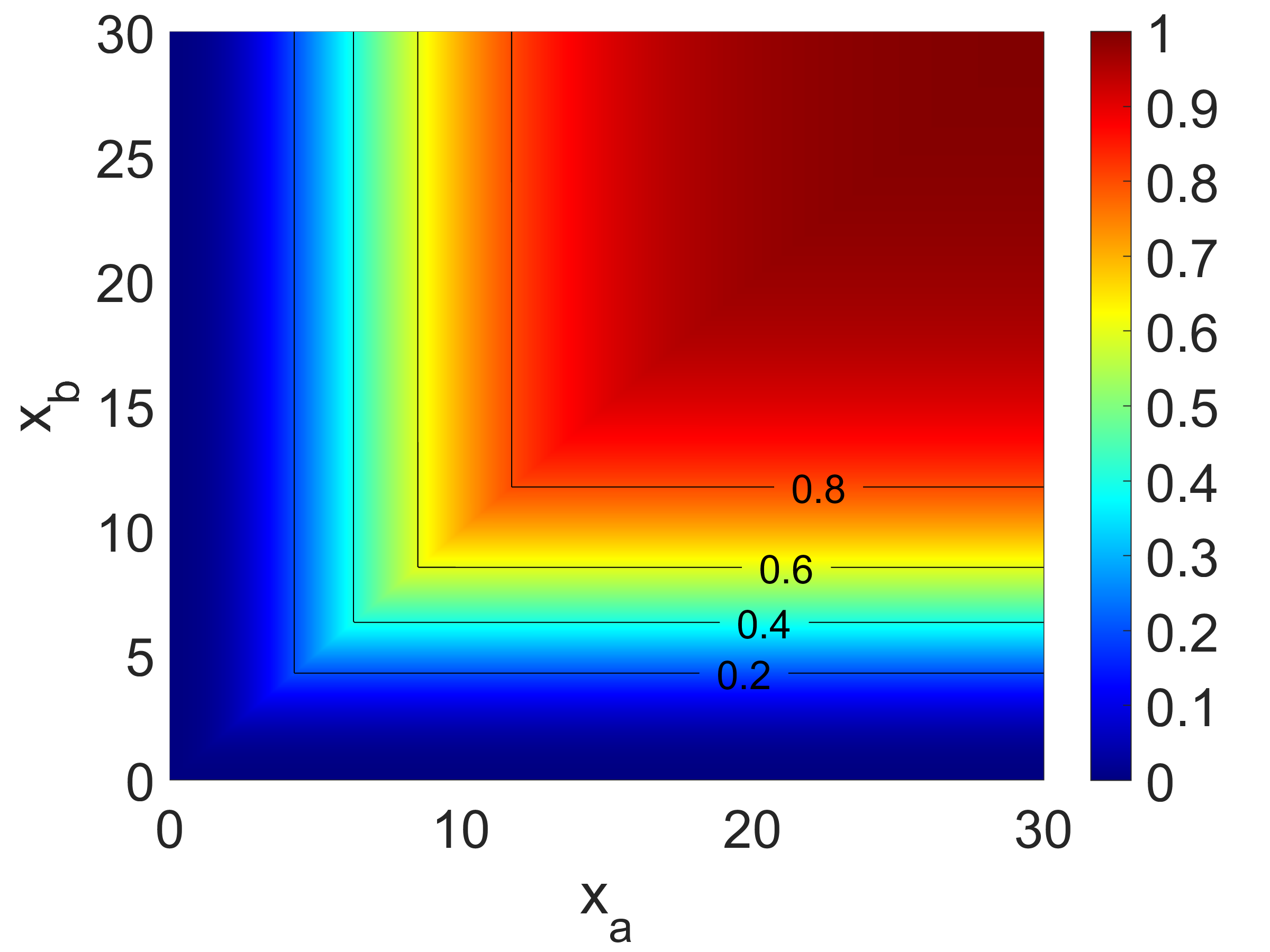}
	\end{center}
	\caption{\firstRevision{Bottleneck model ($\lambda_b$), $\rho=0.0001$.}}
\end{subfigure}
\begin{subfigure}[b]{0.49\textwidth}
	\begin{center}
	\includegraphics[width=\textwidth]{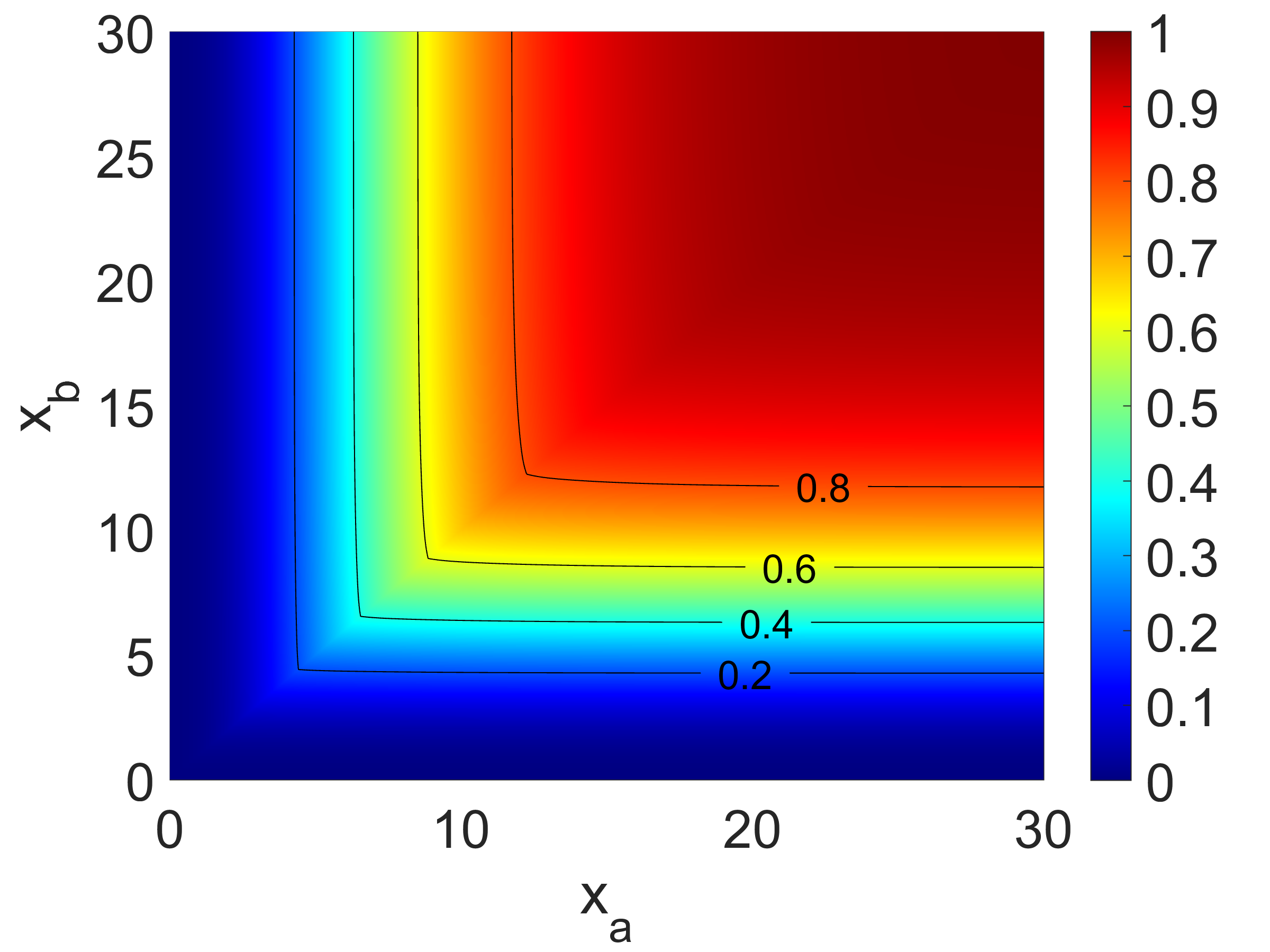}
	\end{center}
	\caption{\firstRevision{Bottleneck model ($\lambda_b$), $\rho=0.1$.}}
\end{subfigure}
\caption{\firstRevision{The joint CDF $\P \{ \totalTreeLength^\aLocus \leq \treeDimA , \totalTreeLength^\bLocus \leq \treeDimB \}$ \edit{(defined in equation~\eqref{eq_joint_cdf})} for the three population models (rows), with different recombination rates $\rho$ (columns). Again, we use $\sampleSize = 10$.}}
\label{fig_joint_cdf_surf}
\end{center}
\end{figure}

Finally, we employ our numerical values of the joint CDF to compute approximations to the correlation coefficient between the tree lengths
\begin{equation}
	\text{corr}(\totalTreeLength^\aLocus, \totalTreeLength^\bLocus ) := \frac{\text{cov}(\totalTreeLength^\aLocus, \totalTreeLength^\bLocus )}{\sqrt{\V \totalTreeLength^\aLocus}\sqrt{\V \totalTreeLength^\bLocus}},
\end{equation}
where $\text{cov}(\cdot,\cdot)$ denotes the covariance.
Figure~\ref{fig_corr} shows this correlation coefficient under the population size \firstRevision{history $\invPopSize_e$} for different values of $\recoRate$, and sample sizes \firstRevision{$\sampleSize = 5$ and $\sampleSize = 20$, respectively.} Recall that our numerical procedure was \firstRevision{derived using} the approximate ancestral process $\ancestralRecoLimitedProcess$ for computational efficiency, where we limited the number of recombination events to $1$. To compare the correlation under the process $\ancestralRecoLimitedProcess$ with the correlation under the regular ancestral process with recombination $\ancestralRecoProcess$,
\firstRevision{we estimated the latter} from repeated simulations using the widely applied coalescent-simulation tool \texttt{ms}~\citep{Hudson2002}, which is based on the regular coalescent with recombination (using \firstRevision{$N=10^7$} repetitions). Naturally, the correlation is close to $1$ for small recombination rates, \firstRevision{and it decreases with increasing recombination rate.}
The values are basically indistinguishable until they start separating around $\recoRate = 0.05$. This is to be expected, since the approximation we introduced limits the number of recombination events to $1$, and thus, as the recombination rate increases, the approximation error also increases.

\firstRevision{To further investigate how restrictive the assumption of at most one recombination event is, we also used the simulated trajectories to estimate the probability that two or more recombination events occur under the regular ancestral process $\ancestralRecoProcess$. The results are shown in Figure~\ref{fig_reco_two_plu}. These probabilities increase with increasing $\recoRate$ and $\sampleSize$. However, they remain small for $\recoRate \leq 0.05$, which is in good agreement with the observation that the correlation is well approximated for $\recoRate$ up to 0.05.}
\firstRevision{In conclusion, Figure~\ref{fig_corr} and Figure~\ref{fig_reco_two_plu} show} that the approximate process can be used without loss \firstRevision{of} accuracy for a large range of recombination rates relevant for human genetics, where recombination rates between neighboring sites are on the order of $10^{-3}$.

\begin{figure}
\begin{subfigure}[b]{0.5\textwidth}
	\begin{center}
	\includegraphics[width=\textwidth]{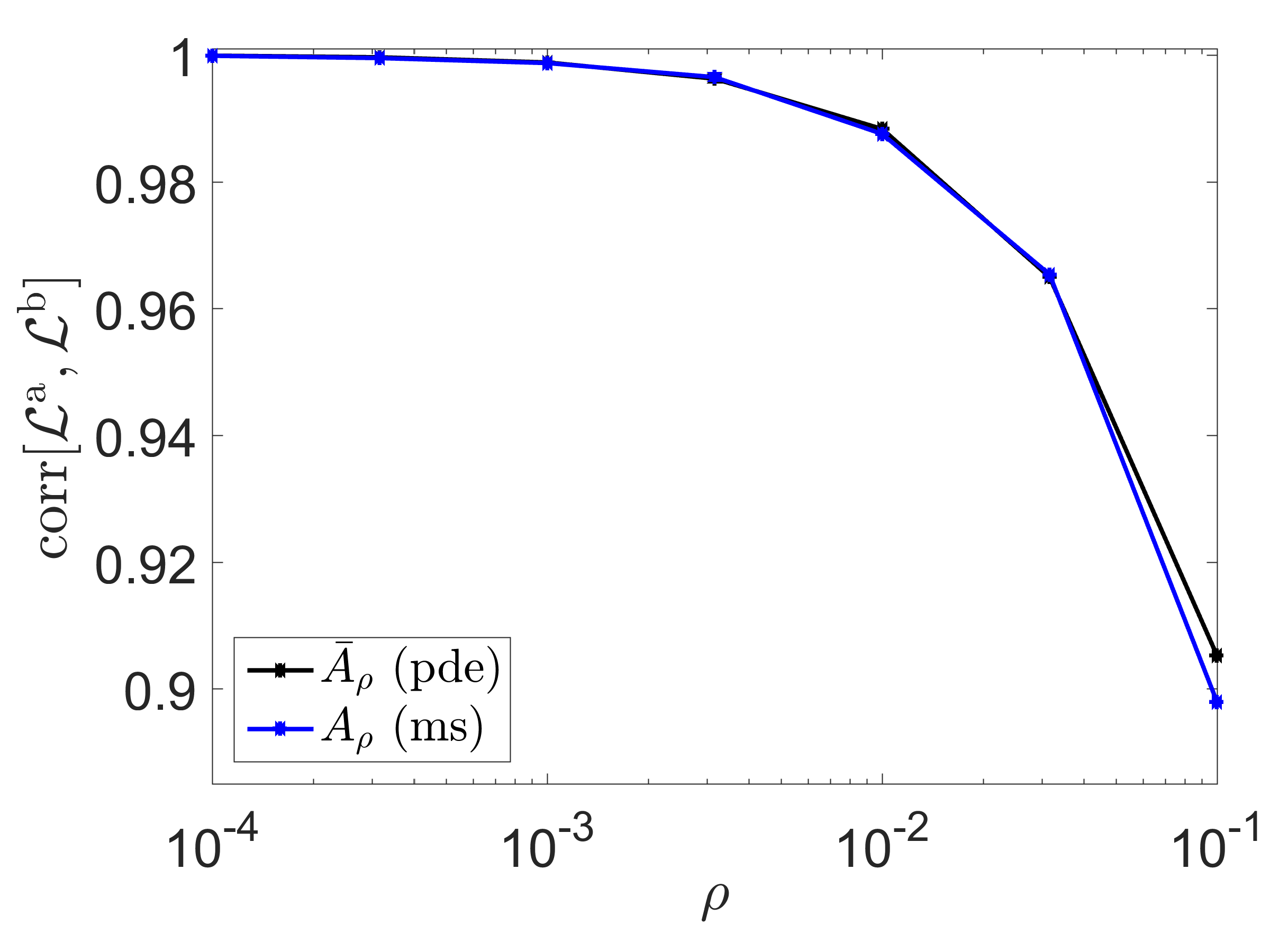}
	\end{center}
	\caption{\firstRevision{Sample size $\sampleSize=5$.}}
\end{subfigure}
\begin{subfigure}[b]{0.5\textwidth}
	\begin{center}
	\includegraphics[width=\textwidth]{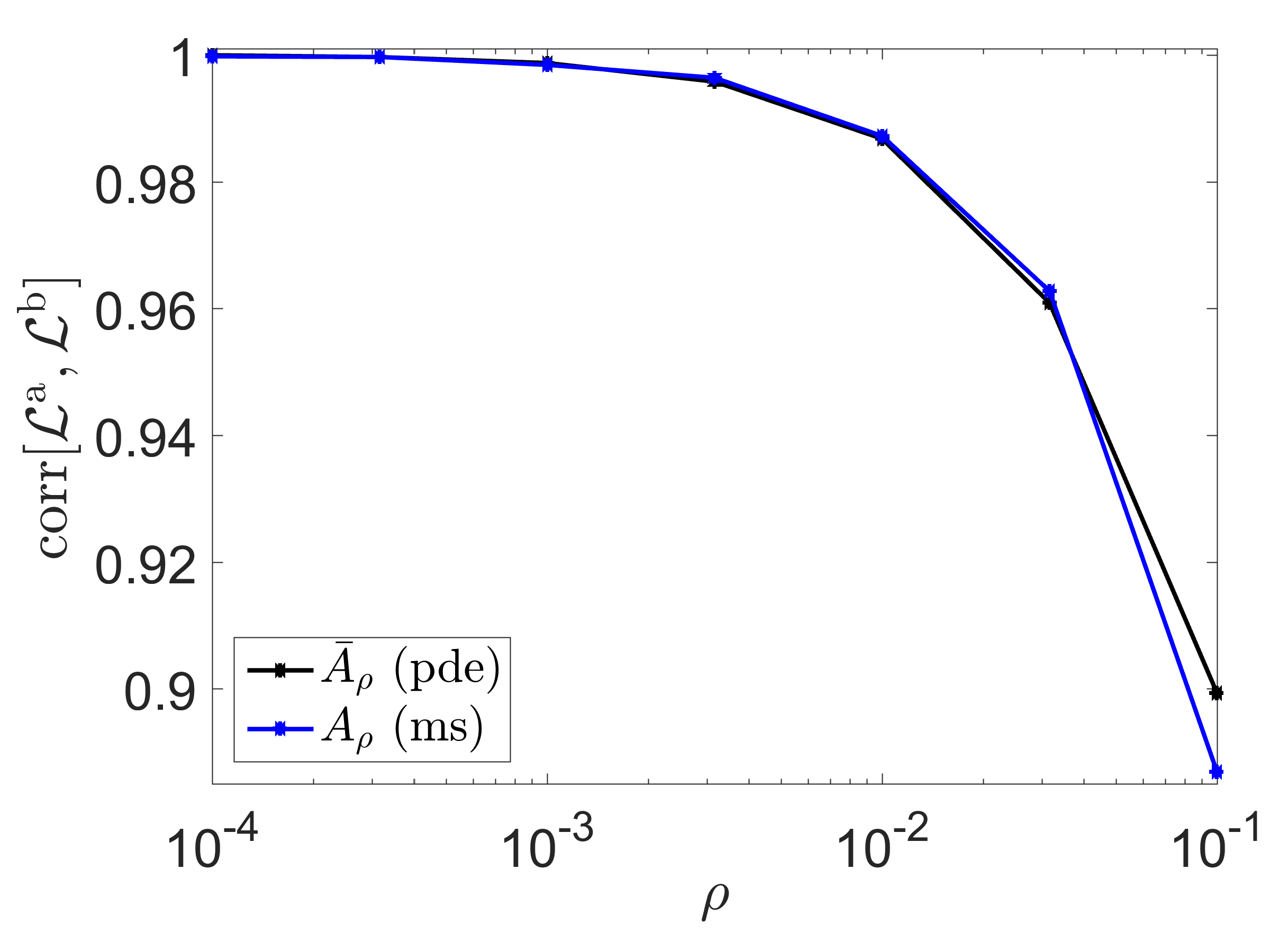}
	\end{center}
	\caption{\firstRevision{Sample size $\sampleSize=20$.}}
\end{subfigure}
\caption{\firstRevision{Correlation between $\totalTreeLength^\aLocus$ and $\totalTreeLength^\bLocus$ \edit{(defined in equations~\eqref{def_total_length_a} and~\eqref{def_total_length_b})} under the exponential growth model ($\invPopSize_e$) for different sample sizes $\sampleSize$ and different recombination rates $\recoRate$. The black lines show the values computed using our method under $\ancestralRecoLimitedProcess$, and the blue lines show values estimated from coalescent simulations under $\ancestralRecoProcess$ using the popular tool \texttt{ms} (using $N=10^7$ repetitions)}.}
\label{fig_corr}
\end{figure}

\begin{figure}
\begin{center}
\includegraphics[width=.5\textwidth]{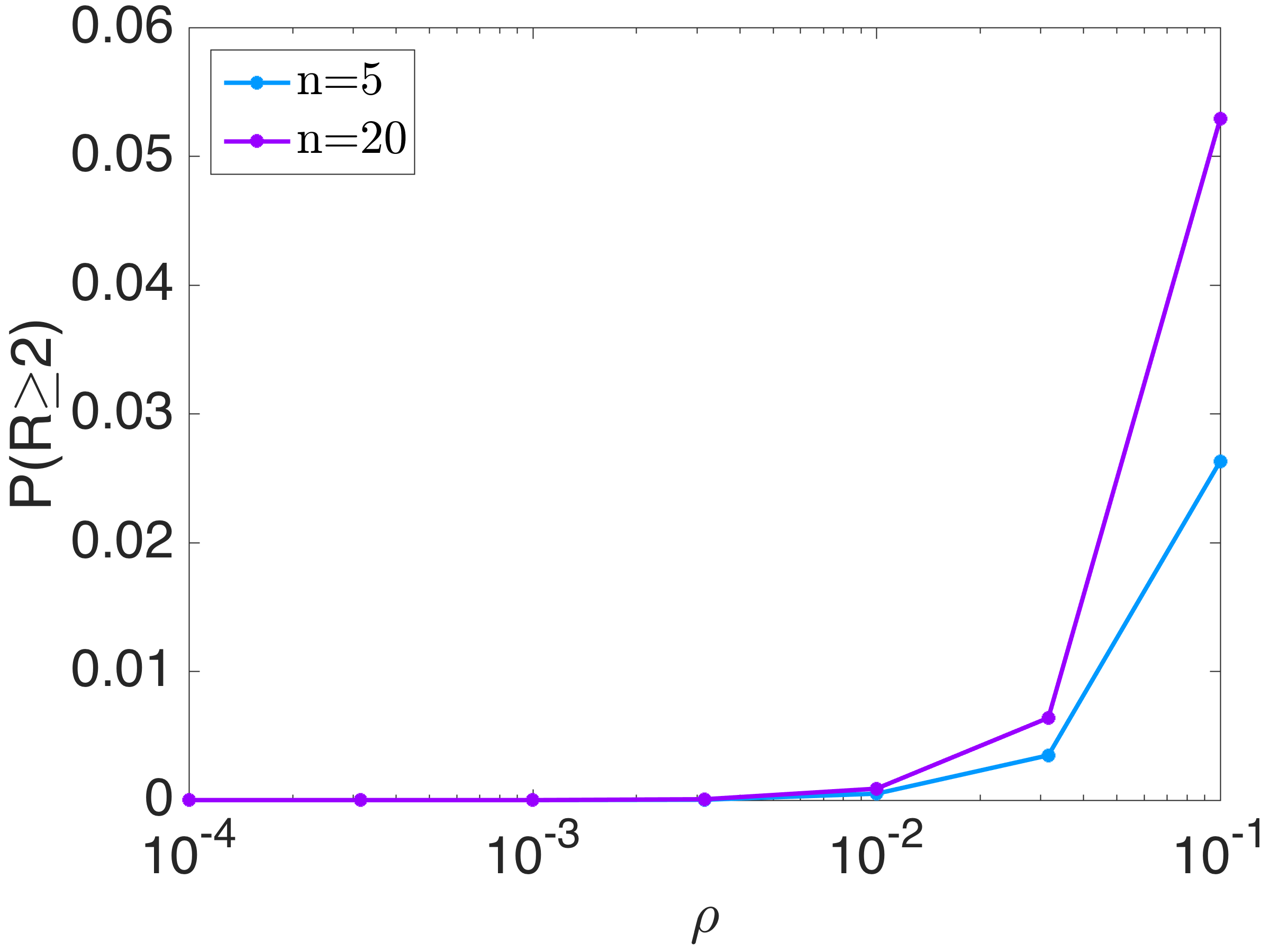}
\end{center}
\caption{\firstRevision{Probability of two or more recombination events $R$ in the regular ancestral process $\ancestralRecoProcess$ \edit{(Definition~\ref{def_anc_proc_reco})}, under exponential growth ($\invPopSize_e$), for different different sample sizes $\sampleSize$ and different recombination rates $\recoRate$. These values were estimated using the coalescent simulation tool \texttt{ms} (using $N=10^7$ repetitions).}}
\label{fig_reco_two_plu}
\end{figure}

\section{Discussion}
\label{sec_discussion}

In this paper, we presented a novel computational framework to compute the marginal and joint CDF of the total tree length in populations with variable size. To our knowledge, these distributions have not been addressed in the literature before, especially in populations of variable size. We introduced a system of linear hyperbolic PDEs and showed that the requisite CDFs can be obtained from the solution of this system. We introduced a numerical algorithm to compute the solution of this system based on the method of characteristics and demonstrated its accuracy in a wide range of biologically relevant scenarios.

The numerical algorithm that we introduced is an upstream-method that computes the requisite solutions step-wise on a grid. We presented the algorithm for a regular, equidistantly spaced grid. We used the trapezoidal rule for the integration steps in the method, and also used linear interpolation to interpolate values that do not fall onto the specified grid.
\firstRevision{We used these basic approaches for ease of exposition.}
Using higher order interpolation and integration schemes, combined with adaptive grids that have more points in regions where the coalescent-rate function is large will most certainly increase the accuracy. However, such higher order schemes come with additional computational cost. This opens numerous avenues for future research to optimize the balance between accuracy and efficiency that is required in the respective applications. 

Moreover, for reasons of computational efficiency, we introduced the \firstRevision{first-order} approximation $\ancestralRecoLimitedProcess$ to the regular ancestral process with recombination $\ancestralRecoProcess$, and computed the joint CDF under this approximate process. We demonstrated that this approximation is accurate for a large range of relevant recombination rates. It is straightforward \firstRevision{to use higher order} approximations, \firstRevision{including more recombination events}, to gain additional accuracy, but computing the joint CDF under the regular ancestral process is desirable.
\firstRevision{Proposition~\ref{prop_app} guarantees that we can use our numerical procedure to compute the requisite CDF under $\ancestralRecoLimitedProcess$, but it is conceivable that it can be extended to more general processes like $\ancestralRecoProcess$ in future work.}

Another research direction is to use our novel framework to study higher order correlations between trees at multiple loci. On the one hand, this could again be correlations between the total tree lengths, \firstRevision{but the distribution of other summary statistics of the genealogical trees could be included}, \firstRevision{for example, the length of the external branches or the length of all branches subtending $k$ leaves. Statistics that have been successfully used in the literature, like the coalescence time between two lineages~\citep{Li2011,Terhorst2017} or the time of the first coalescent event~\citep{Schiffels2014} could be used as well.} Our framework is flexible enough to compute the distribution of multiple path integrals along the trajectories of a given Markov chain. Thus, to implement these additions, one needs to define and implement an appropriate ancestral process and compute suitable integrals along the trajectories.

In this paper, we studied the ancestral process in a single panmictic population. However, in recent years, researchers have gathered an increasing amount of genomic datasets that contain individuals from multiple sub-populations, and studied historical events like migration or population subdivision using these datasets. In light of these studies, it is important to augment our framework to compute joint CDFs of the total tree length in structured populations with complex migration histories. Again, this can be done by introducing suitable ancestral processes and suitable integrals along their trajectories.

\section*{Acknowledgements} 

We thank Yun S. Song for numerous helpful discussions that sparked many of the ideas presented in this paper. This research is supported in part by a National Institutes of Health grant R01-GM094402 (M.S.). We also thank Robin Young for helpful suggestions and fruitful discussions relevant to the design of the numerical scheme.

\appendix

\section*{Appendix}

\section{Path-integrals of Markov chains}
\label{app_proof}

Since the marginal and joint distributions of the tree length can be obtained by integrating a certain function of the ancestral processes, we now consider distributions of path integrals for Markov chains. \edit{We will introduce these distributions assuming Lipschitz continuity in Section~\ref{app_regularity}, and then show in Section~\ref{monotvsec} that these assumptions can be relaxed if the state space is monotone.

\subsection{Path-integrals under regularity assumptions}\label{app_regularity}

Let the Markov chain $X$ be} defined on a probability space $(\Omega,\mathcal{F},\PP)$. We will be using the following assumptions throughout the section.

\begin{itemize}
  \item[(A1)] $\{ X(t,\omega), t \in \RR_+, \omega \in \Omega \}$ is a regular jump Markov process with values in a finite state space ${\mc S}_n$, for convenience labeled ${\mc S}_n = \{1,2,\dots,n\}$, satisfying
  \[
  \PP\{X(t+h)=j|X(t)=i\} = q_{ij}(t) h + o(h)\,, \quad i,j \in {\mc S}_n \quad \text{as} \quad  h \to 0^+.
  \]
We assume that the trajectories of $t \to X(\cdot,\omega)$ are right-continuous. 
  \item[(A2)] The infinitesimal generator $Q(t) = \{ q_{ij}(t)\}_{i,j=1,\ldots,n}$ is conservative, that is
  \[
    q_i(t) := -q_{ii}(t)  = \sum_{i \neq j} q_{ij}(t)\,, 
  \]
  and satisfies $Q \in C(\RR_+; M^{n \times n}) \bigcap L^{\infty}(\RR_+; M^{n \times n})$. In addition, for each $i \in {\mc S}_n$ either $q_i(t)=0$ for all $t \geq 0$ or $q_i(t) >0$ for all $t \geq 0$. In the latter case, we require $\int_0^{\infty} q_i(s) \,ds = \infty$. 
 \end{itemize}

\firstRevision{
 \begin{definition}
Let $X(t)$ satisfy (A1)-(A2). Given  a function 
\[
{\bf v}=(v_1,v_2,\dots,v_d):{\mc S}_n \to \RR^d
\] we define a vector-valued path-integral over the interval $[0,t]$ by
\begin{equation}\label{pint}
  {\bf L}^{\bf v}(t,\omega) := \int_0^t {\bf v}\big(X(s,\omega)\big) \, ds \, \in \, \RR^d \,, \quad t \in \RR_+.
\end{equation}
\end{definition}

\begin{definition}
Let $X(t)$ satisfy (A1)-(A2) and ${\bf v}:{\mc S}_n \to \RR^d$ be some real-valued function defined on the state space. We define a distribution vector-function associated with $\big(X(t), {\bf L}^{\bf v}(t)\big)$  by
\begin{equation}\label{cdfvec}
         {\bf F}^{\bf v}= \big(F^{\bf v}_1,\dots,F^{\bf v}_n \big):\RR_+ \times \RR^d \to \RR_+^n 
\end{equation}
with
\begin{equation}
  \begin{aligned}
F^{\bf v}_k(t,{\bf x}) := \PP\Big\{ X(t) = k\,, {\bf L}^{\bf v}(t) \leq {\bf x} \Big\},
 \end{aligned}
\end{equation}
$k\in\{1,\dots,n\}$, ${\bf x} \in \RR^d$, and the comparison is understood componentwise.
\end{definition}

\begin{definition}\label{def_min_max_v}
Let ${\bf v}:{\mc S}_n \to \RR^d$ be some real-valued function defined on the state space. Define
\begin{equation}
{\bf m}^{\bf v} := \big(\min_{j \in {\mc S}_n} {v_1(j)}, \ldots, \min_{j \in {\mc S}_n} {v_d(j)}\big)
\end{equation}
and
\begin{equation}
{\bf M}^{\bf v} := \big(\max_{j \in {\mc S}_n} {v_1(j)}, \ldots, \max_{j \in {\mc S}_n} {v_d(j)}\big).
\end{equation}
as the componentwise minima and maxima.
\end{definition}

\begin{remark}\rm
Since $X(t)$ is a regular jump process, it is separable. Thus, for each $t_0 \in \RR_+$ the random variable ${\bf L}^{\bf v}(t_0,\cdot)$ is well-defined and $\FF$-measurable, and for each $\omega_0 \in \Omega$ the map $t \to {\bf L}^{\bf v}(t,\omega_0)$ is Lipschitz continuous. This in turn implies that the process ${\bf L}^{\bf v}(t)$ is measurable and separable (see Chapter 12 of \cite{Koralov2007} and Chapter 2 of \cite{Doob1953}).
\end{remark}

\begin{proposition}[{Locally Lipschitz}]\label{pde1dprop}
Let (A1)-(A2) hold. Let ${\bf v}:{\mc S}_n \to \RR^d$ and ${\bf F}^{\bf v}$  be defined by \eqref{cdfvec}. Suppose that ${\bf F}^{\bf v}$ is Lipschitz continuous in an open set  ${\mc U} \subset \RR^+ \times \RR^d$. Then 
\begin{equation}\label{pde1d}
 \del_t {\bf F}^{\bf v}(t,{\bf x}) +  \sum_{j=1}^d \edit{\del_{x_j}{\bf F}^{\bf v}(t,{\bf x}) V_j} = {\bf F}^{\bf v}(t,{\bf x})Q(t) \quad \text{a.e. in \, ${\mc U}$}\,,
 \end{equation}
where $V_j=\diag(v_j(1)\,,\dots,v_j(n))$.
\end{proposition}

To prove this statement, we first need to introduce the following lemmas.

\begin{lemma} \label{est1lmm} Let $X(t)$ satisfy (A1)-(A2),  ${\bf v}:{\mc S}_n \to \RR^d$. Then, for any  $i,k \in {\mc S}_n$, any $\eps>0$, and each ${\bf x} \in \RR^d$, we have
\begin{equation}
\begin{aligned}
   &  \PP \big\{ X(t)=i, \, X(t+\eps) = k  \big\} F_i^{\bf v}(t, {\bf x}- \eps \, {\bf M}^{\bf v})\\[2pt]
   & \quad \leq \PP \big\{ X(t)=i , X(t+\eps) = k\,, {\bf L}^{\bf v}(t+\eps) \leq {\bf x} \big\} \PP\{X(t)=i\} \\[2pt]
   &  \quad \leq \PP \big\{ X(t)=i\,, X(t+\eps) = k   \big\}  F_i^{\bf v}(t, {\bf x}- \eps \, {\bf m}^{\bf v}) \,,
\end{aligned}
\end{equation}
and the comparison is understood componentwise.
\end{lemma}

\begin{proof}
We suppress the superscript ${\bf v}$ in our calculations. Take any $\eps >0$. Observe that
\begin{equation}
{\bf L} (t)+\eps \, {\bf m} \leq {\bf L}(t+\eps)={\bf L}(t)+\int_t^{t+\eps} {\bf v}\big(X(s)\big)\, ds \leq {\bf L} (t)+\eps \, {\bf M}
\end{equation}
and therefore
\begin{equation}
 \begin{aligned}
 &\PP \big\{ X(t)=i,\, X(t+\eps) = k,\, {\bf L}(t) \leq {\bf x} - \eps \, {\bf M} \big\} \\
 &\qquad \leq \PP \big\{ X(t)=i,\, X(t+\eps) = k,\, {\bf L}(t+\eps) \leq {\bf x} \big\}\\
  &\qquad\qquad \leq \PP \big\{ X(t)=i,\, X(t+\eps) = k,\, {\bf L}(t) \leq {\bf x} - \eps \, {\bf m} \big\}\,.
 \end{aligned}
 \end{equation}

Let ${\bf y} \in \RR^d$. Suppose that $0<F_i(t, {\bf y})$. Observe that for any $t_0>0$ the path-integral ${\bf L}(t_0)$ is fully determined by $\big(X(t),t \in [0,t_0)\big)$. This fact (and the separability of the process) enables us to use the Markov property and we obtain
\begin{equation}\label{contprop}
\begin{aligned}
& \PP \big\{ X(t)=i,\, X(t+\eps) = k,\, {\bf L}(t) \leq {\bf y} \big\} \\
&   = \PP \big\{ X(t+\eps) = k| X(t)=i, {\bf L}(t) \leq {\bf y} \big\}\PP\big\{X(t)=i, {\bf L}(t) \leq {\bf y} \big\}\\
&   =  \PP \big\{ X(t+\eps) = k| X(t)=i  \big\}   F_i(t,{\bf y})\,.
\end{aligned}
\end{equation}
This together with the previous inequality finishes the proof.
\end{proof}

\begin{lemma}\label{est2lmm}
Let $X(t)$ satisfy (A1)-(A2) and ${\bf v}:{\mc S}_n \to \RR^d$. For every $(t, {\bf x}) \in \RR_+ \times \RR^d$ and  each $k \in {\mc S}_n$, we have
\begin{equation}\label{charident2}
\begin{aligned}
& \lim_{\eps \to 0^+} \frac{1}{\eps}\bigg| \PP \Big\{ X(t)=k, X(t+\eps)=k, {\bf L}^{\bf v}(t+\eps) \leq {\bf x}+ \eps \, {\bf v}(k) \Big\}  -  (1-q_k(t)\eps) F_k(t, {\bf x})  \bigg| = 0.
\end{aligned}
\end{equation}
\end{lemma}

\begin{proof}
We suppress the superscript ${\bf v}$ in what follows. Due to (A1) the process $X(t)$ is separable and therefore (see \cite[p.~146]{Karlin1981b})
\begin{equation}\label{limdens}
 \begin{aligned}
     &\PP \Big\{ X(s)=k, s \in [t,t+\eps] \Big\}
    = \exp \Big\{ - \int_t^{t+\eps} q_{i}(s) \, ds \Big\} \PP\big\{X(t) = k \big\}.
 \end{aligned}
 \end{equation}

 Suppose now that $F_k(t, {\bf x})>0$. Then, using the Markov property, we obtain
\begin{equation*}
\begin{aligned}
    &\PP \big\{ X(s) = k\,,  s \in [t,t+\eps] \,,   {\bf L}(t+\eps) \leq {\bf x}+\eps \, {\bf v}(k) \big\}\\
&   = \PP \big\{ X(s) = k\,,  s \in [t,t+\eps] \; | X(t)=k, {\bf L}(t) \leq {\bf x} \big\}\PP\big\{X(t)=k, {\bf L}(t) \leq {\bf x}\Big\}\\
&   = \PP \big\{ X(s) = k\,, s \in [t,t+\eps]| X(t)=k  \big\} F_k(t, {\bf x}) = \exp \Big( - \int_t^{t+\eps} q_k(s) \, ds\Big) F_k(t, {\bf x})\,.
\end{aligned}
\end{equation*}

If $F_k(t, {\bf x})=0$, then the first term and the last term in the above identity are zero. Thus, employing \eqref{limdens} and recalling that $Q$ is continuous, we conclude
\begin{equation}
\begin{aligned}
  & \frac{1}{\eps} \Big| \PP \Big\{ X(t)=k, X(t+\eps)=k, {\bf L}(t+\eps) \leq {\bf x}+ \eps \, {\bf v}(k) \Big\}  -  \exp \Big( - \int_t^{t+\eps} q_k(s) \, ds\Big) F_k(t, {\bf x})  \Big| \\
  & \quad = \frac{1}{\eps}  \bigg( \PP \big\{ X(t)=k, X(t+\eps)=k, {\bf L}(t+\eps) \leq {\bf x}+ \eps \, {\bf v}(k) \Big\}  -  \PP \big\{ X(s) = k\,,  s \in [t,t+\eps] \,,  {\bf L}(t+\eps) \leq {\bf x}+ \eps \, {\bf v}(k)\big\} \bigg)  \\
  &  \quad \leq \frac{1}{\eps} \bigg( \PP \Big( X(t)=k, X(t+\eps)=k\Big)  -  \PP \Big( X(s) = k\,,  s \in [t,t+\eps] \Big) \bigg) \\[2pt]
   &  \quad \leq \frac{1}{\eps}\Big(\PP \big\{ X(t)=k, X(t+\eps)=k\big\}  -  \PP\{X(t)=k\}\Big) \\
   & \qquad \qquad - \frac{1}{\eps}\Big[\exp \Big( - \int_t^{t+\eps} q_{i}(s) \, ds \Big) - 1 \Big] \PP\big\{X(t) = k \big\} \to 0 \quad \text{as} \quad \eps \to 0^+.
\end{aligned}
\end{equation}
Since $(1-q_k(t)\eps) F_k(t, {\bf x}) = \exp \Big( - \int_t^{t+\eps} q_k(s) \, ds\Big) F_k(t, {\bf x}) + o(\eps^2)$, the statement of the lemma follows.
\end{proof}

We can now turn back to the proof of Proposition~\ref{pde1dprop}

\begin{proof}[Proof of Proposition~\ref{pde1dprop}]
Let us suppress the superscript ${\bf v}$ in our calculations. Let $\widetilde{\mc U}$ denote the set of all points in ${\mc U}$ at which ${\bf F}$ is differentiable. Since ${\bf F}$ is Lipschitz continuous in ${\mc U}$,  Rademacher's theorem~\citep{Federer1969} implies that ${\bf F}$ is Lebesgue almost surely differentiable in ${\mc U}$ and therefore ${\mc U}\backslash \widetilde{\mc U}$ is of Lebesgue measure zero. Take any $k \in {\mc S}_n$. Fix any $(t, {\bf x}) \in \widetilde{\mc U}$.  For any $\eps>0$ we have
\begin{equation}\label{identchar}
	\begin{split}
		F_k\big(&t+\eps, {\bf x}+ \eps \, {\bf v}(k)\big) -F_k(t,{\bf x}) \\ 
			& = \Big(\sum_{i=1}^n \PP\{ X(t)=i\,,X(t+\eps)=k\,, {\bf L}(t+\eps) \leq {\bf x}+ \eps \, {\bf v}(k)\}\bigg)-F_k(t,{\bf x}).
	\end{split}
\end{equation}

Consider first the terms with $i \neq k$. By Lemma \ref{est1lmm}, we have
\begin{equation}
\begin{aligned}
   &  \frac{1}{\eps}\PP \big\{ X(t+\eps) = k,  X(t)=i  \big\} F_i\big(t+\eps, {\bf x}+ \eps \, {\bf v}(k) - \eps \, {\bf M} \big)\\[2pt]
   & \quad \leq \frac{1}{\eps}\PP \big\{ X(t)=k , X(t+\eps) = i\,, {\bf L}(t+\eps) \leq {\bf x}+ \eps \, {\bf v}(k)\big\} \PP\{X(t)=i\} \\[2pt]
   &  \quad \leq \frac{1}{\eps}\PP \big\{ X(t+\eps) = k, X(t)=i  \big\}  F_i\big(t+\eps, {\bf x}+ \eps \, {\bf v}(k) - \eps \, {\bf m} \big)
\end{aligned}
\end{equation}
where  ${\bf m}$ and ${\bf M}$ are as in Definition~\ref{def_min_max_v}. Since $Q$ is continuous we must have
\[
\lim_{\eps \to 0^+}\frac{1}{\eps}\PP \big\{ X(t+\eps) = i,  X(t)=k  \big\}  = q_{ki}(t)\PP\{X(t)=k\}
\]
and hence, employing the continuity of ${\bf F}$, we conclude
\begin{equation}\label{chard1}
\lim_{\eps \to 0^+}\frac{1}{\eps}\PP \Big\{ X(t)=i, X(t+\eps) = k, \, {\bf L}(t+\eps) \leq {\bf x}+ \eps \, {\bf v}(k) \Big\}  = q_{ki}(t) F_i(t, {\bf x})\,.
\end{equation}

Next consider the case $i=k$. By Lemma \ref{est2lmm} we have
\begin{equation}\label{chard2}
\lim_{\eps \to 0^+}\frac{1}{\eps}\Big(\PP \big\{ X(t)=k, X(t+\eps) = k, \, {\bf L}(t+\eps) \leq {\bf x}+ \eps \, {\bf v}(k) \big\} - F_k(t,{\bf x}) \Big) = -q_{k}(t) F_k(t,{\bf x})\,.
\end{equation}

Combining \eqref{identchar} with \eqref{chard1} and \eqref{chard2}  we obtain
\begin{equation}\label{derivchar}
\begin{aligned}
   &\lim_{\eps \to 0^+} \frac{1}{\eps}\Big(F_k\big(t+\eps, {\bf x}+ \eps \, {\bf v}(k)\big) -F_k(t, {\bf x}) \Big) \\
   &\qquad = -q_k(t)F_k(t, {\bf x})+\sum_{i \neq k} q_{ik}(t) F_i(t, {\bf x}) = \big({\bf F}(t, {\bf x})Q(t)\big)_k\,.
\end{aligned}
\end{equation}

Since ${\bf F}$ is differentiable at $(t, {\bf x}) \in \widetilde{\mc U}$ and the map $\eps \to \big(t+\eps, {\bf x}+ \eps \, {\bf v}(k)\big)$ is differentiable with the image contained in ${\mc U}$ for sufficiently small $\eps$, the chain rule is applicable (see \cite{Rudin1976}[Theorem 9.15]) and we conclude
\begin{equation}\label{chainr}
	\begin{split}
		\lim_{\eps \to 0^+} \frac{1}{\eps}\Big(F_k\big(t+\eps, {\bf x}+ \eps \, {\bf v}(k)\big) - F_k(t, {\bf x}) \Big) & = \frac{d}{d \eps} F_k\big(t+\eps, {\bf x}+ \eps \, {\bf v}(k)\big)\Big|_{\eps=0}\\
        	& = \del_t F_k(t, {\bf x}) + \sum_{j=1}^d v_j(k)\,\del_{x_j} F_k(t,{\bf x}).
	\end{split}
\end{equation}
Since both $k \in {\mc S}_n$ and $(t,{\bf x}) \in \widetilde{\mc U}$ were arbitrary, \eqref{derivchar} implies \eqref{pde1d}.
\end{proof}

We next show that ${\bf F}^{\bf v}$ as $t \to 0^+$ has certain continuity properties.
\begin{proposition}\label{inival}
  Let (A1)-(A2) hold. Let  ${\bf v}:{\mc S}_n \to \RR^d$ and ${\bf F}^{\bf v}$ as defined by \eqref{cdfvec}. Then 
    \begin{equation}\label{Finival}
    \lim_{t  \to 0^+}F^{\bf v}_k(t, {\bf x}) =  \1_{\RR_+^d}({\bf x}) \, \PP\{X(0)=k\} = F^{\bf v}_k(0, {\bf x}) \quad \text{for} \quad  {\bf x} \notin \partial \RR_+^d.
  \end{equation}
\end{proposition}

\begin{proof}[Proof]

First, we note that ${\bf L}(0)=0$ for all $\omega \in \Omega$ and hence $F^{\bf v}_k(0,{\bf x}) = \1_{\RR_+^d}({\bf x}) \PP\{X(0)=k\}$.

Observe that
    \begin{equation}
     {\bf m}^{\bf v}t \leq {\bf L}(t)=\int_0^t {\bf v}(X(s))\, ds \leq {\bf M}^{\bf v}t,
    \end{equation}
    where ${\bf m}^{\bf v}$ and ${\bf M}^{\bf v}$ as in Definition~\ref{def_min_max_v}. Fix $\delta>0$. Then for all $0<t<\delta(1+\max(||m^v||_\infty,||M^v||_\infty))^{-1}$, and every ${\bf x} \in \RR^d$ such that  $||{\bf x} - {\bf y}||>\delta$ for all $y \in \partial \RR_+^d$, we have
    \begin{equation}
    \begin{aligned}
     F_k(t,{\bf x}) &=\1_{\RR_+^d}({\bf x})\PP\{X(t)=k, {\bf L}(t) \leq {\bf x}\}\\
    &=\1_{\RR_+^d}({\bf x})\PP\{X(t)=k\} \to \1_{\RR_+^d}({\bf x})P\{X(0)=k\} \quad \text{as} \quad t \to 0^+.
    \end{aligned}
    \end{equation}
\end{proof}

\begin{remark}
From Proposition~\ref{inival} it follows that the `initial values'  of ${\bf F}^{\bf v}$ are discontinuous. Since the system \eqref{pde1d} a is linear hyperbolic system, discontinuities present at time $t=0$ will travel in space as time $t$ increases and therefore ${\bf F}^{\bf v}$ is not $C^1$ or even continuous. Nevertheless, one can show that \eqref{pde1d} holds in a weaker sense. To do that one needs to employ the notion of weak solutions, and we will pursue this avenue in an upcoming paper. However, for certain type of state space functions ${\bf v}$, relevant to our application, one can show additional regularity properties of ${\bf F}^{\bf v}$. We provide more details in Appendix~\ref{monotvsec}.
\end{remark}
}

\subsection{Path-integrals for monotone state space functions}\label{monotvsec}

Hyperbolic systems of partial differential equations  admit in general solutions that are not classical even if the initial (or boundary data) is smooth. Typically there are two distinct classes of solutions: strong solutions, which are Lipschitz continuous (see~\cite{Dafermos2010}), and weak solutions, which allow for discontinuities. Here, we will be using the first type of solutions.

\firstRevision{
\begin{definition}
Let ${\mc U} \subset \RR_+ \times \RR^d$ be open and let $A_1,\ldots,A_d,B \in L^{\infty}({\mc U};\, \RR^{n \times n})$.  We say that ${\bf u}(t,{\bf x}):\RR_+ \times \RR^d \to \RR^n$ is a {strong} solution of
  \begin{equation}\label{linsystpde}
    \del_t {\bf u}(t, {\bf x}) + \sum_{j=1}^d \edit{\del_{x_j} \big\{ {\bf u}(t,{\bf x}) \big\} A_j(t,{\bf x})} = \edit{ {\bf u}(t, {\bf x}) B(t,{\bf x})} \quad \text{in} \quad {\mc U} \subset \RR_+ \times \RR^d,
  \end{equation}
if ${\bf u}$ is Lipschitz continuous in ${\mc U}$, and the equation \eqref{linsystpde} holds for Lebesgue almost all points $(t, {\bf x})$ in ${\mc U}$.
\end{definition}

\smallskip

\begin{remark}
 By Rademacher's theorem~\citep{Federer1969} a function ${\bf u}(t,{\bf x})$ that is Lipschitz continuous in an open domain ${\mc U}$ is Lebesgue almost sure differentiable in ${\mc U}$. In fact, its pointwise partial derivatives, which exist almost everywhere, coincide with its corresponding weak partial derivatives (see \cite{Evans2010}).
\end{remark}

Regularity of solutions to hyperbolic problems depends on both the initial (or boundary) data and the domain itself. For linear hyperbolic problems as long as the initial data is smooth and the domain has a smooth boundary one may expect a solution to be (locally) smooth. Typically one studies solutions to hyperbolic problems on the domain ${\mc U}=\RR_+ \times \RR^d$ with initial data ${\bf u}_0({\bf x})$ at $t=0$ (Cauchy problem). The initial data for the vector of probabilities ${\bf F}^{\bf v}$ (studied in ${\mc U}$) are unfortunately discontinuous (which is shown below).  To avoid unnecessary difficulties, in Proposition \ref{prop_app} we split the space-time domain into two regions ${\mc U}_I$ and ${\mc U}_E$. In ${\mc U}_E$  the values of ${\bf F}^{\bf v}$ admit a simpler form while in ${\mc U}_I$ the vector ${\bf F}^{\bf v}$ is obtained via solving a linear hyperbolic system with smooth initial data. We note that the components of ${\bf F}^{\bf v}$ are in general merely Lipschitz continuous in ${\mc U}_I$. This is not surprising for two reasons. First, the domain is singular because it has a `corner' and the discontinuities of the derivatives of ${\bf F}^{\bf v}$ originating at points $\partial \RR_+^d$ travel along the corresponding characteristics. Second, the vector ${\bf F}^{\bf v}$ solves the same system of equations in the domain ${\mc U}$ with discontinuous initial data and hence it is in general not smooth.


\begin{definition}
  Let (A1)-(A2) hold. Let ${\bf v}:{\mc S}_n \to \RR^d$. We say that ${\bf v}$ is monotone along the process $X(t)$ if the map $t \to {\bf v}(X(t))$ is either non-increasing $\PP$-almost surely or non-decreasing $\PP$-almost surely.
\end{definition}

\begin{definition}\label{def_us}
Define the following regions in $\RR_+ \times \RR^d$:
\begin{equation}
 {\mc U}_I := \Big\{(t, {\bf x}): t>0,\, {\bf x} < {\bf M}^{\bf v} t\Big\} 
\end{equation}
and
\begin{equation}
{\mc U}_E := \big({\mc U}_I \cup \partial {\mc U}_I\big)^c,
\end{equation}
where the comparison is understood componentwise.
\end{definition}

\begin{proposition}\label{prop_app}
 Let $X(t)$ satisfy (A1)-(A2), $X(0)=n$, and $Q(t)$ be lower triangular. Let ${\bf v}:{\mc S}_n \to \RR^d$ be monotone along $X(t)$. Suppose that  $t \to {\bf v}(X(t))$ is non-increasing on $\Omega$, and ${\bf m}^{\bf v} < {\bf M}^{\bf v}$. For ${\bf x} \in \RR^d$, define $J(t,{{\bf x}}) = \big\{j: x_j<M^{\bf v}_j t \big\}$. Then ${\bf F}^{\bf v}$ defined by \eqref{cdfvec} has the following properties:
 \begin{enumerate}[label=(\roman*)]
 
	\item\label{prop_app_sub_smaller} For each $i \in {\mc S}_n$, with ${\bf v}(i)<{\bf M}^{\bf v}$, $F^{\bf v}_i(t, {\bf x})$ is Lipschitz continuous on $\RR_+ \times \RR^d$.
    
       \item\label{prop_app_sub_max} For each $i \in \mathcal{S}_n$, with ${\bf v}(i)\nless {\bf M}^{\bf v}$, $F^{\bf v}_i(t,{\bf x})=\1_{\overline{{\mc U}_E}}({\bf x}) \, \PP\big\{X(t)=i, L^{\bf v}_j(t) \leq x_j, j\in J(t,{{\bf x}})  \big\}$.

        \item\label{prop_app_sub_strong} ${\bf F}^{\bf v}$ is a strong solution of
        \begin{equation}\label{mpde}
          \del_t {\bf F}^{\bf v}(t, {\bf x}) + \sum_{j=1}^d \edit{\del_{x_j}{\bf F}^{\bf v}(t,{\bf x})V_j} = {\bf F}^{\bf v}(t,{\bf x})Q(t),
		\end{equation}
where $V_j=\diag(v_j(1)\,,\dots,v_j(n))$, in the open region ${\mc U}_I$.
Furthermore, let $(t, {\bf x}) \in \partial {\mc U}_I$. Then
\begin{equation}\label{boundaryval}
	\begin{split}
\lim_{(\bar{t},\bar{\bf x})\in {\mc U}_I \to (t,{\bf x})}F^{\bf v}_i(\bar{t}, \bar{\bf x}) & = \begin{cases}
	\PP\big\{X(t)=i, L^{\bf v}_j(t)\leq x_j, j\in J(t,{{\bf x}})  \big\}, & \text{if ${\bf v}(i)<{\bf M}^{\bf v}$},\\
    0,				& \text{otherwise},
\end{cases}\\
\end{split}
\end{equation}
and
\begin{equation}\label{eq_u_e}
	F^{\bf v}_i(\bar{t}, \bar{\bf x}) = \PP\big\{X(\bar{t})=i, L^{\bf v}_j(\bar{t})\leq\bar{x}_j, j \in J(t,{{\bf \bar{x}}}) \big\} \quad \text{for all}\;(\bar{t},\bar{\bf x})\in {\mc U}_E.\\
\end{equation}

\end{enumerate}
\end{proposition}

\begin{remark}
Computing the solution in Proposition~\ref{prop_app} for a given number $d$ of path-integrals requires computing solutions for $\bar{d}$ integrals with $\bar{d} < d$ on the boundary. These can be obtained by straightforwardly applying the proposition in lower dimensions. Note that for $d=1$, the values on the boundary can be directly obtained from the distribution of $X(t)$.
\end{remark}

\firstRevision{

\begin{remark}\label{rem_zero}
Note that for each $i \in {\mc S}_n$ we have
\begin{equation}
	F^{\bf v}_i(t,{\bf x}) = 0,
\end{equation}
for ${\bf x} \leq {\bf v}(i)t$.
\end{remark}

\begin{remark}\label{rem_generator_d}
The process $\big( X(t), {\bf L}^{\bf v}(t)\big)_{t \in \RR_+}$ is a time-inhomogeneous piecewise-deterministic strong Markov process~\cite[Chapter~2]{Davis1993}, and Proposition~\ref{prop_app} essentially shows that the generator is given by
\begin{equation}
	\mathcal{G}_t {\bf H}({\bf x}) = - \sum_{j=1}^d \edit{\del_{x_j}{\bf H}({\bf x}) V_j} + {\bf H}({\bf x})Q(t),
\end{equation}
for suitably defined functions ${\bf H}({\bf x})$. The stochastic transitions of $X(t)$ are described by $Q(t)$ and the deterministic evolution of ${\bf L}^{\bf v}(t)$ in each dimension is governed by the terms $V_j\,\del_{x_j}$.
However, in addition, Proposition~\ref{prop_app} establishes the regularity of ${\bf F}^{\bf v}(t, {\bf x})$, which is important for numerical computations.
\end{remark}

\begin{remark}\label{rem_extension}
Then the ancestral process with limited recombination satisfies assumptions (A1)-(A2), and thus, we focus on this case here. It is conceivable that these assumptions could be relaxed and Proposition~~\ref{prop_app} could be extended to more general Markov chains $X(t)$ with a (countably) infinite state space, and more general dynamics, for example, a non-triangular rate matrix $Q(t)$, or $\int_0^{\infty} q_i(s) \,ds < \infty$. However, the approach presented here in the proof of Proposition~\ref{prop_app} to show the necessary regularity of ${\bf F}^{\bf v}(t, {\bf x})$ uses the fact that $X(t)$ has absorbing states, and reaches them in finite time, after a finite number of jumps. For a more general version, this strategy would need to be adapted, or a different strategy used.
\end{remark}
}


  \begin{proof}[Proof of Proposition~\ref{prop_app}]

Let $\Delta$ denote the set of absorbing states of the process $X(t)$. Since $Q$ is lower triangular, $1 \in \Delta$ and thus  $\Delta$ is not empty.

Take any $i \in {\mc S}_n $ with  ${\bf v}(i) \nless {\bf M}^{\bf v}$. Since ${\bf v}$ is monotone along the process we conclude that
\begin{equation*}
L_j(t,\omega)=\int_0^t v_j\big(X(s,\omega)\big) \, ds =M_j^{\bf v} t \quad \text{for all} \quad j \notin J\big(1,{\bf v}(i)\big), \, \omega \in \{\tilde{\omega}: \SP X(t,\tilde{\omega})=i\}
\end{equation*}
and this yields~\ref{prop_app_sub_max}.

Recall next that for time-inhomogeneous Markov processes $X(t)$ (under the  assumptions (A1)-(A2)) the jumping times $T_1,T_2,T_3,\dots$ of $X(t)$ satisfy $\PP \big\{ T_1 > \alpha  \big\} =  \exp\big( - \int_0^{\alpha} q_1(s) \, ds \big)$
and for $k \geq 2$
\begin{equation}\label{jumpdistr}
\PP \Big\{ T_{k} > t+\alpha \big| \, T_{k-1}=t, X(T_{k-1})=i \Big\} =  \exp\Big( - \int_t^{t+\alpha} q_i(s) \, ds \Big)\,.
\end{equation}

Take any $i \in {\mc S}_n$ with ${\bf v}(i) < {\bf M}^{\bf v}$, in which case $i<n$.  Since $Q(t)$ is lower triangular, each trajectory of the process has at most $n-1$ jumps before it enters into the absorbing set $\Delta$. Thus we obtain  
\begin{equation*}
\Big\{\omega: X(\cdot,\omega) \;\; \text{enters the state $i$} \Big\} = \bigcup_{k=1}^{n-1} \Omega_{k}^{(i)}\,, \;\;\; \Omega_k^{(i)} = \Big\{\omega: \, \text{$X(\cdot,\omega)$ enters the state $i$ on the $k$-th jump}\Big\}\,.
\end{equation*}

We next denote $T_0=0$, $s_0=n$, $s_i^{(k)}=(s_1,s_2,\dots,s_{k-1},s_k=i) \in ({\mc S}_n)^k$, with $k \geq 1$, and 
\begin{equation*}
\begin{aligned}
A\big(s_i^{(k)}\big)  = \Big\{\omega: X(T_1)=s_1,\dots,X(T_{k-1})=s_{k-1},X(T_k)=s_k=i \Big\} \subset \Omega_k^{(i)}\,.
\end{aligned}
\end{equation*}

First, suppose that $i \notin \Delta$. For $(t,{\bf x}) \in \RR_+  \times \RR^d$, using the above partitioning, we write
\begin{equation}\label{Fipart}
\begin{aligned}
F_i^{\bf v}(t, {\bf x})&=\PP\{ X(t)=i, {\bf L}(t) \leq {\bf x} \} \\
& = \sum_{k=1}^{n-1} \sum_{s_i^{(k)}\in {\mc S}^k} \PP\left\{A\big(s_i^{(k)}\big),\SP T_{k}<t<T_{k+1}, \SP {\bf L}(t) \leq {\bf x} \right\}\\
& = \sum_{k=1}^{n-1} \sum_{s_i^{(k)}\in {\mc S}^k} \PP\Big\{A\big(s_i^{(k)}\big), T_{k}<t<T_{k+1}, \sum_{j=1}^k T_j\big({\bf v}(s_{j-1})-{\bf v}(s_j)\big) \leq {\bf x} - {\bf v}(i)t \Big\}.\,\\
\end{aligned}
\end{equation}

We now show that $F^{\bf v}_i$ is Lipschitz continuous. To this end, consider the function
\begin{equation}\label{temp1}
\begin{aligned}
G\big(t, {\bf x}; s_i^{(k)} \big)= \PP\Big\{A\big(s_i^{(k)}\big), T_{k}<t<T_{k+1}, \sum_{j=1}^k T_j\big({\bf v}(s_{j-1})-{\bf v}(s_j)\big)  \leq {\bf x} \Big\}.\,\\
\end{aligned}
\end{equation}
Observe that $G\big(t,{\bf x}; s_i^{(k)}\big)$ is well-defined for $(t,{\bf x}) \in \RR^{1+d}$. Moreover, since $i \notin \Delta$, the assumption (A2) implies that the process after entering the state $i$ leaves this state in finite time $\PP$-almost surely. Thus $\Omega_{k}^{(i)} \subset \{T_k<\infty\} \subset\{T_{k+1}<\infty\}$ and therefore 
\begin{equation}\label{temp2}
\begin{aligned}
G\big(t,{\bf x}; s_i^{(k)} \big)&= \PP\Big\{A\big(s^{(k)}\big), T_{k}<t, \sum_{j=1}^k T_j\big({\bf v}(s_{j-1})-{\bf v}(s_j)\big)  \leq {\bf x} \Big\}\,\\
& \qquad - \PP\Big\{A\big(s^{(k)}\big), T_{k+1}<t, \sum_{j=1}^k T_j\big({\bf v}(s_{j-1})-{\bf v}(s_j)\big)  \leq {\bf x}  \Big\}\\
&=: G_1\big(t,{\bf x}; s_i^{(k)}\big)-G_2\big(t,{\bf x}; s_i^{(k)}\big)\,.
\end{aligned}
\end{equation}



Now, using \eqref{jumpdistr} and  induction, one can show that for each $r \in {\mc S}_n$ and $k \geq 1$ 
\begin{equation}\label{temp3}
\PP\Big\{ A\big(s_r^{(k)}\big), T_{k+1} \leq z \Big\} = \int_{-\infty}^z f_{k+1}\big(\alpha \SP; s_r^{(k)}\big) \, d\alpha
\end{equation}
where $f_{k+1}\big( \cdot \, ; s_r^{(k)}\big)$ is a globally bounded function.
Thus, we conclude that the map
\begin{equation}
z \to \PP\Big\{ A\big(s_r^{(k)}\big), T_{k+1} \leq z \Big\}
\end{equation}
is globally Lipschitz for each $k \geq 1$.

Since ${\bf v}$ is non-increasing along the process, for each $s_i^{(k)}$ we  have 
\[
{\bf v}(n)={\bf M}^{\bf v} \geq {\bf v}(s_1) \dots \geq {\bf v}(s_k)={\bf v}(i)\,.
\]
By assumption ${\bf v}(i) < {\bf M}^{\bf v}$ and hence for each $l \in \{1,\ldots,d \}$ there exists $k_l \in \{1,\dots,k\}$ such that $v_l(s_{k_l-1})-v_l(s_{k_l})>0$, which guarantees that not all terms in the nonnegative sum $\sum_{j=1}^k T_j\big(v_l(s_{j-1})-v_l(s_j)\big)$ vanish.
Then, in view of the fact that the event $A\big(s_i^{(k)}\big)$ does not depend on the $(t, {\bf x})$-variable,
we can use~\eqref{temp3} and induction to conclude that
\begin{equation}\label{intLip}
\PP\Big\{ A\big(s_i^{(k)}\big), \sum_{j=1}^k T_j\big(v_l(s_{j-1})- v_l(s_j)\big)  \leq x_l \Big\} = \int_{-\infty}^{x_l} \tilde{f}_{kl}(s) \, ds\,, \quad k \geq 1\,, l \in \{1,\ldots,d\}
\end{equation} 
for some globally bounded function  $\tilde{f}_{kl}$. It can be shown that this implies that
\begin{equation}\label{intMap}
{\bf x} \to \PP\Big\{ A\big(s_i^{(k)}\big), \sum_{j=1}^k T_j\big({\bf v}(s_{j-1})- {\bf v}(s_j)\big)  \leq {\bf x} \Big\}
\end{equation} 
is globally Lipschitz. Combining \eqref{temp3} with \eqref{intMap} and using the definition of the Lipschitz continuity we conclude that $G_1\big(t,{\bf x}; s_i^{(k)}\big)$ and $G_2(t,{\bf x}; s_i^{(k)}\big)$ are globally Lipschitz and hence $G\big(t,{\bf x}; s_i^{(k)})$ is as well. 

Furthermore, any Lipschitz continuous function composed with a linear map is also Lipschitz continuous. Thus $\bar{G}\big(t,{\bf x};s_i^{(k)}\big):=G\big(B(t,{\bf x}); s_i^{(k)}\big)$, where $B(t,{\bf x})=\big(t,{\bf x}-{\bf v}(i)t\big)$, is globally Lipschitz. In \eqref{Fipart} each of the terms in the sum is one of the functions  $\bar{G}\big(x,t;s_i^{(k)}\big)$. Hence $F_i$ which is restricted to $(t,{\bf x}) \in [0,\infty) \times \RR^d$ is globally Lipschitz on this domain. 

Lastly, if $i\in \Delta$, observe that 
\begin{equation}
\Big\{ T_k<\infty, X(T_k)=i \Big\}  \subset \Big\{T_{k+1}=\infty\Big\}
\end{equation}
and therefore
\begin{equation}\label{Fnpart}
\begin{aligned}
F_i^{\bf v}(t,{\bf x})&=\PP\{ X(t)=i, {\bf L}(t)\leq{\bf x}\} \\
& = \sum_{k=1}^{n-1} \sum_{s_i^{(k)}\in {\mc S}^k} \PP\Big\{A\big(s_i^{(k)}\big), T_{k}<t, \sum_{j=1}^k T_j\big({\bf v}(s_{j-1})-{\bf v}(s_j)) \leq {\bf x} - {\bf v}(i)t \Big\}.\,\\
\end{aligned}
\end{equation}
Using an analogous approach (to the one in the case $i \notin \Delta$) one can show that each term in the above expression is globally Lipschitz continuous. This yields~\ref{prop_app_sub_smaller}.

From~\ref{prop_app_sub_smaller} and~\ref{prop_app_sub_max} it follows that ${\bf F}^{\bf v}$ is Lipschitz continuous in the open region ${\mc U}_I$. Then, by Proposition \ref{inival} we conclude that ${\bf F}^{\bf v}$ is a strong solution of \eqref{mpde} in ${\mc U}_I$. The boundary conditions~\eqref{boundaryval} and equation~\eqref{eq_u_e} follow directly from the definition of ${\bf F}^{\bf v}$. This proves~\ref{prop_app_sub_strong}.


\end{proof}

}

\section{Numerical Schemes}

\subsection{Upstream Numerical Scheme for Single-Locus Case}
\label{sec_num_alg_marginal}

Here we present a numerical algorithm for computing solutions to the system \eqref{eq_marginal_pde}. The numerical scheme is an upstream scheme based on the method of characteristics. In particular, the numerical scheme we develop makes use of the integral representation formulas \eqref{eq_sol_marginal_ode} and \eqref{eq_marginal_rate_int}.

To define a grid in the $(\time,\treeDimA)$-space suitable for computation, choose $\treeDimA_\text{max}$, the maximum value that the CDF $\P \{ \totalTreeLength \leq \treeDimA_\text{max}\}$ should be computed for. Due to Lemma~\ref{lem_cdf}, the relation $\P \{ \totalTreeLength \leq \treeDimA\} = \eff_1 ( \time_\text{max}, \treeDimA)$ holds for all $\treeDimA \leq \treeDimA_\text{max}$, with $\time_\text{max} := \frac{\treeDimA_\text{max}}{2}$. Thus $\time_\text{max}$ is set as the maximal gridpoint for $\time$. In addition to the maximum gridpoints, choose small step sizes $\Delta \time$ and $\Delta \treeDimA$. The number of gridpoints in the $\time$ dimension is then given by $\numTimePoints := \lceil \frac{\time_\text{max}}{\Delta \time} \rceil + 1$, and the set of gridpoints is given as
\begin{equation}\label{def_time_grid}
  \timeGrid := \big\{ 0, \Delta \time, 2 \Delta \time, \ldots, (\numTimePoints-1) \Delta \time, \min (\numTimePoints \Delta \time, \time_\text{max})  \big\}.
\end{equation}
For each point $\timeGrid_i$, define a grid in the $\treeDimA$-dimension as
\begin{equation}\label{def_tree_grid}
  \treeAGrid_{i} := \big\{ 0, \Delta \treeDimA, \ldots, \min (U \Delta \treeDimA, \sampleSize \timeGrid_i) \big\} \cup \big\{ 2 \Delta \time + \bar{\treeAGrid}_{i-1}, 3 \Delta \time + \bar{\treeAGrid}_{i-1}, \ldots, \sampleSize \Delta \time + \bar{\treeAGrid}_{i-1} \big\} \cup \big\{ 2 \timeGrid_i, 3 \timeGrid_i, \ldots, \sampleSize \timeGrid_i\big\},
\end{equation}
with $U = \lceil \frac{\sampleSize \timeGrid_i}{\Delta \treeDimA} \rceil$ and $\bar{\treeAGrid}_{i-1} := \max (\treeAGrid_{i-1})$. Furthermore, set $\numTreeAPoints_i := |\treeAGrid_i|$. The same grid will be used for all $\numLineages \in \{1,\ldots,\sampleSize\}$. The points $\numLineages \Delta \time + \bar{\treeAGrid}_{i-1}$ and $\numLineages \timeGrid_i$ are added for numerical stability reasons, to improve the accuracy of the interpolation we will perform in the subsequent steps.

Now fix $i \in \{0,\ldots,\numTimePoints\}$ and $\numLineages \in \{1,\ldots,\sampleSize\}$, and assume that $\eff_\ell ( \timeGrid_{i-1}, \treeAGrid_{i-1,j})$ has been computed for all $\ell \in \{1,\ldots,\sampleSize\}$ and $\treeAGrid_{i-1,j} \in \treeAGrid_{i-1}$. Furthermore, assume that $\eff_\ell ( \timeGrid_{i}, \treeAGrid_{i,j})$ has been computed for all $\ell \in \{\numLineages+1,\ldots,\sampleSize\}$ and $\treeAGrid_{i,j} \in \treeAGrid_{i}$. Under these assumptions, $\eff_\numLineages ( \timeGrid_{i}, \treeAGrid_{i,j})$ can be computed for all $\treeAGrid_{i,j} \in \treeAGrid_{i}$ as follows. If $\treeAGrid_{i,j} < \stateFunction (\numLineages)\timeGrid_{i}$, then
\begin{equation}\label{eq_sol_marginal_ode_zero}
  \eff_\numLineages ( \timeGrid_{i}, \treeAGrid_{i,j}) = 0.
\end{equation}
If $\treeAGrid_{i,j} =\sampleSize \timeGrid_{i}$, the maximal value of $\treeAGrid_{i}$, then
\begin{equation}\label{eq_sol_marginal_ode_boundary}
  \eff_\numLineages ( \timeGrid_{i}, \treeAGrid_{i,j}) = \P \Big\{ \ancestralProcess\big(\timeGrid_{i}\big) = \numLineages \Big\}.
\end{equation}

The values on the right-hand side can be pre-computed for all $\numLineages$ and $\timeGrid_{i} \in \timeGrid$ by solving the ODE~\eqref{eq_anc_proc_ode} numerically. In the general case, note that the characteristic of $\eff_\numLineages$ that goes through the point $(\timeGrid_{i}, \treeAGrid_{i,j})^{\T}$ and the boundary $\treeDimA = \sampleSize \time$ intersect at the point $(\intersectTimeA, \sampleSize \intersectTimeA)^{\T}$, with $\intersectTimeA := \frac{\treeAGrid_{i,j} - \stateFunction (\numLineages) \timeGrid_{i}}{\sampleSize - \stateFunction (\numLineages)}$. Thus, define
\begin{equation}\label{eq_trace_points}
  (\treeAGrid^\downarrow_{i,j}, \timeGrid^\downarrow_{i,j})^{\T} := \begin{cases}
      (\timeGrid_{i-1}, \treeAGrid_{i,j} - \stateFunction (\numLineages) \Delta \time)^{\T},  & \text{if $\intersectTimeA < \timeGrid_{i-1}$},\\
      (\intersectTimeA, \sampleSize \intersectTimeA)^{\T},                  & \text{otherwise},
    \end{cases} 
\end{equation}
the projection of $(\timeGrid_{i}, \treeAGrid_{i,j})^{\T}$ back along the corresponding characteristic to the previous time-slice $\timeGrid_{i-1}$, or onto the boundary $\treeDimA = \sampleSize \time$, whichever has the larger $\time$-component. This backward projection step is illustrated in Figure~\ref{fig_nalg_marg}(\subref{fig_nalg_marg_a}). Then, according to equation~\eqref{eq_sol_marginal_ode}
\begin{equation}\label{eq_sol_marginal_ode_step}
  \eff_\numLineages ( \timeGrid_{i}, \treeAGrid_{i,j}) = e^{-(\rateInt^{(1)}_\numLineages(\timeGrid_{i})- \rateInt^{(1)}_\numLineages(\timeGrid^\downarrow_{i,j}))} \Bigg( \int_{\timeGrid^\downarrow_{i,j}}^{\timeGrid_{i}} \inhODE^{(1)}_\numLineages (\alpha) e^{(\rateInt^{(1)}_\numLineages(\alpha)- \rateInt^{(1)}_\numLineages(\timeGrid^\downarrow_{i,j}))} d\alpha+ \eff_\numLineages ( \treeAGrid^\downarrow_{i,j}, \timeGrid^\downarrow_{i,j}) \Bigg)
\end{equation}
holds.

\begin{figure}
\begin{subfigure}[t]{0.45\textwidth}
\begin{center}
  \includegraphics[width=\textwidth]{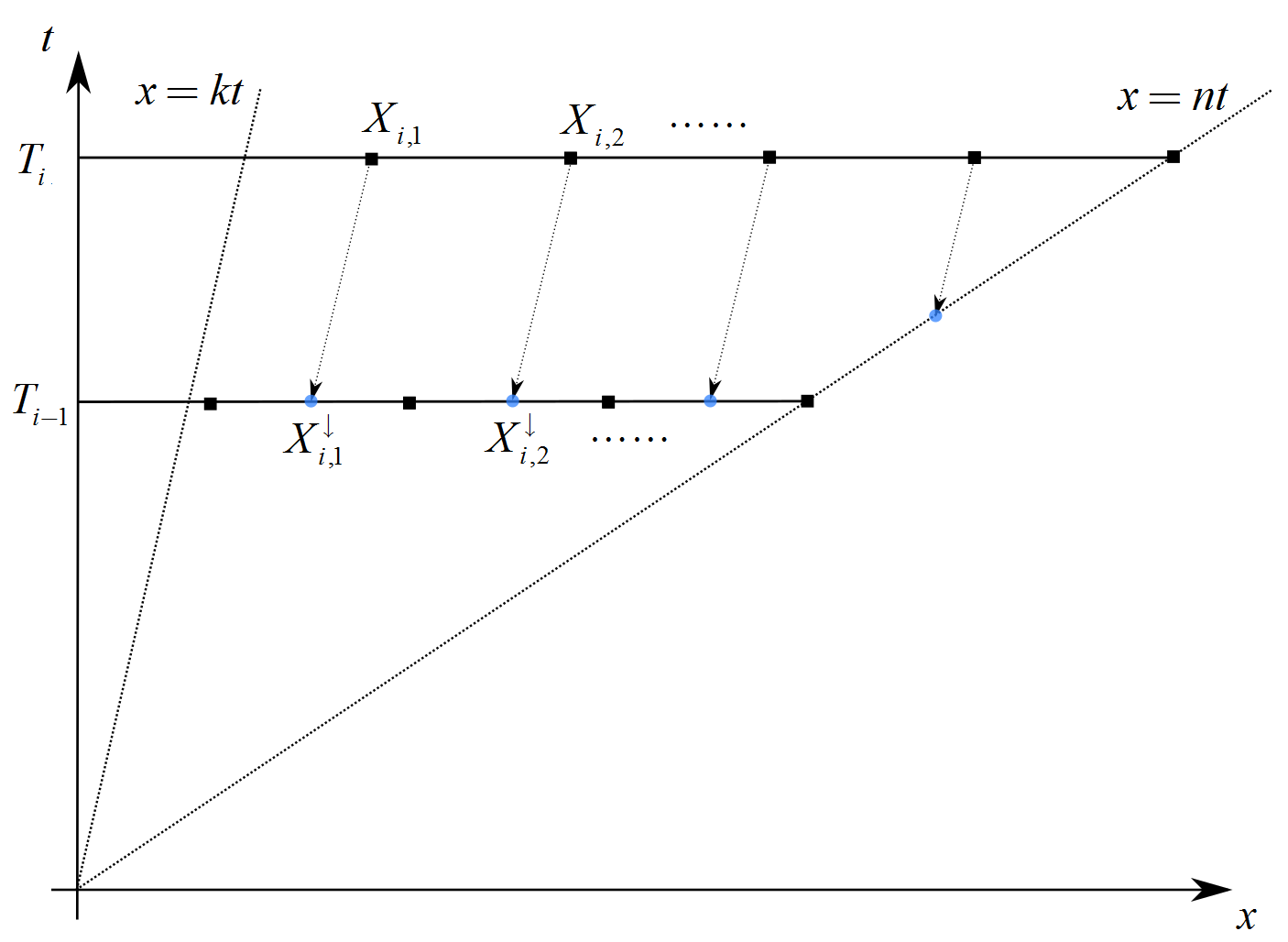}
  \caption{Projection the grid points backwards along the characteristics from time layer $\timeGrid_{i}$ to $\timeGrid_{i-1}$.}
  \label{fig_nalg_marg_a}
\end{center}
\end{subfigure}
\hspace{30pt} 
\begin{subfigure}[t]{0.45\textwidth}
\begin{center}
  \includegraphics[width=\textwidth]{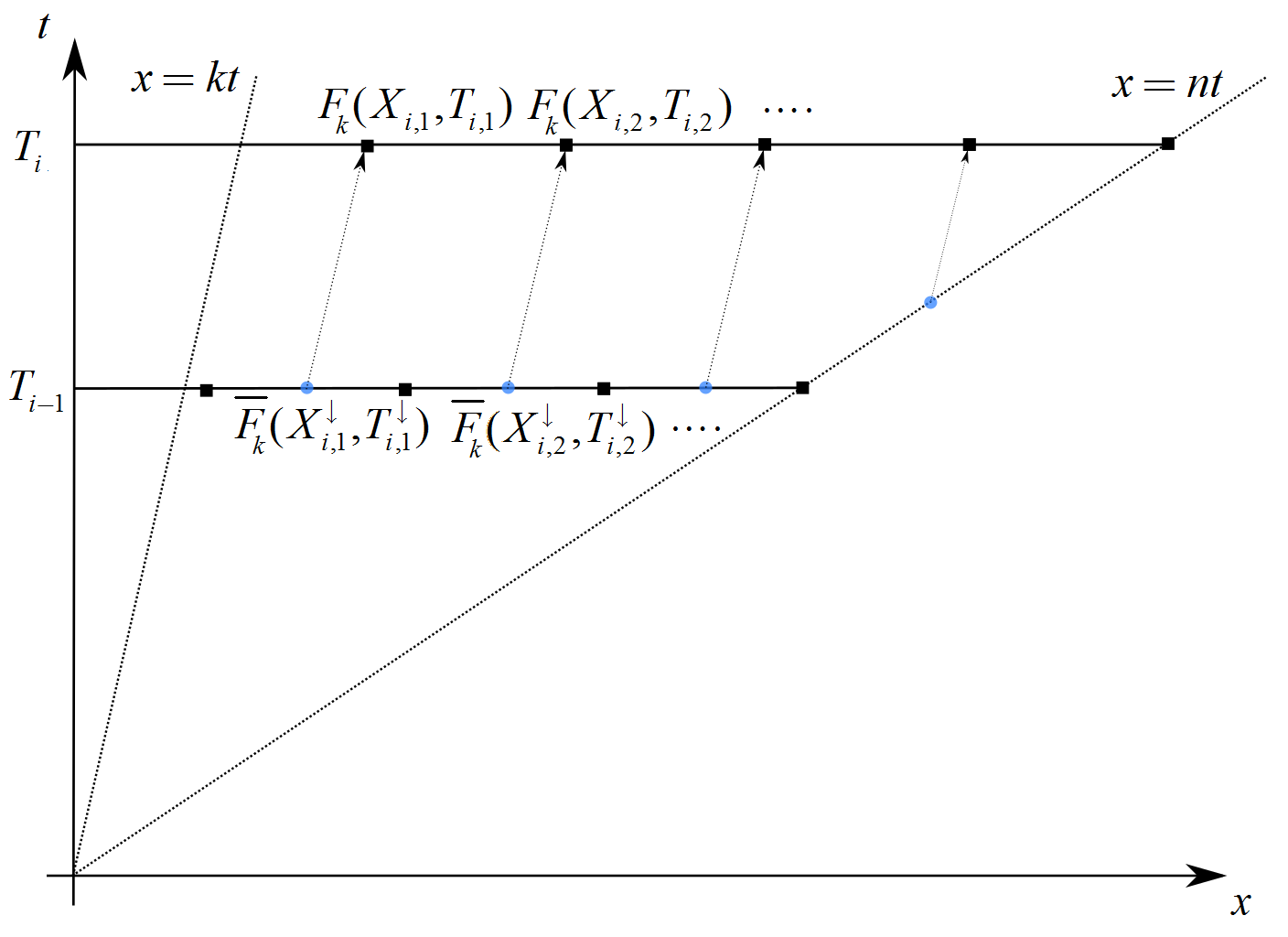}
  \caption{Propagating the interpolated values of the function $\eff_\numLineages$ via numerical integration.}
  \label{fig_nalg_marg_b}
\end{center}
\end{subfigure}
\caption{The back-tracing and propagation step of the upstream numerical scheme to compute $\eff_\numLineages$ at all points of the grid.}\label{fig_nalg_marg}
\end{figure}

The right-hand side of the equation~\eqref{eq_sol_marginal_ode_step} can now be computed using two approximations. Note that the point $( \timeGrid^\downarrow_{i,j}, \treeAGrid^\downarrow_{i,j})^{\T}$ is in general not on the grid $\treeAGrid_{i}$, and thus $\eff_\numLineages ( \timeGrid^\downarrow_{i,j}, \treeAGrid^\downarrow_{i,j})$ has not been pre-computed. If $\intersectTimeA  < \timeGrid_{i-1}$, the point is equal to $(\timeGrid_{i-1}, \treeAGrid_{i,j} - \stateFunction (\numLineages) \Delta \time)^{\T}$. In that case, identify the two grid points in $\treeAGrid_{i-1}$ that are closest to $\treeAGrid_{i,j} - \stateFunction (\numLineages)\Delta \time$ to the right and \st{to the} left. Then let $\bar{\eff}_\numLineages (  \timeGrid^\downarrow_{i,j},  \treeAGrid^\downarrow_{i,j})$ be the linear interpolation between the values of $\eff_\numLineages$ at those gridpoints. If $\intersectTimeA  \geq \timeGrid_{i-1}$, then the point is given by $(\timeGrid^\downarrow_{i,j}, \treeAGrid^\downarrow_{i,j})^{\T} = (\intersectTimeA , \sampleSize \intersectTimeA )^{\T}$ and is located on the boundary. Thus
\begin{equation}
  \eff_\numLineages ( \treeAGrid^\downarrow_{i,j}, \timeGrid^\downarrow_{i,j}) = \P \Big\{ \ancestralProcess\big(\intersectTimeA\big) = \numLineages \Big\},
\end{equation}
which is also not pre-computed. However, $\P \Big\{ \ancestralProcess\big(\timeGrid_{i-1}\big) = \numLineages \Big\}$ and $\P \Big\{ \ancestralProcess\big(\timeGrid_{i}\big) = \numLineages \Big\}$ have been pre-computed, and thus set $\bar{\eff}_\numLineages ( \timeGrid^\downarrow_{i,j}, \treeAGrid^\downarrow_{i,j})$ as the linear interpolation between these two values.

The second approximation is to compute the integral on the right-hand side of equation~\eqref{eq_sol_marginal_ode_step} using the trapezoidal rule. Thus, the values of $\eff_\numLineages$ on the grid can be computed using
\begin{equation}\label{eq_sol_marginal_ode_step_approx}
\begin{aligned}
  \eff_\numLineages ( \timeGrid_{i}, \treeAGrid_{i,j}) = & \frac{\Delta \time}{2} \Bigg( \inhODE^{(1)}_\numLineages (\timeGrid_{i}) + e^{-(\rateInt^{(1)}_\numLineages(\timeGrid_{i})- \rateInt^{(1)}_\numLineages(\timeGrid^\downarrow_{i,j}))} \inhODE^{(1)}_\numLineages (\timeGrid^\downarrow_{i,j}) \Bigg) + e^{-(\rateInt^{(1)}_\numLineages(\timeGrid_{i})- \rateInt^{(1)}_\numLineages(\timeGrid^\downarrow_{i,j}))} \bar{\eff}_\numLineages ( \timeGrid^\downarrow_{i,j}, \treeAGrid^\downarrow_{i,j})\\
& \qquad + o(\Delta\time^3)+o({\Delta \treeDimA}^2)       
\end{aligned}
\end{equation}
The terms $\inhODE_\numLineages (\cdot)$ depend on the values of $\eff_\ell$ with $\numLineages < \ell \leq \sampleSize$ that might not have been pre-computed on the grid either. However, the same interpolation schemes as for $\bar{\eff}_\numLineages$ can be applied. At this stage it is important though to strictly set $\eff_\ell$ to 0 if it should be 0 according to equation~\eqref{eq_eff_regions}. Lastly, the values for $\rateInt^{(1)}_\numLineages(\cdot)$ can either be obtained using analytic formulas in equation~\eqref{eq_marginal_rate_int} for certain classes of coalescent-speed functions we will consider (e.g. piece-wise constant), or by computing the requisite integrals using the trapezoidal rule, which can be done incrementally. The integration step of out numerical scheme is illustrated in Figure~\ref{fig_nalg_marg}(\subref{fig_nalg_marg_b}).

These equations lead naturally to a dynamic programming algorithm to compute $\eff_\numLineages$ on the specified grid. To this end, iterate through the values $\timeGrid_{i} \in \timeGrid$ in increasing order. For each $\timeGrid_{i}$, iterate through $\numLineages \in \{1,2,\ldots,\sampleSize\}$ in decreasing order, starting with $\numLineages = \sampleSize$. Then, for each fixed $\timeGrid_{i}$ and $\numLineages$, $\eff_\numLineages ( \timeGrid_{i}, \treeAGrid_{i,j})$ can be computed for every $\treeAGrid_{i,j} \in \treeAGrid_{i}$ using equations~\eqref{eq_sol_marginal_ode_zero}, \eqref{eq_sol_marginal_ode_boundary}, and~\eqref{eq_sol_marginal_ode_step_approx}. The order of iteration guarantees that all necessary quantities have been pre-computed. This dynamic program can be employed to compute $\eff_\numLineages$ on the specified grid for all $\numLineages$. Due to Lemma~\ref{lem_cdf}, the relation
\begin{equation}
  \P \{ \totalTreeLength \leq \treeAGrid_{\numTimePoints,j}\} = \eff_1 (\timeGrid_{M}, \treeAGrid_{\numTimePoints,j})
\end{equation}
holds, which yields the values of the CDF $\P \{ \totalTreeLength \leq \treeDimA\}$ on the specified grid $\treeAGrid_{\numTimePoints}$.

\subsection{Upstream Numerical Scheme for Two-Locus Case}
\label{sec_algo_joint}

In the two-locus case, we can compute $\eff_\state$ efficiently on a chosen grid similar to the marginal case. To this end, we again choose $\treeDimA_\text{max} = \treeDimB_\text{max}$, set $\time_\text{max} := \frac{1}{n} \treeDimA_\text{max}$, and choose step sizes $\Delta \time$ and $\Delta \treeDimA = \Delta \treeDimB$. Then, define the grid $\timeGrid$ as in definition~\eqref{def_time_grid}, $\numTimePoints = |\timeGrid|$ and for each $\timeGrid_{i}$, define $\treeAGrid_{i}$ as in definition~\eqref{def_tree_grid}. Furthermore, set $\treeBGrid_{i} := \treeAGrid_{i}$ and $\numTreeAPoints_{i} := |\treeBGrid_{i}|$. Thus, we use the regular grid $\treeAGrid_{i} \times \treeBGrid_{i}$ in the $(\treeDimA, \treeDimB)$-space.

Now fix $\timeGrid_{i}$ and $\state \in \recoLimitedStates$, and assume that $\eff_{\state'} (\timeGrid_{i-1}, \treeAGrid_{i-1,j}, \treeBGrid_{i-1,\ell})$ has been computed for all $\state' \in \recoLimitedStates$, $\treeAGrid_{i-1,j} \in \treeAGrid_{i-1}$, and $\treeBGrid_{i-1,\ell} \in \treeBGrid_{i-1}$. Furthermore, assume that $\eff_{\state'} (\timeGrid_{i}, \treeAGrid_{i,j}, \treeBGrid_{i,\ell})$ has been computed for all $\state'$ with $\state \prec \state'$, $\treeAGrid_{i,j} \in \treeAGrid_{i}$, and $\treeBGrid_{i,\ell} \in \treeBGrid_{i}$. To compute $\eff_\state ( \timeGrid_{i}, \treeAGrid_{i,j}, \treeBGrid_{i,\ell} )$, first check using equation~\eqref{eq_joint_boundaries} whether the requisite point lies on the boundary, or in the zero region. The values on the boundary according to equation~\eqref{eq_joint_boundaries} are computed as time-dependent CDFs of marginal integrals along the trajectories of the process $\ancestralRecoLimitedProcess$, and thus they can be computed using exactly the same procedure as detailed in Section~\ref{sec_num_alg_marginal}, replacing $\ancestralProcess$ by $\ancestralRecoLimitedProcess$. In the interior region, applying the trapezoidal rule to the solution of the first-order ODE, for all $\treeAGrid_{i,j} \in \treeAGrid_{i}$, and $\treeBGrid_{i,\ell} \in \treeBGrid_{i}$ the value of $\eff_{\state} (\timeGrid_{i}, \treeAGrid_{i,j}, \treeBGrid_{i,\ell})$ can be computed using
\begin{equation}
  \begin{split}
    \eff_\state ( \timeGrid_{i},& \treeAGrid_{i,j}, \treeBGrid_{i,\ell} )\\
      = &\frac{\Delta \time}{2} \Bigg( \inhODE^{(2)}_\state (\timeGrid_{i}) + e^{-(\rateInt^{(2)}_\state(\timeGrid_{i})- \rateInt^{(2)}_\state(\timeGrid^\downarrow_{i,j,\ell}))} \inhODE^{(2)}_\state (\timeGrid^\downarrow_{i,j,\ell}) \Bigg) + e^{-(\rateInt^{(2)}_\state(\timeGrid_{i})- \rateInt^{(2)}_\state(\timeGrid^\downarrow_{i,j,\ell}))} \bar{\eff}_\state ( \timeGrid^\downarrow_{i,j,\ell}, \treeAGrid^\downarrow_{i,j}, \treeBGrid^\downarrow_{i,\ell} )\\
      & \qquad + o(\Delta\time^3)+o({\Delta \treeDimA}^2) + o({\Delta \treeDimB}^2)\,.
  \end{split}
\end{equation}
Here
\begin{equation}
  (\timeGrid^\downarrow_{i,j,\ell}, \treeAGrid^\downarrow_{i,j}, \treeBGrid^\downarrow_{i,\ell})^{\T} := \begin{cases}
      (\timeGrid_{i-1}, \treeAGrid_{i,j} - \stateFunction^{\aLocus} (\state) \Delta \time, \treeBGrid_{i,\ell} - \stateFunction^{\bLocus} (\state) \Delta \time)^{\T},  & \text{if $\max (\intersectTimeA, \intersectTimeB) < \timeGrid_{i-1}$},\\
      \big(\intersectTimeA, \sampleSize \intersectTimeA, \treeBGrid_{i,\ell} - \stateFunction^{\bLocus} (\state) \cdot (\timeGrid_{i} - \intersectTimeA) \big)^{\T},                  & \text{if $\max (\timeGrid_{i-1}, \intersectTimeB) \leq \intersectTimeA$},\\
      \big(\intersectTimeB, \treeAGrid_{i,j} - \stateFunction^{\aLocus} (\state) \cdot (\timeGrid_{i} - \intersectTimeB), \sampleSize \intersectTimeB\big)^{\T},                  & \text{if $\max (\intersectTimeA, \timeGrid_{i-1}) \leq \intersectTimeB$},
    \end{cases} 
\end{equation}
with
\begin{equation}
  \intersectTimeA := \frac{\treeAGrid_{i,j} - \stateFunction^{\aLocus} (\state) \timeGrid_{i}}{\sampleSize - \stateFunction^{\aLocus} (\state)}
\end{equation}
being the $\time$-coordinate of the point of intersection between the characteristic through the point $(\timeGrid_{i}, \treeAGrid_{i,j}, \treeBGrid_{i,\ell})^{\T}$ and the boundary $\treeDimA = \sampleSize \time$, and
\begin{equation}
  \intersectTimeB := \frac{\treeBGrid_{i,\ell} - \stateFunction^{\bLocus} (\state) \timeGrid_{i}}{\sampleSize - \stateFunction^{\aLocus} (\state)}
\end{equation}
likewise for the boundary $\treeDimB = \sampleSize \time$.

The points $(\timeGrid^\downarrow_{i,j,\ell}, \treeAGrid^\downarrow_{i,j}, \treeBGrid^\downarrow_{i,\ell})^{\T}$ will in general not be on the grid of pre-computed values, and thus the approximation $\bar{\eff}_\state ( \timeGrid^\downarrow_{i,j,\ell}, \treeAGrid^\downarrow_{i,j}, \treeBGrid^\downarrow_{i,\ell} )$ has to be used. In the case $\max (\intersectTimeA, \intersectTimeB) < \timeGrid_{i-1}$, this value can be obtained by identifying the four points in $\treeAGrid_{i-1} \times \treeBGrid_{i-1}$ surrounding $(\treeAGrid^\downarrow_{i,j}, \treeBGrid^\downarrow_{i,\ell})$, and interpolating the respective values of $\eff_\state ( \timeGrid_{i-1}, \cdot, \cdot )$ linearly. In the case $\max (\timeGrid_{i-1}, \intersectTimeB) \leq \intersectTimeA$, the point $(\timeGrid^\downarrow_{i,j,\ell}, \treeAGrid^\downarrow_{i,j}, \treeBGrid^\downarrow_{i,\ell})^{\T}$ is on the boundary $\treeDimA = \sampleSize \time$, and
\begin{equation}
  \eff_\state (\timeGrid^\downarrow_{i,j,\ell}, \treeAGrid^\downarrow_{i,j}, \treeBGrid^\downarrow_{i,\ell}) = \P \big\{ \ancestralRecoLimitedProcess(\intersectTimeA) = \state,  \treeLength^\bLocus(\intersectTimeA) \leq \treeBGrid_{i,\ell} - \stateFunction^{\bLocus} (\state) \cdot (\timeGrid_{i} - \intersectTimeA) \big\}
\end{equation}
holds. The value of the time-dependent CDF on the right-hand can be obtained as the linear interpolation between the values $\P \big\{ \ancestralRecoLimitedProcess(\timeGrid_{i-1}) = \state,  \treeLength^\bLocus(\timeGrid_{i-1}) \leq \treeBGrid_{i,\ell} - \stateFunction^{\bLocus} (\state) \Delta \time \big\}$ and $\P \big\{ \ancestralRecoLimitedProcess(\timeGrid_{i}) = \state,  \treeLength^\bLocus(\timeGrid_{i}) \leq \treeBGrid_{i,\ell} \big\}$, which we pre-compute (or approximations thereof) using the numerical scheme for the marginal case (see Appendix~\ref{sec_num_alg_marginal}) on the boundary. By symmetry, the case $\max (\intersectTimeA, \timeGrid_{i-1}) \leq \intersectTimeB$ can be handled in the same way. Computing $\inhODE^{(2)}_\state(\cdot)$ will require some $\eff_{\state'}$ with $\state \prec \state'$, which can be obtained by similar interpolation procedures, or setting it to zero in the appropriate regions. The values of $\rateInt^{(2)}_\state(\cdot)$ can be computed according to equation~\eqref{eq_rate_int_joint} analytically or numerically, as before.

Again, we can implement these formulas in an efficient dynamic programming algorithm to compute the values of $\eff_\state ( \time, \treeDimA, \treeDimB )$ on the specified grid for all $\state \in \recoLimitedStates$, and thus compute
\begin{equation}
  \P \{ \totalTreeLength^\aLocus \leq \treeAGrid_{\numTimePoints,j} , \totalTreeLength^\bLocus \leq \treeBGrid_{\numTimePoints,\ell} \} = \eff_{(1,0,0,0)} ( \time_\text{max}, \treeAGrid_{\numTimePoints,j}, \treeBGrid_{\numTimePoints,\ell} ) + \eff_{(1,0,0,1)} ( \time_\text{max}, \treeAGrid_{\numTimePoints,j}, \treeBGrid_{\numTimePoints,\ell} ),
\end{equation}
the joint CDF of the total tree length at two linked loci evaluated on the specified grid.

\end{document}